\documentclass[english,11pt,letter]{article}
\usepackage[T1]{fontenc}
\usepackage[utf8]{luainputenc}
\usepackage{geometry}
\usepackage{amsthm}
\usepackage{amsmath}
\usepackage{amssymb}
\usepackage{esint}

\makeatletter

\providecommand{\tabularnewline}{\\}

\numberwithin{equation}{section}
\numberwithin{figure}{section}

\usepackage{amsmath, amsthm, amssymb, amsfonts, bbold}
\usepackage{bbm}
\usepackage[basic]{complexity}
\usepackage[pdftex,bookmarks,colorlinks]{hyperref}
\usepackage{enumitem}
\usepackage{color}
\usepackage{fullpage}

\newtheorem{theorem}{Theorem}
\newtheorem{lemma}[theorem]{Lemma}

\newtheorem{definition}[theorem]{Definition}

\DeclareMathOperator*{\argmaxTex}{arg\,max}
\DeclareMathOperator*{\argminTex}{arg\,min}
\DeclareMathOperator*{\signTex}{sign}
\DeclareMathOperator*{\rankTex}{rank}
\DeclareMathOperator*{\diagTex}{diag}
\DeclareMathOperator*{\imTex}{im}

\hypersetup{colorlinks=true,citecolor=red}

\renewcommand{\varepsilon}{\epsilon}

\makeatother

\usepackage{babel}
\begin{document}

\global\long\def\R{\mathbb{R}}
 \global\long\def\Rn{\mathbb{R}^{n}}
 \global\long\def\Rm{\mathbb{R}^{m}}
 \global\long\def\Rmn{\mathbb{R}^{m \times n}}
 \global\long\def\Rnm{\mathbb{R}^{n \times m}}
 \global\long\def\Rmm{\mathbb{R}^{m \times m}}
 \global\long\def\Rnn{\mathbb{R}^{n \times n}}
 \global\long\def\Z{\mathbb{Z}}
 \global\long\def\rPos{\R_{> 0}}
 \global\long\def\rNonNeg{\R_{\geq0}}


\global\long\def\ellOne{\ell_{1}}
 \global\long\def\ellTwo{\ell_{2}}
 \global\long\def\ellInf{\ell_{\infty}}
 \global\long\def\ellP{\ell_{p}}

\global\long\def\otilde{\widetilde{O}}

\global\long\def\argmax{\argmaxTex}

\global\long\def\argmin{\argminTex}

\global\long\def\sign{\signTex}

\global\long\def\rank{\rankTex}

\global\long\def\diag{\diagTex}

\global\long\def\im{\imTex}

\global\long\def\enspace{\quad}

\global\long\def\boldVar#1{\mathbf{#1}}

\global\long\def\mvar#1{\boldVar{#1}}

\global\long\def\vvar#1{\vec{#1}}



\global\long\def\defeq{\stackrel{\mathrm{{\scriptscriptstyle def}}}{=}}

\global\long\def\diag{\mathrm{{diag}}}

\global\long\def\mDiag{\mvar{diag}}
 \global\long\def\ceil#1{\left\lceil #1 \right\rceil }

\global\long\def\E{\mathbb{E}}
 \global\long\def\abs#1{\left|#1\right|}

\global\long\def\gradient{\nabla}
\global\long\def\grad{\gradient}
 \global\long\def\hessian{\gradient^{2}}
 \global\long\def\hess{\hessian}
 \global\long\def\jacobian{\mvar J}
 \global\long\def\gradIvec#1{\vvar{f_{#1}}}
 \global\long\def\gradIval#1{f_{#1}}


\global\long\def\onesVec{\vec{\mathbb{1}}}
 \global\long\def\setVec#1{\onesVec_{#1}}
 \global\long\def\indicVec#1{\onesVec_{#1}}
\global\long\def\indic{1}


\global\long\def\pseudo#1{{#1}^{\dagger}}

\global\long\def\ith#1#2{{#1}^{(#2)}}

\global\long\def\specGeq{\succeq}
 \global\long\def\specLeq{\preceq}
 \global\long\def\specGt{\succ}
 \global\long\def\specLt{\prec}

\global\long\def\dualVec#1{{#1}^{\#}}
 \global\long\def\dualVecOne#1{{#1}^{\#_{1}}}
 \global\long\def\dualVecInf#1{{#1}^{\#_{\infty}}}
 \global\long\def\dualVecP#1{{#1}^{\#_{p}}}
 \global\long\def\dualVecGen#1#2{{#1}^{\#_{#2}}}
 \global\long\def\dualVecKFull#1{\left({#1}\right)^{\#_{(k)}}}
 \global\long\def\sharpv#1{\dualVec{#1}}
 \global\long\def\sharpgrad#1{\dualVec{\gradient#1}}

\global\long\def\va{\vvar a}
 \global\long\def\vb{\vvar b}
 \global\long\def\vc{\vvar c}
 \global\long\def\vd{\vvar d}
 \global\long\def\ve{\vvar e}
 \global\long\def\vf{\vvar f}
 \global\long\def\vg{\vvar g}
 \global\long\def\vh{\vvar h}
 \global\long\def\vl{\vvar l}
 \global\long\def\vm{\vvar m}
 \global\long\def\vn{\vvar n}
 \global\long\def\vo{\vvar o}
 \global\long\def\vp{\vvar p}
 \global\long\def\vs{\vvar s}
 \global\long\def\vu{\vvar u}
 \global\long\def\vv{\vvar v}
 \global\long\def\vw{\vvar w}
 \global\long\def\vx{\vvar x}
 \global\long\def\vy{\vvar y}
 \global\long\def\vz{\vvar z}
 \global\long\def\vxi{\vvar{\xi}}
 \global\long\def\valpha{\vvar{\alpha}}
 \global\long\def\veta{\vvar{\eta}}
 \global\long\def\vlambda{\vvar{\lambda}}
 \global\long\def\vDelta{\vvar{\Delta}}
\global\long\def\vphi{\vvar{\phi}}
\global\long\def\vpsi{\vvar{\psi}}
 \global\long\def\vsigma{\vvar{\sigma}}
 \global\long\def\vgamma{\vvar{\gamma}}
 \global\long\def\vzero{\vvar 0}
 \global\long\def\vones{\vvar 1}
\global\long\def\vq{\vvar q}
\global\long\def\vr{\vvar r}

\global\long\def\ma{\mvar A}
 \global\long\def\mb{\mvar B}
 \global\long\def\mc{\mvar C}
 \global\long\def\md{\mvar D}
 \global\long\def\mf{\mvar F}
 \global\long\def\mg{\mvar G}
 \global\long\def\mh{\mvar H}
 \global\long\def\mj{\mvar J}
 \global\long\def\mk{\mvar K}
 \global\long\def\mm{\mvar M}
 \global\long\def\mn{\mvar N}
 \global\long\def\mq{\mvar Q}
 \global\long\def\mr{\mvar R}
 \global\long\def\ms{\mvar S}
 \global\long\def\mt{\mvar T}
 \global\long\def\mU{\mvar U}
 \global\long\def\mv{\mvar V}
 \global\long\def\mw{\mvar W}
 \global\long\def\mx{\mvar X}
 \global\long\def\my{\mvar Y}
 \global\long\def\mz{\mvar Z}
 \global\long\def\mSigma{\mvar{\Sigma}}
 \global\long\def\mLambda{\mvar{\Lambda}}
\global\long\def\mPhi{\mvar{\Phi}}
 \global\long\def\mZero{\mvar 0}
 \global\long\def\iMatrix{\mvar I}
\global\long\def\mDelta{\mvar{\Delta}}

\global\long\def\paramLengths{\vvar l}
 \global\long\def\paramCapacity{\vvar{\mu}}
 \global\long\def\paramWeights{\vvar w}

\global\long\def\oracle{\mathcal{O}}
 \global\long\def\moracle{\mvar O}
 \global\long\def\oracleOf#1{\oracle\left(#1\right)}

\global\long\def\timeNearlyOp{\tilde{\mathcal{O}}}
 \global\long\def\timeNearlyLinear{\timeNearlyOp}
 \global\long\def\runtime{\mathcal{T}}
 \global\long\def\timeOf#1{\runtime\left(#1\right)}

\global\long\def\dInterior{S^{0}}
 \global\long\def\dSlack{\rPos^{m}}
 \global\long\def\dWeights{\rPos^{m}}
 \global\long\def\dFull{\{\dInterior\times\dWeights\}}

\global\long\def\shurProd{\circ}
 \global\long\def\shurSquared#1{{#1}^{(2)}}

\global\long\def\weight{w}
 \global\long\def\vWeight{\vvar{\weight}}
 \global\long\def\mWeight{\mvar W}

\global\long\def\slack{s}
 \global\long\def\vSlack{\vs}
 \global\long\def\vSlackX{\vs(\vx)}
 \global\long\def\mSlack{\ms}
 \global\long\def\mSlackX{\mSlack_{x}}
 \global\long\def\slackXi{s(\vx)_{i}}

\global\long\def\mProj{\mvar P}
 \global\long\def\mProjStandard{\mProj_{\ma_{\vx}}(\vw)}

\global\long\def\mResist{\mr}
 \global\long\def\vLever{\vsigma}
 \global\long\def\mLever{\mSigma}
 \global\long\def\mLapProj{\mvar{\Lambda}}

\global\long\def\boundary{\partial}
 \global\long\def\boundaryX{\vx_{\boundary}}
 \global\long\def\boundaryW{\vw_{\boundary}}

\global\long\def\penalizedObjective{f_{t}}
 \global\long\def\penalizedObjectiveWeight{\hat{f}}

\global\long\def\fvWeight{\vg}
 \global\long\def\fvWeightCurr{\fvWeight(\vSlackCurr)}
 \global\long\def\fvWeightNext{\fvWeight(\vSlackNext)}
 \global\long\def\fmWeight{\mg}
 \global\long\def\fmWeightCurr{\fmWeight(\vxCurr)}
 \global\long\def\fmWeightNext{\fmWeight(\vxNext)}
 \global\long\def\fvWeightSimple{\vg_{s}}
 \global\long\def\fvWeightFull{\vg}
 \global\long\def\fmWeightFull{\mg}

\global\long\def\gradW{\grad_{\vWeight}}
 \global\long\def\gradX{\grad_{\vx}}
 \global\long\def\hessWW{\hess_{\vWeight\vWeight}}
 \global\long\def\hessWX{\hess_{\vWeight\vx}}
 \global\long\def\hessXW{\hess_{\vx\vWeight}}
 \global\long\def\hessXX{\hess_{\vx\vx}}

\global\long\def\vNewtonStep{\vh}

\global\long\def\stepSizeX{\alpha_{\vx}}
 \global\long\def\stepSizeW{\beta_{\vw}}
 \global\long\def\stepRatio{r}

\global\long\def\oldText{old}
 \global\long\def\newText{new}
 \global\long\def\finText{apx}
 \global\long\def\curr#1{\ith{#1}{\oldText}}
 \global\long\def\next#1{\ith{#1}{\newText}}
 \global\long\def\old#1{\ith{#1}{\oldText}}
 \global\long\def\new#1{\ith{#1}{\newText}}
 \global\long\def\vaOld{\curr{\va}}
 \global\long\def\vaNew{\new{\va}}
 \global\long\def\vbOld{\curr{\vb}}
 \global\long\def\vbNew{\new{\vb}}
 \global\long\def\maOld{\ma_{\newText}}
 \global\long\def\maNew{\ma_{\oldText}}
 \global\long\def\mbOld{\mb_{\oldText}}
 \global\long\def\mbNew{\mb_{\newText}}
 \global\long\def\vxCurr{\curr{\vx}}
 \global\long\def\vxNext{\next{\vx}}
 \global\long\def\mxCurr{\mx_{(\oldText)}}
 \global\long\def\mxNext{\mx_{(\oldText)}}
 \global\long\def\weightCurr{\curr{\weight}}
 \global\long\def\weightNext{\next{\weight}}
 \global\long\def\vWeightCurr{\curr{\vWeight}}
 \global\long\def\vWeightNext{\next{\vWeight}}
 \global\long\def\vWeightFin{\ith{\vWeight}{\finText}}
 \global\long\def\slackCurr{\curr{\slack}}
 \global\long\def\slackNext{\next{\slack}}
 \global\long\def\vSlackCurr{\curr{\vSlack}}
 \global\long\def\vSlackNext{\next{\vSlack}}
 \global\long\def\mSlackCurr{\mSlack_{(old)}}
 \global\long\def\mSlackNext{\mSlack_{(new)}}
 \global\long\def\mWeightCurr{\mWeight_{(\oldText)}}
 \global\long\def\mWeightNext{\mWeight_{(\newText)}}
 \global\long\def\vGradCurr{\curr{\vg}}
 \global\long\def\vGradNext{\next{\vg}}
 \global\long\def\mHessCurr{\mh_{(\oldText)}}
 \global\long\def\mHessNext{\mh_{(\newText)}}

\global\long\def\innerProduct#1#2{\left\langle #1 , #2 \right\rangle }
 \global\long\def\norm#1{\big\|#1\big\|}
 \global\long\def\normFull#1{\left\Vert #1\right\Vert }
 \global\long\def\normA#1{\norm{#1}_{\ma}}
 \global\long\def\normFullInf#1{\normFull{#1}_{\infty}}
 \global\long\def\normInf#1{\norm{#1}_{\infty}}
 \global\long\def\normOne#1{\norm{#1}_{1}}
 \global\long\def\normTwo#1{\norm{#1}_{2}}
 \global\long\def\normCoord#1#2{\norm{#1}_{(#2)}}
 \global\long\def\normDual#1{{\norm{#1}^{*}}}
 \global\long\def\normSpecial#1{\norm{#1}_{\infty, \mg}}
 \global\long\def\normFullSpecial#1{\normFull{#1}_{\infty, \mg}}
 \global\long\def\normLeverage#1{\norm{#1}_{\mSigma}}
 \global\long\def\normWeight#1{\norm{#1}_{\fmWeight}}

\global\long\def\normCurrFun#1{\norm{#1}_{\fmWeight(\vSlackCurr)+\infty}}
 \global\long\def\normCurrWeight#1{\norm{#1}_{\mWeightCurr+\infty}}
 \global\long\def\normNextFun#1{\norm{#1}_{\fmWeight(\vSlackNext)+ \infty}}
 \global\long\def\normNextWeight#1{\norm{#1}_{\vWeightNext+\infty}}

\global\long\def\cWeightSize{c_{1}}
 \global\long\def\cWeightStab{c_{\energyStab}}
 \global\long\def\cWeightCons{c_{r}}

\global\long\def\proj{\mathrm{proj}}

\global\long\def\vWeightError{\vvar{\Psi}}
\global\long\def\proximity{\delta}

\global\long\def\energyW{E_{w}}
 \global\long\def\energyWCurr{{\energyW^{\oldText}}}
 \global\long\def\energyWNext{{\energyW^{\newText}}}
 \global\long\def\energyStab{\gamma}

\global\long\def\fBarrier#1{\phi(#1)}
\global\long\def\fBarrierStandard#1{\phi_{l}(#1)}
\global\long\def\fBarrierVolumetric#1{\phi_{v}(#1)}

\global\long\def\TODO#1{{\color{red}TODO:\text{#1}}}

\global\long\def\runningTimeBest{\otilde((\mathrm{nnz}(\ma)+\left(\rank(\ma)\right)^{\omega})\sqrt{\rank(\ma)}L)}

\global\long\def\runningTimePrev{O(m^{1.5}nL)}

\global\long\def\nnz{\mathrm{nnz}}
\global\long\def\code#1{\texttt{#1}}
\global\long\def\updateStep{\mathrm{\code{step}}}
\global\long\def\centralityCurr{\curr{\delta_{t}}}
 \global\long\def\centralityNext{\next{\delta_{t}}}
\global\long\def\centeringExact{\mathrm{\code{centeringExact}}}
\global\long\def\centeringInexact{\mathrm{\code{centeringInexact}}}

\title{Path Finding I :\\
Solving Linear Programs with $\otilde(\sqrt{rank})$ Linear System
Solves }

\author{Yin Tat Lee\\
MIT\\
yintat@mit.edu\and  Aaron Sidford\\
MIT\\
sidford@mit.edu}

\date{}
\maketitle
\begin{abstract}
In this paper we present a new algorithm for solving linear programs
that requires only $\otilde(\sqrt{\rank(\ma)}L)$ iterations to solve
a linear program with $m$ constraints, $n$ variables, and constraint
matrix $\ma$, and bit complexity $L$. Each iteration of our method
consists of solving $\otilde(1)$ linear systems and additional nearly
linear time computation. 

Our method improves upon the previous best iteration bound by factor
of $\tilde{\Omega}((m/\rank(\ma))^{1/4})$ for methods with polynomial
time computable iterations and by $\tilde{\Omega}((m/\rank(\ma))^{1/2})$
for methods which solve at most $\otilde(1)$ linear systems in each
iteration. Our method is parallelizable and amenable to linear algebraic
techniques for accelerating the linear system solver. As such, up
to polylogarithmic factors we either match or improve upon the best
previous running times for solving linear programs in both depth and
work for different ratios of $m$ and $\rank(\ma)$. 

Moreover, our method matches up to polylogarithmic factors a theoretical
limit established by Nesterov and Nemirovski in 1994 regarding the
use of a ``universal barrier'' for interior point methods, thereby
resolving a long-standing open question regarding the running time
of polynomial time interior point methods for linear programming.
\end{abstract}

\section{Introduction}

Given a matrix, $\ma\in\R^{m\times n}$, and vectors, $\vb\in\Rm$
and $\vc\in\Rn$, solving the linear program%
\footnote{This expression is the dual of a linear program written in standard
form. It is well known that all linear programs can be written as
\eqref{eq:intro:problem}. Note that this notation of $m$ and $n$
differs from that in some papers. Here $m$ denotes the number of
constraints and $n$ denotes the number of variables. To avoid confusion
we state many of our results in terms of $\sqrt{\rank(\ma)}$ instead
of $\sqrt{n}$ . %
}
\begin{equation}
\min_{\vx\in\Rn~:~\ma\vx\geq\vb}\vc^{T}\vx\label{eq:intro:problem}
\end{equation}
is a core algorithmic task for both the theory and practice of computer
science. 

Since Karmarkar's breakthrough result in 1984, proving that interior
point methods can solve linear programs in polynomial time for a relatively
small polynomial, interior point methods have been an incredibly active
area of research with over 1200 papers written just as of 1994 \cite{Nesterov1994}.
Currently, the fastest asymptotic running times for solving \eqref{eq:intro:problem}
in many regimes are interior point methods. Previously, state of the
art interior point methods for solving \eqref{eq:intro:problem} require
either $O(\sqrt{m}L)$ iterations of solving linear systems \cite{renegar1988polynomial}
or $O((m\rank(\ma))^{1/4}L)$ iterations of a more complicated but
still polynomial time operation \cite{vaidya89convexSet,vaidya90parallel,vaidya1993technique,anstreicher96}.%
\footnote{Here and in the rest of the paper $L$ denotes the standard ``bit
complexity'' of the linear program. The parameter $L$ is at most
the number of bits needed to represent \eqref{eq:intro:problem}.
For integral $\ma$, $\vb$, and $\vc$ the quantity $L$ is often
defined to be the potentially smaller quantity $L=\log(m)+\log(1+d_{max})+\log(1+\max\{\normInf{\vc},\normInf{\vb}\})$
where $d_{max}$ is the largest absolute value of the determinant
of a square sub-matrix of $\ma$ \cite{karmarkar1984new}.%
}

However, in a breakthrough result of Nesterov and Nemirovski in 1994,
they showed that there exists a \emph{universal barrier} function
that if computable would allow \eqref{eq:intro:problem} to be solved
in $O(\sqrt{\rank(\ma)}L)$ iterations \cite{nesterov1997self}. Unfortunately,
this barrier is more difficult to compute than the solutions to \eqref{eq:intro:problem}
and despite this existential result, the $O((m\rank(\ma))^{1/4}L)$
iteration bound for polynomial time linear programming methods has
not been improved in over 20 years.

In this paper we present a new interior point method that solves general
linear programs in $\otilde(\sqrt{\rank(\ma)}L)$ iterations thereby
matching the theoretical limit proved by Nesterov and Nemirovski up
to polylogarithmic factors.%
\footnote{Here and in the remainder of the paper we use $\otilde(\cdot)$ to
hide $\polylog(n,m)$ factors.%
} Furthermore, we show how to achieve this convergence rate while only
solving $\otilde(1)$ linear systems and performing additional $\otilde(\nnz(\ma))$
work in each iteration.%
\footnote{We assume that $\ma$ has no rows or columns that are all zero as
these can be remedied by trivially removing constraints or variables
respectively or immediately solving the linear program. Therefore
$\nnz(\ma)\geq\min\{m,n\}$.%
} Our algorithm is parallelizable and we achieve the first $\otilde(\sqrt{\rank(\ma)}L)$
depth polynomial work method for solving linear programs. Furthermore,
using one of the regression algorithms in \cite{nelson2012osnap,li2012iterative},
our linear programming algorithm has a running time of $\runningTimeBest$
where $\omega<2.3729$ is the matrix multiplication constant \cite{williams2012matrixmult}.
This is the first polynomial time algorithm for linear programming
to achieve a nearly linear dependence on $\nnz(\ma)$ for fixed $n$.
Furthermore, we show how to use acceleration techniques as in \cite{vaidya1989speeding}
to decrease the amortized per-iteration costs of solving the requisite
linear system and thereby achieve a linear programming algorithm with
running time faster than the previous fastest running time of $\runningTimePrev$
whenever $m=\tilde{\Omega}\left(n\right)$. This is the first provable
improvement on both running time and the number of iterations for
general interior point methods in over 20 years.

We achieve our results through an extension of standard path following
techniques for linear programming \cite{renegar1988polynomial,gonzaga1992path}
that we call \emph{weighted path finding}. We study what we call the
\emph{weighted central path}, an idea of adding weights to the standard
logarithmic barrier function \cite{todd1994scaling,freund_weighted,megiddo_weighted}
that was recently used implicitly by Mądry to make an important breakthrough
improvement on the running time for solving unit-capacity instances
of the maximum flow problem \cite{madryFlow}. We provide a general
analysis of properties of the weighted central path, discuss tools
for manipulating points along the path, and ultimately produce an
efficiently computable path that converges in $\otilde(\sqrt{\rank(\ma)}L)$
steps. We hope that these results may be of independent interest and
serve as tools for further improving the running time of interior
point methods in general. While the analysis in this paper is quite
technical, our linear programming method is straightforward and we
hope that these techniques may prove useful in practice.

\subsection{Previous Work}

Linear programming is an extremely well studied problem with a long
history. There are numerous algorithmic frameworks for solving linear
programming problems, e.g. simplex methods \cite{dantzig1951maximization},
ellipsoid methods \cite{khachiyan1980polynomial}, and interior point
methods \cite{karmarkar1984new}. Each method has a rich history and
an impressive body of work analyzing the practical and theoretical
guarantees of the methods. We couldn't possibly cover the long line
of beautiful work on this important problem in full, and we make no
attempt. Instead, here we present the major improvements on the number
of iterations required to solve \eqref{eq:intro:problem} and discuss
the asymptotic running times of these methods. For a more comprehensive
history of polynomial time algorithms for linear programming and interior
point we refer the reader to one of the many excellent references
on the subject, e.g. \cite{Nesterov1994,ye2011interior}.

In 1984 Karmarkar \cite{karmarkar1984new} provided the first proof
of an interior point method running in polynomial time. This method
required $O(mL)$ iterations where the running time of each iteration
was dominated by the time needed to solve a linear system of the form
$\left(\ma^{T}\md\ma\right)\vx=\mbox{\ensuremath{\vy}}$ for some
positive diagonal matrix $\md\in\Rmm$ and some $\vy\in\Rn$. Using
low rank matrix updates and preconditioning Karmarkar achieved a running
time of $O(m^{3.5}L)$ for solving \eqref{eq:intro:problem} inspiring
a long line of research into interior point methods.%
\footnote{Here and in the remainder of the paper when we provide asymptotic
running times for linear programming algorithms, for simplicity we
hide additional dependencies on $L$ that may arise from the need
to carry out arithmetic operations to precision $L$. %
}

Karmarkar's result sparked interest in a particular type of interior
point methods, known as \emph{path following methods}. These methods
solve \eqref{eq:intro:problem} by minimizing a penalized objective
function $f_{t}(\vx)$,
\[
\min_{\vx\in\Rn}f_{t}(\vx)\enspace\text{ where }\enspace f_{t}(\vx)\defeq t\cdot\vc^{T}\vx+\phi(\vx)
\]
where $\phi:\Rn\rightarrow\R$ is a \emph{barrier function }such that
$\phi(\vx)\rightarrow\infty$ as $\vx$ tends to boundary of the polytope
and $t$ is a parameter. Usually, the standard\emph{ log barrier}
$\phi(\vx)\defeq-\sum_{i\in[m]}\log([\ma\vx-\vb]_{i})$ is used. Path
following methods first approximately minimize $f_{t}$ for small
$t$, then use this minimizer as an initial point to minimize $f_{\left(1+c\right)t}$
for some constant $c$, and then repeat until the minimizer is close
to the optimal solution of \eqref{eq:intro:problem}. 

Using this approach Renegar provided the first polynomial time interior
point method which solves \eqref{eq:intro:problem} in $O\left(\sqrt{m}L\right)$
iterations \cite{renegar1988polynomial}. As with Karmarkar's result
the running time of each iteration of this method was dominated by
the time needed to solve a linear system of the form $\left(\ma^{T}\md\ma\right)\vx=\vy$.
Using a combination of techniques involving low rank updates, preconditioning
and fast matrix multiplication, the amortized complexity of each iteration
was improved \cite{vaidya1987speeding,gonzaga1992path,Nesterov1994}.
The previously fastest running time achieved by such techniques was
$\runningTimePrev$ \cite{vaidya1989speeding}.

In a seminal work of Nesterov and Nemirovski \cite{Nesterov1994},
they showed that path-following methods can in principle be applied
to minimize any linear cost function over any convex set by using
a suitable barrier function. Using this technique they showed how
various problems such as semidefinite programming, finding extremal
ellipsoids, and more can all be solved in polynomial time via path
following. In this general setting, the number of iterations required
depended on the square root of a quantity associated with the barrier
called \emph{self-concordance}. They showed that for any convex set
in $\Rn$, there exists a barrier function, called the \emph{universal
barrier} function, with self-concordance $O(n)$. Therefore, in theory
any such convex optimization problems with $n$ variables can be solved
in $O\left(\sqrt{n}L\right)$ iterations. However, this result is
generally considered to be only of theoretical interest as the universal
barrier function is defined as the volume of certain polytopes, a
problem which in full generality is NP-hard and its derivatives can
only approximated by solving $O(n^{c})$ linear programs for some
large constant $c$ \cite{lovaszV06}.

Providing a barrier that enjoys a fast convergence rate and is easy
minimize approximately is an important theoretical question with numerous
implications. Renegar's path-following method effectively reduces
solving a linear program to solving $O(\sqrt{m}L)$ linear systems.
Exploiting the structure of these systems yields the fastest known
algorithms for combinatorial problems such as minimum cost flow \cite{daitch2008faster}
and multicommodity flow \cite{vaidya1989speeding}. Given recent breakthroughs
in solving two broad class of linear systems, symmetric diagonally
dominant linear systems \cite{spielman2004nearly,KMP11,Kelner2013,lee2013ACDM}
and overdetermined system of linear equations \cite{clarkson2013low,nelson2012osnap,li2012iterative}
improving the convergence rate of barrier methods while maintaining
easy to compute iterations could have far reaching implications%
\footnote{Indeed, in Part II \cite{lsMaxflow}we show how ideas in this paper
can be used to yield the first general improvement to the running
time of solving the maximum flow problem on capacitated directed graphs
since 1998 \cite{GoldbergRao}.%
}

In 1989, Vaidya \cite{vaidya1993technique} made an important breakthrough
in this direction. He proposed two barrier functions related to the
volume of certain ellipsoids which were shown to yield $O(\left(m\rank(\ma)\right)^{1/4}L)$
and $O(\rank(\ma)L)$ iteration linear programming algorithms \cite{vaidya90parallel,vaidya1993technique,vaidya89convexSet}.
Unfortunately each iteration of these methods required explicit computation
of the projection matrix $\md^{1/2}\ma(\ma^{T}\md\ma)^{-1}\ma^{T}\md^{1/2}$
for a positive diagonal matrix $\md\in\Rmm$. This was slightly improved
by Anstreicher \cite{anstreicher96} who showed it sufficed to compute
the diagonal of this projection matrix. Unfortunately both these methods
do not yield faster running times than \cite{vaidya1989speeding}
unless $m\gg n$ and neither are immediately amenable to take full
advantage of improvements in solving structured linear system solvers.

\begin{center}
\begin{tabular}{|c|l|c|c|}
\hline 
Year  & Author  & Number of Iterations & Nature of iterations\tabularnewline
\hline 
\hline 
1984  & Karmarkar \cite{karmarkar1984new}  & $O(mL)$  & Linear system solve\tabularnewline
\hline 
1986  & Renegar \cite{renegar1988polynomial}  & $O(\sqrt{m}L)$  & Linear system solve\tabularnewline
\hline 
1989  & Vaidya \cite{vaidya1996new}  & $O(\left(m\rank(\ma)\right)^{1/4}L)$  & Expensive linear algebra\tabularnewline
\hline 
1994  & Nesterov and Nemirovskii \cite{Nesterov1994} & $O(\sqrt{\rank(\ma)}L)$  & Volume computation\tabularnewline
\hline 
2013  & This paper  & $\tilde{O}(\sqrt{\rank(\ma)}L)$  & $\otilde(1)$ Linear system solves\tabularnewline
\hline 
\end{tabular}
\par\end{center}

These results seem to suggest that you can solve linear programs closer
to the $\tilde{O}(\sqrt{\rank(\ma)}L)$ bound achieved by the universal
barrier only if you pay more in each iteration. In this paper we show
that this is not the case. Up to polylogarithmic factors we achieve
the convergence rate of the universal barrier function while only
having iterations of cost comparable to that of Karmarkar's and Renegar's
algorithms.

\subsection{Our Approach}

\label{sec:our_approach}In this paper our central goal is to produce
an algorithm to solve \eqref{eq:intro:problem} in $\otilde(\sqrt{\rank(\ma)}L)$
iterations where each iteration solves $\otilde(1)$ linear systems
of the form $\left(\ma^{T}\md\ma\right)\vx=\vy$. To achieve our goal
ideally we would produce a barrier function $\phi$ such that standard
path following yields a $\otilde(\sqrt{\rank(\ma)}L)$ iteration algorithm
with low iterations costs. Unfortunately, we are unaware of a barrier
function that both yields a fast convergence rate and has a gradient
that can be computed with high accuracy using $\otilde(1)$ linear
system solves. Instead, we consider manipulating a barrier that we
can easily compute the gradient of, the standard logarithmic barrier,
$\phi(\vx)=-\sum_{i\in[m]}\log[\ma\vx-\vb]_{i}$.

Note that the behavior of the logarithmic barrier is highly dependent
on the representation of \eqref{eq:intro:problem}. Just duplicating
a constraint, i.e. a row of $\ma$ and the corresponding entry in
$\vb$, corresponds to doubling the contribution of some log barrier
term $-\log[\ma\vx-\vb]_{i}$ to $\phi$. It is not hard to see that
repeating a constraint many times can actually slow down the convergence
of standard path following methods. In other words, there is no intrinsic
reason to weight all the $-\log[\ma\vx-\vb]_{i}$ the same and the
running time of path following methods do depend on the weighting
of the $-\log[\ma\vx-\vb]_{i}$. Recently, Mut and Terklaky proved
that by duplicating constraints on Klee-Minty cubes carefully, the
standard logarithmic barrier really requires $O(\sqrt{m}\log(1/\varepsilon))$
iterations \cite{mut2013tight}.

To alleviate this issue, we add weights to the log barrier that we
change during the course of the algorithm. We show that by carefully
manipulating these weights we can achieve a convergence rate that
depends on the dimension of the polytope, $\rank(\ma)$, rather than
the number of constrains $m$. In Section~\ref{sec:weighted_path},
we study this \emph{weighted log barrier function }given by
\[
\phi(\vx)=-\sum_{i\in[m]}g_{i}(\ma\vx-\vb)\cdot\log([\ma\vx-\vb]_{i})
\]
where $\vg:\rPos^{m}\rightarrow\rPos^{m}$ is a \emph{weight function}
of the current point and we investigate what properties of $\vg(\vx)$
yield a faster convergence rate. 

To illustrate the properties of the weighted logarithmic barrier,
suppose for simplicity that we normalize $\ma$ and $\vb$ so that
$\ma\vx-\vb=\onesVec$ and let $\vg\defeq\vg(\vones)$. Under these
assumptions, we show that the rate of convergence of path following
depends on $\norm{\vg}_{1}$ and
\begin{equation}
\max_{i\in[m]}\indicVec i^{T}\ma\left(\ma^{T}\mDiag\left(\vg\right)\ma\right)^{-1}\ma^{T}\indicVec i.\label{eq:introduction:1}
\end{equation}
To improve the convergence rate we would like to keep both these quantities
small. For a general matrix $\ma$, the quantity \eqref{eq:introduction:1}
is related to the leverage scores of the rows of $\ma$, a commonly
used measure for the importance of rows in a linear system \cite{mahoney11survey}. 

For illustration purposes, if we assume that $\ma$ is the incidence
matrix of a certain graph and put a resistor of resistance $1/g_{i}$
on the edge $i$. Then, $\indicVec i^{T}\ma\left(\ma^{T}\mDiag\left(\vg\right)\ma\right)^{-1}\ma^{T}\indicVec i$
is the effective resistance of the edge $i$ \cite{spielmanS08sparsRes}.
Hence, we wish to to find $g$ to minimize the maximum effective resistance
of the graph while keeping $\norm{\vg}_{1}$ small. Thus, if it exists,
an optimal $\vg$ would simply make all effective resistances the
same.

This \emph{electric network inverse problem} is well studied \cite{strang1991inverse}
and motivates us to considering the following weight function
\begin{equation}
\vg(\vs)\defeq\argmax_{\vWeight\in\Rm}-\onesVec^{T}\vWeight+\frac{1}{\alpha}\log\det\left(\ma^{T}\ms^{-1}\mWeight^{\alpha}\ms^{-1}\ma\right)+\beta\sum_{i\in[m]}\log w_{i}.\label{eq:intro:saddle_point_barrier}
\end{equation}
for carefully chosen constants $\alpha,\beta$ where $\ms\defeq\mDiag(\vSlack(\vx))$
and $\mw=\mDiag(\vWeight)$. The optimality conditions of this optimization
problem imply that the effective resistances are small, the total
weight is small, no weight is too small, and every term in the logarithmic
barrier is sufficiently penalized. This barrier is related to the
volumetric barrier function used by Vaidya \cite{vaidya1996new} and
can be viewed as searching for the best function in a family of volumetric
barrier function. This formulation with some careful analysis can
be made to yield an $\tilde{O}(\sqrt{n}L)$ iteration path-following
algorithm by solving the following minimax problem
\begin{equation}
\min_{\vx\in\Rn}\max_{\vWeight\in\Rm}t\vc^{T}\vx-\onesVec^{T}\vWeight+\frac{1}{\alpha}\log\det\left(\ma^{T}\ms^{-1}\mWeight^{\alpha}\ms^{-1}\ma\right)+\beta\sum_{i\in[m]}\log w_{i}\label{eq:minimax_problem}
\end{equation}
where $\vs(\vx)\defeq\ma\vx-\vb$, $\ms\defeq\mDiag(\vSlack(\vx))$
and $\mw=\mDiag(\vWeight)$.

Unfortunately, computing the derivative of the minimax formula still
requires computing the diagonal of the projection matrix as in Vaidya
and Anstreicher's work \cite{vaidya1989speeding,anstreicher96} and
is therefore too inefficient for our purposes. In Section~\ref{sec:approx_path}
we show how to compute $\vWeight$ approximately up to certain multiplicative
coordinate-wise error using dimension reduction techniques. However,
this error is still too much for path following to handle the directly
as multiplicatively changing weights can hurt our measures of centrality
too much.

Therefore, rather than using the weighted log barrier
\[
\phi(\vx)=-\sum_{i\in[m]}g_{i}(\vx)\log(\slackXi)
\]
where the weights $\vg(\vx)$ depends on the $\vx$ directly, we maintain
separate weights $\vWeight$ and current point $\vx$ and use the
barrier
\[
\phi(\vx,\vWeight)=-\sum_{i\in[m]}w_{i}\log(\slackXi).
\]
We then maintain two invariants, (1) $\vx$ is centered, i.e. $\vx$
close to the minimum point of $t\cdot\vc^{T}\vx+\phi(\vx,\vw)$ and
(2) $\vw$ close to $\vg(\vx)$ multiplicatively.

We separate the problem of maintaining these invariants into two steps.
First, we design a step for changing $\vx$ and $\vw$ simultaneously
that improves centrality without moving $\vw$ too far away from $\vg(\vx)$.
We do this by decomposing a standard Newton step into a change in
$\vx$ and a change in $\vw$ with a ratio chosen using properties
of the particular weight function. Second, we show that given a multiplicative
approximation to $\vg(\vx)$ and bounds for how much $\vg(\vx)$ may
have changed, we can maintain the invariant that $\vg(\vx)$ is close
to $\vw$ multiplicatively without moving $\vw$ too much. We formulate
this as a general two player game and prove that there is an efficient
strategy to maintain our desired invariants. Combining these and standard
techniques in path-following methods, we obtain an $\tilde{O}(\sqrt{\rank(\ma)}L)$
iterations path-following algorithm where each iterations consists
of $\tilde{O}(1)$ linear system solves.

We remark that a key component of our result is a better understanding
of the effects of weighting the logarithmic barrier and note that
recently Mądry \cite{madryFlow} has shown another way of using weighted
barrier functions to achieve a $\tilde{O}(m^{10/7})$ time path-following
method for the maximum flow problem on unweighted graphs. We hope
this provides further evidence of the utility of the weighted central
path discussed in later sections.

\subsection{Geometric Interpretation of the Barrier}

While to the best of our knowledge the specific weighted barrier,
\eqref{eq:intro:saddle_point_barrier}, presented in the previous
section is new, the minimax problem, \eqref{eq:minimax_problem},
induced by the weight function is closely related to fundamental problems
in convex geometry. In particular, if we set $\alpha=1$, $t=0$,
and consider the limit as $\beta\rightarrow0$ in \eqref{eq:minimax_problem}
then we obtain the following minimax problem
\begin{equation}
\min_{\vx\in\Rn}\max_{\vWeight\geq0}-\onesVec^{T}\vWeight+\log\det\left(\ma^{T}\ms^{-1}\mWeight\ms^{-1}\ma\right).\label{eq:geom_interp}
\end{equation}
The maximization problem inside \eqref{eq:geom_interp} is often referred
to as \emph{$D$-optimal design} and is directly related to computing
the John Ellipsoid of the polytope $\left\{ \vy\in\Rn:\left|\left[\ma\left(\vy-\vx\right)\right]_{i}\right|\leq\slackXi\right\} $
\cite{khachiyan1996rounding}. In particular, \eqref{eq:geom_interp}
is directly computing the John Ellipsoid of the polytope $\left\{ \vx\in\Rn:\ma\vx\geq\vb\right\} $
and hence, one can view our linear programming algorithm as using
approximate John Ellipsoids to improve the convergence rate of interior
point methods.

Our algorithm is not the first instance of using John Ellipsoids in
convex optimization or linear programming. In a seminal work of Tarasov,
Khachiyan and Erlikh in 1988 \cite{khachiyan1988method}, they showed
that a general convex problem can be solved in $O(n)$ steps of computing
John Ellipsoid and querying a separating hyperplane oracle. Furthermore,
in 2008 Nesterov \cite{Nesterov:2008:RCS:1451525.1451532} also demonstrated
how to use a John ellipsoid to compute approximate solutions for certain
classes of linear programs in $O(\sqrt{n}/\epsilon)$ iterations and
$\tilde{O}(n^{2}m+n^{1.5}m/\epsilon)$ time. 

From this geometric perspective, there are two major contributions
of this paper. First, we show that the logarithmic volume of an approximate
John Ellipsoid is an almost optimal barrier function for linear programming
and second, that computing approximate John Ellipsoids can be streamlined
such that the cost of these operations is comparable to pert-iteration
cost of using the standard logarithmic barrier function.

\subsection{Overview}

The rest of the paper is structured as follows. In Section \ref{sec:Notation}
we provide details regarding the mathematical notation we use throughout
the paper. In Section \ref{sec:preliminaries} we provide some preliminary
information on linear programming and interior point methods. In Section
\ref{sec:weighted_path} we formally introduce the weighted path and
analyze this path assuming access to weight function. In Section \ref{sec:weights_full}
we present our weight function. In Section \ref{sec:approx_path}
we showed approximate weights suffices and in Section \ref{sec:algorithm}
we put everything together to present a $\otilde(\sqrt{\rank(\ma)}L)$
iteration algorithm for linear programming where in each iteration
we solve $\otilde(1)$ linear systems. Finally, in the Appendix we
provide some additional mathematical tools we use throughout the paper.
Note that throughout this paper we make little attempt to reduce polylogarithmic
factors in our running time.

\section{Notation\label{sec:Notation}}

Here we introduce various notation that we will use throughout the
paper. This section should be used primarily for reference as we reintroduce
notation as needed later in the paper. (For a summary of linear programming
specific notation we use, see Appendix~\ref{sec:glossary}.)

\medskip{}

\textbf{Variables:} We use the vector symbol, e.g. $\vx$, to denote
a vector and we omit the symbol when we denote the vectors entries,
e.g. $\vx=(x_{1},x_{2},\ldots)$. We use bold, e.g. $\ma$, to denote
a matrix. For integers $z\in\Z$ we use $[z]\subseteq\Z$ to denote
the set of integers from 1 to $z$. We let $\indicVec i$ denote the
vector that has value $1$ in coordinate $i$ and is $0$ elsewhere.

\medskip{}

\textbf{Vector Operations:} We frequently apply scalar operations
to vectors with the interpretation that these operations should be
applied coordinate-wise. For example, for vectors $\vx,\vy\in\Rn$
we let $\vx/\vy\in\Rn$ with $[\vx/\vy]_{i}\defeq(x_{i}/y_{i})$ and
$\log(\vx)\in\Rn$ with $[\log(\vx)]_{i}=\log(x_{i})$ for all $i\in[n]$
. 

\medskip{}

\textbf{Matrix Operations:} We call a symmetric matrix $\ma\in\Rnn$
positive semidefinite (PSD) if $\vx^{T}\ma\vx\geq0$ for all $\vx\in\Rn$
and we call $\ma$ positive definite (PD) if $\vx^{T}\ma\vx>0$ for
all $\vx\in\Rn$. For a positive definite matrix $\ma\in\Rnn$ we
denote let $\|\cdot\|_{\ma}:\R^{n}\rightarrow\R$ denote the norm
such that for all $\vx\in\Rn$ we have $\|\vx\|_{\ma}\defeq\sqrt{\vx^{T}\ma\vx}$.
For symmetric matrices $\ma,\mb\in\Rnn$ we write $\ma\specLeq\mb$
to indicate that $\mb-\ma$ is PSD (i.e. $\vx^{T}\ma\vx\leq\vx^{T}\mb\vx$
for all $\vx\in\Rn$) and we write $\ma\specLt\mb$ to indicate that
$\mb-\ma$ is PD (i.e. that $\vx^{T}\ma\vx<\vx^{T}\mb\vx$ for all
$\vx\in\Rn$). We define $\specGt$ and $\specGeq$ analogously. For
$\ma,\mb\in\R^{n\times m}$, we let $\ma\shurProd\mb$ denote the
Schur product, i.e. $[\ma\shurProd\mb]_{ij}\defeq\ma_{ij}\cdot\mb_{ij}$
for all $i\in[n]$ and $j\in[m]$, and we let $\shurSquared{\ma}\defeq\ma\shurProd\ma$.
We use $\nnz(\ma)$ to denote the number of nonzero entries in $\ma$.

\medskip{}

\textbf{Diagonal Matrices:} For $\ma\in\R^{n\times n}$ we let $\diag(\ma)\in\R^{n}$
denote the vector such that $\diag(\ma)_{i}=\ma_{ii}$ for all $i\in[n]$.
For $\vx\in\R^{n}$ we let $\mDiag(\vx)\in\R^{n\times n}$ be the
diagonal matrix such that $\diag(\mDiag(\vx))=\vx$. For $\ma\in\R^{n\times n}$
we let $\mDiag(\ma)$ be the diagonal matrix such that $\diag(\mDiag(\ma))=\diag(\ma)$.
For a vector $\vx\in\Rn$ when the meaning is clear from context we
use $\mx\in\Rnn$ to denote $\mx\defeq\mDiag(\vx)$.

\medskip{}

\textbf{Multiplicative Approximations:} Frequently in this paper we
need to convey that two vectors $\vx$ and $\vy$ are close multiplicatively.
We often write $\normInf{\mx^{-1}(\vy-\vx)}\leq\epsilon$ to convey
the equivalent facts that $y_{i}\in[(1-\epsilon)x_{i},(1+\epsilon)x_{i}]$
for all $i$ or $(1-\epsilon)\mx\specLeq\my\specLeq(1+\epsilon)\mx$.
At times we find it more convenient to write $\norm{\log\vx-\log\vy}_{\infty}\leq\epsilon$
which is approximately equivalent for small $\epsilon$. In Lemma~\ref{lem:appendix:log_helper},
we bound the quality of this approximation.

\medskip{}

\textbf{Matrices:} We use $\rPos^{m}$ to denote the vectors in $\Rm$
where each coordinate is positive and for a matrix $\ma\in\Rmn$ and
vector $\vx\in\rPos^{m}$ we define the following matrices and vectors 
\begin{itemize}
\item Projection matrix $\mProj_{\ma}(\vx)\in\Rmm$: $\mProj_{\ma}(\vx)\defeq\mx^{1/2}\ma(\ma^{T}\mx\ma)^{-1}\ma^{T}\mx^{1/2}$.
\item Leverage scores $\vLever_{\ma}(\vx)\in\Rm$: $\vLever_{\ma}(\vx)\defeq\diag(\mProj_{\ma}(\vx))$.
\item Leverage matrix $\mLever_{\ma}(\vx)\in\Rmm$: $\mLever_{\ma}(\vx)\defeq\mDiag(\mProj_{\ma}(\vx))$.
\item Projection Laplacian $\mLapProj_{\ma}(\vx)\in\Rmm$: $\mLapProj_{\ma}(\vx)\defeq\mLever_{\ma}(\vx)-\shurSquared{\mProj_{\ma}(\vx)}$. 
\end{itemize}
The definitions of projection matrix and leverage scores are standard
when the rows of $\ma$ are reweighed by the values in vector $\vx$.

\medskip{}

\textbf{Convex Sets:} We call a set $U\subseteq\R^{k}$ \emph{convex}
if for all $\vx,\vy\in\R^{k}$ and all $t\in[0,1]$ it holds that
$t\cdot\vx+(1-t)\cdot\vy\in U$. We call $U$ \emph{symmetric} if
$\vx\in\R^{k}\Leftrightarrow-\vx\in\R^{k}$. For any $\alpha>0$ and
convex set $U\subseteq\R^{k}$ we let $\alpha U\defeq\{\vx\in\R^{k}|\alpha^{-1}\vx\in U\}$.
For any $p\in[1,\infty]$ and $r\in\rNonNeg$ the \emph{$\ellP$ ball
of radius $r$} is given by $\{\vx\in\R^{k}|\norm{\vx}_{p}\leq r\}$.

\medskip{}

\textbf{Calculus:} For a function $f:\Rn\rightarrow\R$ differentiable
at $x\in\Rn$, we denote the gradient of $f$ at $\vx$ by $\grad f(\vx)\in\Rn$
where we have $[\grad f(\vx)]_{i}=\frac{\partial}{\partial x_{i}}f(\vx)$
for all $i\in[n]$. If $f\in\mathbb{R}^{n}\rightarrow\R$ is twice
differentiable at $x\in\Rn$, we denote the Hessian of $f$ at $x$
by$\hess f(\vx)$ where we have $[\grad f(\vx)]_{ij}=\frac{\partial^{2}}{\partial x_{i}\partial x_{j}}f(\vx)$
for all $i,j\in[n]$. Often we will consider functions of two vectors,
$g:\R^{n_{1}\times n_{2}}\rightarrow\R$, and wish to compute the
gradient and Hessian of $g$ restricted to one of the two vectors.
For $\vx\in\Rn$ and $\vy\in\Rm$ then we let $\grad_{\vx}\vg(\va,\vb)\in\R^{n_{1}}$
denote the gradient of $\vg$ for fixed $\vy$ at point $\{\va,\vb\}\in\R^{n_{1}\times n_{2}}$.
We define $\grad_{\vy}$, $\hess_{\vx\vx}$, and $\hess_{\vy\vy}$
similarly. Furthermore for $h:\R^{n}\rightarrow\R^{m}$ differentiable
at $\vx\in\Rn$ we let $\mj(\vh(\vx))\in\R^{m\times n}$ denote the
Jacobian of $\vh$ at $\vx$ where for all $i\in[m]$ and $j\in[n]$
we let $[\mj(\vh(\vx))]_{ij}\defeq\frac{\partial}{\partial x_{j}}h(\vx)_{i}$.
For functions of multiple vectors we use subscripts, e.g. $\mj_{\vx}$,
to denote the Jacobian of the function restricted to the $\vx$ variable.

\section{Preliminaries\label{sec:preliminaries}}

Here we provide a brief introduction to path following methods for
linear programming. The purpose of this section is to formally introduce
interior point terminology and methodology that we build upon to obtain
$\otilde(\sqrt{{\rank(\ma)}}L)$ iteration solver. The algorithm and
the analysis discussed in this section can be viewed as a special
case of the framework presented in Section~\ref{sec:weighted_path}.
The reader well versed in path following methods can likely skip this
section and to the more curious reader we encourage them to consider
some of the many wonderful expositions on this subject \cite{nesterov1997self,ye2011interior,gonzaga1992path}
for further reading.

\subsection{The Setup\label{sec:preliminaries:setup}}

Given a matrix, $\ma\in\R^{m\times n}$, and vectors, $\vb\in\Rm$
and $\vc\in\Rn$, the central goal of this paper is to efficiently
compute a solution to the following linear program
\begin{equation}
\min_{\vx\in\Rn~:~\ma\vx\geq\vb}\vc^{T}\vx\label{eq:preliminaries:lp_the_problem}
\end{equation}
It is well known that this is the dual of the \emph{standard form}
of a linear program and hence all linear programs can be expressed
by \eqref{eq:preliminaries:lp_the_problem}. We call a vector $\vx\in\Rm$
\emph{feasible} if $\ma\vx\geq\vb$, we call $\vc^{T}\vx$ the \emph{cost}
of such a vector. therefore our goal is to either compute a minimum
cost feasible vector or determine that none exists.

We assume that $\ma$ is full rank, i.e. $\rank(\ma)=n$, and that
$m\geq n$. Nevertheless, we still write many of our results using
$\rank(\ma)$ rather than $n$ for two reasons. First, this notation
makes clear that $\rank(\ma)$ is referring to the smaller of the
two quantities $m$ and $n$. Second, if $\rank(\ma)<n$, then we
can reduce the number of variables to $\rank(\ma)$ by a change of
basis.%
\footnote{In general, computing this change of basis may be computationally
expensive. However, this cost can be diminished by using a subspace
embedding~\cite{nelson2012osnap} to replace $\vx$ with $\mvar{\Pi}\vy$
for subspace embedding $\mvar{\Pi}$ and $\otilde(\rank(\ma))$ dimensional
$\vy$. Then using the reduction in Appendix~\ref{sec:app:bit_complexity}
we only need to work with an $\otilde(\rank(\ma))$ rank matrix.%
} Hence, we only need to solve linear programs in the full rank version.

\subsection{Path Following Interior Point\label{sec:preliminaries:central-path}}

Interior point methods\emph{ }solve \eqref{eq:preliminaries:lp_the_problem}
by maintaining a point $\vx$ that is in the \emph{interior} of the
feasible region, i.e. $\vx\in\dInterior$ where
\[
\dInterior\defeq\{\vx\in\Rn~:~\ma\vx>\vb\}.
\]
These methods attempt to iteratively decrease the cost of $\vx$ while
maintaining strict feasibility. This is often done by considering
some measurement of the distance to feasibility such as $\vSlackX\defeq\ma\vx-\vb$,
called the \emph{slacks}, and creating some penalty for these distances
approaching 0. Since $\vSlackX>0$ if and only if $\vx\in\dInterior$
by carefully balancing penalties for small $\vSlack(\vx)$ and penalties
for large $\vc^{T}\vx$ these methods eventually compute a point close
enough to the optimum solution that it can be computed exactly.

Path following methods fix ratios between the the penalty for large
$\vc^{T}\vx$ and the penalty for small $\vSlackX$ and alternate
between steps of optimizing with respect to this ratio and changing
the ratio. These methods typically encode the penalties through a
\emph{barrier function} $\phi~:~\dSlack\rightarrow\R$ such that $\phi(\vSlack(\vx))\rightarrow\infty$
as $\slack(\vx)_{i}\rightarrow0$ for any $i\in[m]$ and they encode
the ratio through some parameter $t>0$. Formally, they attempt to
solve optimization problems of the following form for increasing values
of $t$
\begin{equation}
\min_{\vx\in\Rm}\penalizedObjective(\vx)\enspace\text{ where }\enspace\penalizedObjective(\vx)\defeq t\cdot\vc^{T}\vx+\phi(\vSlack(\vx))\label{eq:preliminaries:standard-central-path}
\end{equation}
Since $\phi(\vSlack(\vx))\rightarrow\infty$ as $\slack(\vx)_{i}\rightarrow0$
the minimizer of $\penalizedObjective(\vx)$, denoted $\vx^{*}(t)$,
is in $S^{0}$ for all $t$. As $t$ increases the effect of the cost
vector on $\vx^{*}(t)$ increases and the distance from the boundary
of the feasible region as measured by $\vSlack(\vx)$ decreases. One
can think of the points $\{\vx^{*}(t)~|~t>0\}$ as a path in $\Rn$,
called \emph{the central path}, where $\vx^{*}(t)$ approaches a solution
to \eqref{eq:preliminaries:lp_the_problem} as $t\rightarrow\infty$.
A standard choice of barrier is \emph{the standard log barrier, $\phi(\vSlack(\vx))\defeq-\sum_{\{i\}}\log(\slack(\vx)_{i})$
}and for this choice of barrier we refer to $\{\vx^{*}(t)~|~t>0\}$
as the \emph{standard central path}.

Path following methods typically follow the following framework:
\begin{description}
\item [{(1)}] \emph{Compute Initial Point:} Compute an approximation $\vx^{*}(t)$
for some $t$.
\item [{(2)}] \emph{Follow the central path}: Repeatedly increase $t$
and compute an approximation to $\vx^{*}(t)$.
\item [{(3)}] \emph{Round to optimal solution}: Use the approximation to
$\vx^{*}(t)$ to compute the solution to \eqref{eq:preliminaries:lp_the_problem}.
\end{description}
Steps (1) and (3) are typically carried out by standard interior point
techniques. These techniques are fairly general and covered briefly
in Section \ref{sec:algorithm} and Appendix \ref{sec:app:bit_complexity}.
However, the manner in which (2) is performed varies greatly from
method to method. In the following subsection we provide a simple
technique for performing (2) that yields reasonable running times
and serves as the foundation for the algorithms considered in the
remainder of the paper.

\subsection{Following the Path\label{sec:preliminaries:following-the-path}}

There are numerous techniques to \emph{follow the central path}, i.e.
approximately compute $\vx^{*}(t)$ for increasing values of $t$.
Even with the barrier fixed there are numerous schemes to balance
maintaining a point close to a central path point, advancing to a
further central path point, and performing the numerical linear algebra
needed for these operations \cite{vaidya1989speeding,gonzaga1992path,mizuno1993adaptive,Nesterov1994}.

In this section we present a simple and common method whereby we simply
alternate between improving our distance to $\vx^{*}(t)$ for some
fixed $t$, and increasing $t$ by some fixed multiplicative factor.
This method reduces the analysis of path following to bounding the
computational complexity of \emph{centering}, i.e. improve the distance
to $\vx^{*}(t)$, and bounding how much increasing\emph{ $t$ }hurts
\emph{centrality}, i.e. increases the distance to $\vx^{*}(t)$. In
the remainder of this section we show how to perform this analysis
for the standard central path, \emph{$\phi(\vx)\defeq-\sum_{i\in[m]}\log(\slack(\vx)_{i})$.}

Typically path following methods center, i.e. minimize $\penalizedObjective(\vx)$,
using Newton's method or some variant thereof. While for an arbitrary
current point $\vx\in\dInterior$ and $t>0$ the function $\penalizedObjective(\vx)$
can be ill-behaved, in a region near $\vx^{*}(t)$ the Hessian of
$\penalizedObjective(\vx)$ given by $\hessian f_{t}(\vx)=\ma^{T}\ms^{-2}\ma$
for $\ms\defeq\mDiag(\vSlack(\vx))$ changes fairly slowly. More precisely,
if one considers the second order approximation of $f_{t}(\vz)$ around
some point $\vx\in\dInterior$ ``close enough'' to $\vx^{*}(t)$
, 
\[
f_{t}(\vz)\approx f_{t}(\vx)+\innerProduct{\grad f_{t}(\vx)}{\vz-\vx}+\frac{1}{2}\left(\vz-\vx\right)^{T}(\hessian f_{t}(\vx))\left(\vz-\vx\right)\,,
\]
and applies one step of Newton's method, i.e. minimizes this quadratic
approximation to compute
\begin{align*}
\vxNext & :=\vx-(\hessian f_{t}(\vx))^{-1}\grad f_{t}(\vx)\\
 & =\vx-(\ma^{T}\ms^{-2}\ma)^{-1}(t\vc-\ma^{T}\vs)
\end{align*}
for $\vs\defeq\vs(\vx)$ then this procedure rapidly converges to
$\vx^{*}(t)$.

To quantify this, we measure \emph{centrality}, i.e. how close the
current point $\vx\in\dInterior$ is to $\vx^{*}(t)$, by the size
of this \emph{Newton step} in the Hessian induced norm. For $\vx\in\dInterior$
and Newton step $\vh_{t}(\vx)\defeq(\hessian f_{t}(\vx))^{-1}\grad f_{t}(\vx)$
we denote centrality by $\delta_{t}(\vx)\defeq\norm{\vh_{t}(\vx)}_{\hessian f_{t}(\vx)}$.
Standard analysis of Newton's method shows that if $\delta_{t}(\vx)$
is less than some constant then for $\vxNext:=\vx-\vNewtonStep(\vx)$
we have $\delta_{t}(\vxNext)=O(\delta_{t}(\vx)^{2})$ (See Lemma~\ref{lem:weighted_path:stab_update_step}).
Furthermore, under these conditions it is not hard to show that for
$t'=t(1+(m)^{-1/2})$ we have $\delta_{t'}(\vxNext)\leq O\left(\delta_{t}(\vx)\right)$
(See Lemma \ref{lem:weighted_path:t_step}). 

Combining these facts yields that in $O(\sqrt{m})$ iterations we
can double $t$ while maintaining a \emph{nearly centered} $\vx$,
i.e. $\delta_{t}(\vx)$ at most a constant. With some additional work
discussed briefly in Section \ref{sec:algorithm} it can be shown
that by maintaining a nearly centered $\vx$ and changing $t$ by
a constant factor at most $\otilde(L)$ times one can compute a solution
to \eqref{eq:preliminaries:lp_the_problem}. Therefore, this method
solves \eqref{eq:preliminaries:lp_the_problem} in $O(\sqrt{m}L)$
iterations where the cost of each iteration is $O(\nnz(\ma))$ plus
the time need to solve a linear system in the matrix $\ma^{T}\ms^{-2}\ma$.

\section{Weighted Path Finding}

\label{sec:weighted_path}

In this section we introduce the optimization framework we use to
solve the linear programs, the \emph{weighted central path}. After
formally defining the path (Section~\ref{sec:weighted_path:path}),
we prove properties of the path (Section~\ref{sec:weighted_path:properties})
and show how to center along the path (Section \ref{sec:weighted_path:centering}).
We show that the performance of path following methods using a weighted
central path depends crucially on how the weights are computed and
in Section~\ref{sec:weighted_path:function} we characterize the
properties we require of such a weight function in order to ensure
that our weighted path following scheme converges efficiently. In
Section~\ref{sec:weighted_path:properties} we analyze the convergence
rate of our weighted path following scheme assuming the ability to
compute these weights exactly. In the following section we then show
how it suffices to compute the weights approximately (Section~\ref{sec:approx_path}),
we show how to compute these weights efficiently (Section~\ref{sec:weights_full}),
and we show how this yields an efficient linear program solver (Section
\ref{sec:algorithm}).

\subsection{The Weighted Path\label{sec:weighted_path:path}}

Our weighted path following method is a generalization of the path
following scheme presented in Section \ref{sec:preliminaries:central-path}.
Rather than keeping the barrier function $\phi(\vx)=-\sum_{i\in[m]}\log\slackXi$
fixed we allow for greater flexibility in how we penalize slack variables
and adaptively modify the barrier function in order to take larger
steps. In addition to maintaining a feasible point $\vx$ and a path
parameter $t$ we maintain a set of positive weights $\vWeight\in\dWeights$
and attempt to minimize the \emph{penalized objective function} $\penalizedObjective:\dInterior\times\dWeights\rightarrow\R$
given for all $\vx\in\dInterior$ and $\vWeight\in\dWeights$ by
\begin{equation}
\penalizedObjective(\vx,\vWeight)\defeq t\cdot\vc^{T}\vx-\sum_{i\in[m]}\weight_{i}\log\slackXi.\label{eq:penalized_object_function}
\end{equation}
We maintain a feasible point $\{\vx,\vWeight\}\in\dFull$ and our
goal is to compute a sequence of feasible points for increasing $t$
and changing $\vWeight$ such that $\penalizedObjective(\vx,\vWeight)$
is nearly minimized with respect to $\vx$.

Note that trivially any $\vx\in\dInterior$ can be expressed as $\argmin_{\vy\in\Rn}\penalizedObjective(\vy,\vWeight)$
for some $\vWeight\in\dWeights$ and therefore, every $\vx\in\dInterior$
is a \emph{weighted central path poin}t for some choice of weights.
However, in order to to convert a weighted central path point $\{\vx,\vWeight\}\in\dFull$
to a solution for (\ref{eq:intro:problem}) we will need to have $t$
large and $\normOne{\vw}$ small which precludes this trivial choice
of $t$ and $\vWeight$. 

In the remainder of the paper, we show that by careful updating $\vx$,
$\vWeight$, and $t$ we can stay close to the weighted central path
while making large increases in $t$ and maintaining $\norm{\vWeight}_{1}$
small. Ultimately, this will allow us to solve linear programs in
$\otilde(\sqrt{\rank(\ma)}L)$ iterations while only solving $\otilde(1)$
linear systems in each iteration.

\subsection{Properties of the Weighted Path\label{sec:weighted_path:properties}}

As in Section \ref{sec:preliminaries:following-the-path} for a feasible
$\{\vx,\vw\}\in\dFull$ we measure the \emph{centrality} of this point
by the size of the Newton step on $\vx$ in the Hessian norm, denoted
by $\proximity_{t}(\vx,\vw)$ and we call $\{\vx,\vw\}$ a central
path point if $\delta_{t}(\vx,\vw)=0$. For the penalized objective
function $\penalizedObjective$, we see that the Newton step, $\vNewtonStep_{t}(\vx,\vWeight)$,
is given by
\begin{align}
\vNewtonStep_{t}(\vx,\vWeight) & =(\hessXX\penalizedObjective(\vx,\vWeight))^{-1}\gradX\penalizedObjective(\vx,\vWeight)\nonumber \\
 & =(\ma^{T}\ms^{-1}\mWeight\ms^{-1}\ma)^{-1}(t\vc-\ma^{T}\ms^{-1}\vWeight)\label{eq:weighted_path:newton_step}
\end{align}
and the centrality, $\delta_{t}(\vx,\vw)$, is given by for all $\{\vx,\vWeight\}\in\dFull$
by 
\begin{align}
\delta_{t}(\vx,\vw) & \defeq\norm{\vNewtonStep_{t}(\vx,\vWeight)}_{\hessXX\penalizedObjective(\vx,\vWeight)}=\norm{t\vc-\ma^{T}\ms^{-1}\vWeight}_{\left(\ma^{T}\ms^{-1}\mWeight\ms^{-1}\ma\right)^{-1}}\label{eq:weighted_path:energy_x}
\end{align}

Whereas in the standard central path we saw that the centrality increased
at a rate of $\sqrt{m}$ as $t$ increased, here we show that in this
more general case, the $m$ is replaced by the total weight $\normOne{\vWeight}=\sum_{i\in[m]}\weight_{i}$.

\begin{lemma}[Weighted Path Step] \label{lem:weighted_path:t_step}
For all $\{\vx,\vWeight\}\in\dFull$ and $t,\alpha\geq0$, we have
\[
\proximity_{(1+\alpha)t}(\vx,\vWeight)\leq(1+\alpha)\proximity_{t}(\vx,\vWeight)+\alpha\sqrt{\normOne{\vWeight}}
\]
\end{lemma}\begin{proof} Let $\vs\defeq\vs(\vx)$. By (\ref{eq:weighted_path:energy_x})
we have
\[
\proximity_{(1+\alpha)t}(\vx,\vWeight)=\norm{(1+\alpha)t\vc-\ma^{T}\ms^{-1}\vWeight}_{(\ma^{T}\ms^{-1}\mWeight\ms^{-1}\ma)^{-1}}.
\]
Now, $\norm{\cdot}_{(\ma^{T}\ms^{-1}\mWeight\ms^{-1}\ma)^{-1}}$ is
a norm and therefore by the triangle inequality and the definition
of $\proximity_{t}(\vx,\vw)$ yields
\begin{equation}
\proximity_{(1+\alpha)t}(\vx,\vWeight)\leq(1+\alpha)\proximity_{t}(\vx,\vWeight)+\alpha\norm{\ma^{T}\ms^{-1}\vWeight}_{(\ma^{T}\ms^{-1}\mWeight\ms^{-1}\ma)^{-1}}.\label{eq:lem:weighted_path:t_step:1}
\end{equation}
Recall that $\mProj_{\ms^{-1}\ma}\left(\vWeight\right)=\mw^{1/2}\ms^{-1}\ma(\ma^{T}\ms^{-1}\mw\ms^{-1}\ma)^{-1}\ma^{T}\ms^{-1}\mw^{1/2}$
is a projection matrix. Consequently $\mProj_{\ms^{-1}\ma}\left(\vWeight\right)\specLeq\iMatrix$
and we have
\begin{equation}
\norm{\ma^{T}\ms^{-1}\vWeight}_{(\ma^{T}\ms^{-1}\mWeight\ms^{-1}\ma)^{-1}}=\norm{\mWeight^{-1/2}\vWeight}_{\mProj_{\ms^{-1}\ma}\left(\vWeight\right)}\leq\norm{\mWeight^{-1/2}\vWeight}_{2}=\sqrt{\sum_{i\in[m]}\weight_{i}}.\label{eq:lem:weighted_path:t_step:2}
\end{equation}
Combining (\ref{eq:lem:weighted_path:t_step:1}) and (\ref{eq:lem:weighted_path:t_step:2})
yields the result. \end{proof}

Now to see how well a Newton step on $\vx$ can center, i.e. decrease
$\proximity_{t}(\vx,\vw)$, we need to bound how fast the second order
approximation of $\penalizedObjective(\vx,\vw)$ can change, i.e.
how much the Hessian, $\hessXX\penalizedObjective(\vx,\vw)$, changes
as we change $\vx$. We do this by bounding how much the slacks can
change as we change $\vx$. As $\hessXX\penalizedObjective(\vx,\vw)=\ma^{T}\ms^{-1}\mWeight\ms^{-1}\ma$
this immediately bounds how much the Hessian can change as we change
$\vx$. The following lemma is motivated by similar results in \cite{vaidya1996new,anstreicher96}.

\begin{lemma}[Relative Change of Slacks] \label{lem:weighted_path:relative_change_slacks}
Let $\vxNext=\vx+\vDelta$ for some $\vx\in\dInterior$ and $\vDelta\in\R^{n}$.
Let $\vSlackNext$ and $\vs$ denote the slacks associated with $\vxNext$
and $\vx$ respectively. If $\normInf{\mSlack^{-1}\ma\vDelta}<1$
then $\vxNext\in\dInterior$ and
\begin{equation}
\normInf{\mSlack^{-1}\ma\vDelta}\leq\norm{\vDelta}_{\ma^{T}\mSlack^{-1}\mWeight\mSlack^{-1}\ma}\cdot\max_{i\in[m]}\norm{\mWeight^{-1/2}\indicVec i}_{\mProj_{\mSlack^{-1}\ma}\left(\vWeight\right)}.\label{eq:weighted_path:change_slacks}
\end{equation}
In particular, choosing $\vDelta=-\vNewtonStep_{t}(\vx,\vWeight)$
yields
\[
\normInf{\mSlack^{-1}\ma\vDelta}\leq\delta_{t}(\vs,\vWeight)\cdot\max_{i\in[m]}\norm{\mWeight^{-1/2}\indicVec i}_{\mProj_{\mSlack^{-1}\ma}\left(\vWeight\right)}.
\]
\end{lemma}

\begin{proof}

Clearly $\vSlackNext=\vs+\ma\vDelta$ and therefore the multiplicative
change in slacks is given by $\normInf{\mSlack^{-1}(\vSlackNext-\vs)}=\normInf{\mSlack^{-1}\ma\vDelta}$.
Consequently $\vxNext\in\dInterior$ if and only if $\normInf{\mSlack^{-1}\ma\vDelta}<1$. 

To prove (\ref{eq:weighted_path:change_slacks}) we note that by definition
of $\normInf{\cdot}$ 
\[
\normInf{\mSlack^{-1}\ma\vDelta}=\max_{i\in[m]}\left|\left\langle \mSlack^{-1}\ma\vDelta,\indicVec i\right\rangle \right|_{i}.
\]
Using that $\ma$ is full rank and therefore $\ma^{T}\mSlack^{-1}\mWeight\mSlack^{-1}\ma\specGt\mZero$
then yields
\[
\normInf{\mSlack^{-1}\ma\vDelta}=\max_{i\in[m[}\left|\left\langle \left(\ma^{T}\mSlack^{-1}\mWeight\mSlack^{-1}\ma\right)^{1/2}\vDelta,\left(\ma^{T}\mSlack^{-1}\mWeight\mSlack^{-1}\ma\right)^{-1/2}\ma^{T}\mSlack^{-1}\indicVec i\right\rangle \right|.
\]
Applying Cauchy Schwarz we have
\[
\normInf{\mSlack^{-1}\ma\vDelta}\leq\norm{\vDelta}_{\ma^{T}\mSlack^{-1}\mWeight\mSlack^{-1}\ma}\cdot\max_{i\in[m[}\norm{\ma^{T}\mSlack^{-1}\indicVec i}_{\left(\ma^{T}\mSlack^{-1}\mWeight\mSlack^{-1}\ma\right)^{-1}}.
\]
Recalling the definition $\mProj_{\mSlack^{-1}\ma}\left(\vWeight\right)=\mWeight^{1/2}\mSlack^{-1}\ma\left(\ma\mSlack^{-1}\mWeight\mSlack^{-1}\ma\right)^{-1}\ma^{T}\mSlack^{-1}\mWeight^{1/2}$
yields the result.\end{proof}

Lemma \ref{lem:weighted_path:relative_change_slacks} implies that
as $\norm{\mWeight^{-1/2}\indicVec i}_{\mProj_{\mSlack^{-1}\ma}\left(\vWeight\right)}$
decreases, the region over which Newton steps do not change the Hessian
too much increases. We call this quantity, $\norm{\mWeight^{-1/2}\indicVec i}_{\mProj_{\mSlack^{-1}\ma}\left(\vWeight\right)}$,
the \emph{slack sensitivity} as it measures how much slack changes
during a Newton step.

\begin{definition}[Slack Sensitivity] For $\vSlack,\vWeight\in\dSlack$
the \emph{slack sensitivity}%
\footnote{In the previous version in ArXiv, we called it weighted condition
number which is confusing. We are indebted to an anonymous reviewer
for suggesting this name.%
}, $\energyStab(\vSlack,\vWeight)$ is given by
\[
\energyStab(\vSlack,\vWeight)\defeq\max_{i\in[m]}\norm{\mWeight^{-1/2}\indicVec i}_{\mProj_{\mSlack^{-1}\ma}\left(\vWeight\right)}.
\]
\end{definition}

Geometrically, slack sensitivity indicates how much a relative slack
can change during a Newton step, equivalently, how small is the Newton
step region compared to the original polytope. From Lemmas \ref{lem:weighted_path:t_step}
and \ref{lem:weighted_path:relative_change_slacks} our goal in using
the weighted central path is clear. We wish to keep $\normOne{\vWeight}$
small so that we can make large increases to $t$ without increasing
centrality and we wish to keep $\energyStab(\vSlackX,\vWeight)$ small
so that over a large region we can improve centrality quickly. Unfortunately,
while it is not too difficult to produce weights that meet these criterion,
changing the weights can also increase $\proximity_{t}$. Therefore,
we also need to choose weights in such a way that they do not change
too drastically as we take Newton steps. In the next subsection we
introduce the step that we use to improve centrality and account for
possible changes in the weights.

\subsection{Centering Steps\label{sec:weighted_path:centering}}

Here we define the centering step we use to decrease $\delta_{t}(\vx,\vWeight)$.
There are two ways to decrease $\delta_{t}$, one is to perform a
Newton step on $\vx$ which corresponds to move $\vx$ closer to the
central path., one is to set $\vWeight$ such that $\delta_{t}(\vx,\vWeight)=0$
which corresponds to move the path itself to closer to $\vx$. By
mixing two steps, we can slow down progress along a specific weighted
path as much as we want but still obtaining the guarantee of Newton
method. We call this $r$-step where $r$ controls the ratio of how
much we change $\vWeight$ and $\vx$. Setting $r=0$ corresponds
to a standard Newton step on $\vx$ where the weights are not updated.
Setting $r=\infty$ coresponds to changing $\vWeight$ to make $\vx$
completely centered. There are two reasons we do thisinstead of a
standard Newton step:
\begin{enumerate}
\item When we change $\vx$, we need to change the weights $\vWeight$ accordingly
to maintain the the properties we want. However, when we change the
weights $\vWeight$, we need to update $\vx$ again, and so on. For
the weight function we consider in Section~\ref{sec:weights_full}
the change of $\vWeight$ required is large. Consequently, after updating
the weights we need to move $\vx$ even more and it is not clear how
to maintain good weights and good centrality at the same time if we
neglect the direction in which the weights change. However, the weights
we use actual change in a direction which partial helps improve centrality.
Considering a $r$-step helps us account for this progress directly. 
\item We cannot compute the weights we want to use exactly. Instead we only
know how to compute them approximately up to $1/\polylog(m)$ multiplicative
error using Johnson–Lindenstrauss. . Therefore, if we take a full
Newton step on $\vx$ and update the weights using the weight function,
the error in our approximation is possibly so large that the step
in full would not help centrality. To control this error and center
when we cannot compute the weights exactly, we exploit that the $r$-step
gives us part of the change in the weights that we can compute precisely.
\end{enumerate}
\begin{definition}[$r$-step]\label{def:weighted_path:r_step} Given
a feasible point $\{\vxCurr,\vWeightCurr\}\in\dFull$, a path parameter
$t$, and a $r$-step
\[
\{\vxNext,\vWeightNext\}=\updateStep_{t}(\vxCurr,\vWeightCurr,r)
\]
 is defined as follows
\begin{align*}
\vxNext & \defeq\vxCurr-\frac{1}{1+\stepRatio}\vNewtonStep_{t}(\vxCurr,\vWeightCurr),\\
\vWeightNext & \defeq\vWeightCurr+\frac{\stepRatio}{1+\stepRatio}\mWeightCurr\mSlackCurr^{-1}\ma\vNewtonStep_{t}(\vxCurr,\vWeightCurr)
\end{align*}
where we recall that
\[
\vNewtonStep_{t}(\vxCurr,\vWeightCurr)\defeq(\ma^{T}\mSlackCurr^{-1}\mWeightCurr\mSlackCurr^{-1}\ma)^{-1}(t\vc-\ma^{T}\mSlackCurr^{-1}\vWeightCurr)
\]
and we let $\vSlackCurr$ and $\vSlackNext$ denote the slacks with
$\vxCurr$ and $\vxNext$ respectively.\end{definition}

Note that for a $r$-step we have
\begin{equation}
\vSlackNext=\vSlackCurr-\frac{1}{1+r}\ma\vNewtonStep(\vxCurr,\vWeightCurr)\label{eq:weighted_path:rstep_slack}
\end{equation}
and therefore
\begin{equation}
\mWeightCurr^{-1}(\vWeightNext-\vWeightCurr)=-\stepRatio\mSlackCurr^{-1}(\vSlackNext-\vSlackCurr).\label{eq:weighted_path:rstep_weight_vs_slack}
\end{equation}
In other words, a $\stepRatio$-step performs a multiplicative update
on the weights that is exactly $\stepRatio$ times larger than the
update on the slacks.

Using Lemma \ref{lem:weighted_path:relative_change_slacks} we now
show that so long as $\delta_{t}(\vxCurr,\vWeightCurr)$ is reasonably
small with respect to the slack sensitivity, any $\stepRatio$-step
produces a feasible $\{\vxNext,\next{\vWeight}\}$ and does not change
the Hessian too much.

\begin{lemma}[Stability of $r$-step] \label{lem:weighted_path:stab_update_step}
Let $\{\vxNext,\vWeightNext\}=\updateStep_{t}(\vSlackCurr,\vWeightCurr,r)$
where
\[
\energyStab\defeq\energyStab(\vxCurr,\vWeightCurr)\enspace\text{ and }\enspace\delta_{t}\defeq\delta_{t}(\vxCurr,\vWeightCurr)\leq\frac{1}{8\energyStab}.
\]
Under these conditions we have
\begin{align}
\norm{\mSlackCurr^{-1}(\vSlackNext-\vSlackCurr)}_{\mWeightCurr} & \leq\frac{1}{1+\stepRatio}\cdot\delta_{t},\label{eq:lem:weighted_path:stab_update_step:0}\\
\normInf{\mSlackCurr^{-1}(\vSlackNext-\vSlackCurr)} & \leq\frac{1}{1+\stepRatio}\cdot\delta_{t}\cdot\energyStab,\label{eq:lem:weighted_path:stab_update_step:1}\\
\normInf{\mWeightCurr^{-1}(\vWeightNext-\vWeightCurr)} & \leq\frac{\stepRatio}{1+\stepRatio}\cdot\delta_{t}\cdot\energyStab.\label{eq:lem:weighted_path:stab_update_step:2}
\end{align}
Consequently $\{\vxNext,\vWeightNext\}$ is feasible and 
\begin{align}
(1-3\delta_{t}\energyStab)\hessXX\penalizedObjective(\vxCurr,\vWeightCurr)\specLeq\hessXX\penalizedObjective(\vxNext,\vWeightNext) & \specLeq(1+3\delta_{t}\energyStab)\hessXX\penalizedObjective(\vxCurr,\vWeightCurr)\label{eq:weighted_path:rstep_hessian}
\end{align}
\end{lemma}

\begin{proof} Equation (\ref{eq:lem:weighted_path:stab_update_step:0})
follows from the definition of $\delta_{t}$ and (\ref{eq:weighted_path:rstep_slack}).
Equations (\ref{eq:lem:weighted_path:stab_update_step:1}) and (\ref{eq:lem:weighted_path:stab_update_step:2})
follow from Lemma \ref{lem:weighted_path:relative_change_slacks},
the definition of $\energyStab$, (\ref{eq:weighted_path:rstep_slack}),
and (\ref{eq:weighted_path:rstep_weight_vs_slack}). Since $\delta_{t}\leq\frac{1}{8\gamma}$
this implies that slack or weight changes by more than a multiplicative
factor of $\frac{1}{8}$ and therefore clearly $\{\vSlackNext,\vWeightNext\}\in\dFull$. 

To prove (\ref{eq:weighted_path:rstep_hessian}) note that (\ref{eq:lem:weighted_path:stab_update_step:1})
and (\ref{eq:lem:weighted_path:stab_update_step:2}) imply that
\begin{eqnarray*}
 & \left(1-\frac{r}{1+r}\delta_{t}\energyStab\right)\mWeightCurr\specLeq\mWeightNext\specLeq\left(1+\frac{r}{1+r}\delta_{t}\energyStab\right)\mWeightCurr,\\
 & \left(1-\frac{1}{1+r}\delta_{t}\energyStab\right)\mSlackCurr\specLeq\mSlackNext\specLeq\left(1+\frac{1}{1+r}\delta_{t}\energyStab\right)\mSlackCurr.
\end{eqnarray*}
Since $\hessXX\penalizedObjective(\vx,\vWeight)=\ma^{T}\ms^{-1}\mWeight\ms^{-1}\ma$
for ${\vx,\vWeight}\in\dFull$ we have that
\[
\frac{\left(1-\frac{r}{1+r}\delta_{t}\energyStab\right)}{\left(1+\frac{1}{1+r}\delta_{t}\energyStab\right)^{2}}\hessXX\penalizedObjective(\vxCurr,\vWeightCurr)\specLeq\hessXX\penalizedObjective(\vxNext,\vWeightNext)\specLeq\frac{\left(1+\frac{r}{1+r}\delta_{t}\energyStab\right)}{\left(1-\frac{1}{1+r}\delta_{t}\energyStab\right)^{2}}\hessXX\penalizedObjective(\vxCurr,\vWeightCurr).
\]
Using that $0\leq\delta_{t}\energyStab\leq\frac{1}{8}$ and computing
the Taylor series expansions%
\footnote{Throughout this paper, when we use taylor series expansions we may
use more than just the second order approximation to the function.%
} yields that
\[
\frac{\left(1+\frac{r}{1+r}\delta_{t}\energyStab\right)}{\left(1-\frac{1}{1+r}\delta_{t}\energyStab\right)^{2}}\leq1+3\delta_{t}\energyStab\enspace\text{and}\enspace\frac{\left(1-\frac{r}{1+r}\delta_{t}\energyStab\right)}{\left(1+\frac{1}{1+r}\delta_{t}\energyStab\right)^{2}}\geq1-3\delta_{t}\gamma.
\]

\end{proof}

Using Lemma \ref{lem:weighted_path:stab_update_step} we now bound
how much a $\stepRatio$-step improves centrality.

\begin{lemma}[Centrality Improvement of $r$-step] \label{lem:weighted_path:x_progress}
Let $\{\vxNext,\vWeightNext\}=\updateStep_{t}(\vxCurr,\vWeightCurr,r)$
where
\[
\energyStab\defeq\energyStab(\vxCurr,\vWeightCurr)\enspace\text{ and }\enspace\delta_{t}\defeq\delta_{t}(\vxCurr,\vWeightCurr)\leq\frac{1}{8\energyStab}.
\]
We have the following bound on the change in centrality
\[
\delta_{t}(\vxNext,\vWeightNext)\leq\frac{2}{1+r}\cdot\energyStab\cdot\delta_{t}^{2}.
\]
\end{lemma}

\begin{proof} Let $\vNewtonStep_{t}\defeq\vNewtonStep_{t}(\vxCurr,\vWeightCurr)$
and let $\vDelta\defeq\mSlackCurr^{-1}(\vSlackNext-\vSlackCurr)=\frac{-1}{1+\stepRatio}\mSlackCurr^{-1}\ma\vh_{t}$.
Recalling the definition of $\updateStep_{t}$, we see that
\begin{align}
\frac{\vWeightNext_{i}}{\vSlackNext_{i}} & =\frac{\vWeightCurr_{i}-\stepRatio\vWeightCurr_{i}\vDelta_{i}}{\vSlackCurr_{i}+\vSlackCurr_{i}\vDelta_{i}}=\left(\frac{\vWeightCurr_{i}}{\vSlackCurr_{i}}\right)\cdot\left(\frac{1-\stepRatio\vDelta_{i}}{1+\vDelta_{i}}\right)\nonumber \\
 & =\left(\frac{\vWeightCurr_{i}}{\vSlackCurr_{i}}\right)\left(1-\frac{(1+\stepRatio)\vDelta_{i}}{1+\vDelta_{i}}\right)\label{eq:weighted_path:x_progress:1}
\end{align}
Using the definition of $\vh_{t}$ we have
\begin{align*}
\grad_{x}\penalizedObjective(\vxCurr,\vWeightCurr) & =t\vc-\ma^{T}\mSlackCurr^{-1}\vWeightCurr=\left(\ma^{T}\mSlackCurr^{-1}\mWeightCurr\mSlackCurr^{-1}\ma\right)\vh_{t}\\
 & =-(1+r)\ma^{T}\mSlackCurr^{-1}\mWeightCurr\vDelta
\end{align*}
and therefore
\begin{equation}
t\vc=\ma^{T}\mSlackCurr^{-1}\mWeightCurr\left(\onesVec-(1+r)\vDelta\right).\label{eq:weighted_path:x_progress2}
\end{equation}
Combining (\ref{eq:weighted_path:x_progress:1}) and (\ref{eq:weighted_path:x_progress2})
and using the definition of $\vDelta$ then yields
\begin{align}
\grad_{x}\penalizedObjective(\vxNext,\vWeightNext) & =t\vc-\ma^{T}\mSlackNext^{-1}\vWeightNext\nonumber \\
 & =\ma^{T}\mSlackCurr^{-1}\mWeightCurr\left(\onesVec-(1+r)\vDelta-\onesVec+\frac{(1+r)\vDelta}{\onesVec+\vDelta}\right)\nonumber \\
 & =-(1+r)\ma^{T}\mSlackCurr^{-1}\mWeightCurr\frac{\vDelta^{2}}{\onesVec+\vDelta}\nonumber \\
 & =\ma^{T}\mSlackCurr^{-1}\mWeightCurr\mSlackCurr^{-1}\mDiag(\vDelta)(\iMatrix+\mDiag(\vDelta))^{-1}\ma\vh_{t}\label{eq:sec_weighted_path:x_progress3}
\end{align}
Now by Lemma~\ref{lem:weighted_path:stab_update_step} we know that
\[
\ma^{T}\mSlackNext^{-1}\mWeightNext\mSlackNext^{-1}\ma\specGeq(1-3\delta_{t}\energyStab)\ma^{T}\mSlackCurr^{-1}\mWeightCurr\mSlackCurr^{-1}\ma.
\]
Therefore by (\ref{eq:sec_weighted_path:x_progress3}) and the fact
that
\[
\mProj_{\mSlackCurr^{-1}\ma}\left(\vWeightCurr\right)=\mWeightCurr^{1/2}\mSlackCurr^{-1}\ma\left(\ma^{T}\mSlackCurr^{-1}\mWeight\mSlackCurr^{-1}\ma\right)^{-1}\ma^{T}\mSlackCurr^{-1}\mWeightCurr^{1/2}\specLeq\iMatrix,
\]
we have
\begin{align*}
\delta_{t}(\vxNext,\vWeightNext) & =\norm{\grad_{x}\penalizedObjective(\vxNext,\vWeightNext)}_{\left(\ma^{T}\mSlackNext^{-1}\mWeightNext\mSlackNext^{-1}\ma\right)^{-1}}\\
 & \leq(1-3\delta_{t}\energyStab)^{-1/2}\norm{\mDiag(\vDelta)(\iMatrix+\mDiag(\vDelta))^{-1}\mWeightCurr^{1/2}\mSlackCurr^{-1}\ma\vNewtonStep}_{\mProj_{\mSlackCurr^{-1}\ma}\left(\vWeightCurr\right)}\\
 & \leq(1-3\delta_{t}\energyStab)^{-1/2}\norm{\mDiag(\vDelta)(\iMatrix+\mDiag(\vDelta))^{-1}\mWeightCurr^{1/2}\mSlackCurr^{-1}\ma\vNewtonStep}_{2}\\
 & \leq(1-3\delta_{t}\energyStab)^{-1/2}\frac{\normInf{\vDelta}}{1-\normInf{\vDelta}}\norm{\mWeightCurr^{1/2}\mSlackCurr^{-1}\ma\vNewtonStep}_{2}\\
 & =\left(1-3\delta_{t}\gamma\right)^{-1/2}\cdot\left(\frac{\normInf{\vDelta}}{1-\normInf{\vDelta}}\delta_{t}\right)\leq\frac{2}{1+r}\gamma\delta_{t}^{2}
\end{align*}
where in the last step we use that by Lemma~\ref{lem:weighted_path:stab_update_step},
$\normInf{\vDelta}\leq\frac{1}{1+\stepRatio}\delta_{t}\gamma$ and
that $\delta\leq\frac{1}{8\gamma}$ by assumption. \end{proof}

\subsection{Weight Functions\label{sec:weighted_path:function}}

In Sections \ref{sec:weighted_path:path}, \ref{sec:weighted_path:properties},
and \ref{sec:weighted_path:centering} we saw that to make our weighted
path following schemes to converge quickly we need to maintain weights
such that $\normOne{\vw}$, $\energyStab(\vSlack,\vw)$, and $\delta_{t}(\vx,\vw)$
are small. Rather than showing how to do this directly, here we assume
we have access to some fixed differentiable function for computing
the weights and we characterize when such a weight function yields
an efficient weighted path following scheme. This allows us to decouple
the problems of using weights effectively and computing these weights
efficiently.

For the remainder of this paper we assume that we have a fixed differentiable
\emph{weight function} $\fvWeight~:~\dSlack\rightarrow\dWeights$
from slacks to positive weights (see Section \ref{sec:weights_full}
for a description of the function we use). For slacks $\vSlack\in\dSlack$
we let $\fmWeight(\vSlack)\defeq\mDiag(\vg(\vSlack))$ denote the
diagonal matrix associated with the slacks and we let $\fmWeight'(\vSlack)\defeq\jacobian_{\vSlack}(\fvWeight(\vSlack))$
denote the Jacobian of the weight function with respect to the slacks. 

For the weight function to be useful, in addition to yielding weights
of small \emph{size}, i.e. $\normOne{\fvWeight(\vSlack)}$ bounded,
and good \emph{slack sensitivity}, i.e. $\energyStab(\vx,\fvWeight(\vs(\vx)))$
small, we need to ensure that the weights do not change too much as
we change $\vx$. For this, we use the operator norm of $\iMatrix+r^{-1}\fmWeight(\vSlack)^{-1}\fmWeight'(\vSlack)\mSlack$
to measure for how much the weight function can diverge from the change
in weights induced by a $\stepRatio$-step, i.e. how \emph{consistent}
$\fvWeight$ is to the central path. Lastly, to simplify the analysis
we make a uniformity assumption that none of the weights are two big,
i.e. $\normInf{\vg(\vs)}$ is bounded. Formally we define a weight
function as follows. \begin{definition}[Weight Function] \label{def:sec_weighted_path:weight_function}
A \emph{weight function} is a differentiable function from $\fvWeight:\rPos^{m}\rightarrow\rPos^{m}$
such that for constants $\cWeightSize(\vg)$, $\cWeightStab(\vg)$,
and $\cWeightCons(\vg)$, we have the following for all $\vs\in\dSlack$:
\begin{itemize}
\item \emph{Size }: The size \emph{$\cWeightSize(\fvWeight)$} satisfies
$\cWeightSize(\fvWeight)\geq\normOne{\fvWeight(\vSlack)}$
\item \emph{Slack Sensitivity}: The slack sensitivity $\cWeightStab(\fvWeight)$\emph{
}satisfies \emph{$\cWeightStab(\fvWeight)\geq1$} and $\energyStab(\vSlack,\fvWeight(\vSlack))\leq\cWeightStab(\fvWeight)$.
\item \emph{Step Consistency }:\emph{ }The step consistency\emph{ $\cWeightCons(\fvWeight)$}
satisfies $\cWeightCons(\fvWeight)\geq1$ and $\forall r\geq\cWeightCons(\fvWeight)$
and $\forall\vy\in\Rm$
\[
\norm{\iMatrix+r^{-1}\fmWeight(\vSlack)^{-1}\fmWeight'(\vSlack)\mSlack}_{\fmWeight(\vSlack)}\leq1\quad\text{{and}\quad}\norm{\left(\iMatrix+r^{-1}\fmWeight(\vSlack)^{-1}\fmWeight'(\vSlack)\mSlack\right)\vy}_{\infty}\leq\normInf{\vy}+\cWeightCons\norm{\vy}_{\fmWeight(\vSlack)}.
\]

\item \emph{Uniformity }: The weight function satisfies $\normInf{\fvWeight(\vSlack)}\leq2$
\end{itemize}
\end{definition}When the weight function $\vg$ is clear from context
we often write $\cWeightSize$, $\cWeightStab$, and $\cWeightCons$.

To get a sense of the magnitude of these parameters, in Theorem \ref{thm:weights_full:weight_properties}
we prove that there is a weight function with size $O(\sqrt{\rank\ma})$,
slack sensitivity $O(1)$ and step consistency $O\left(\log\left(\frac{m}{\rank\ma}\right)\right)$;
hence lemmas with polynomial dependence of slack sensitivity and step
consistency suffice for our purposes. However, for the remainder of
this section and Section \ref{sec:approx_path} we let the weight
function be fixed but arbitrary. 

Ideally, in our weighted path following schemes we would just set
$\vw=\vg(\vs)$ for any slacks $\vs$ we compute. However, actually
computing $\vg(\vs)$ may be expensive to compute exactly and therefore
we analyze schemes that maintain separate weights, $\vWeight\in\dWeights$
with the invariant that $\vw$ is close to $\fvWeight(\vs)$ multiplicatively.
Formally, we define $\vWeightError(\vs,\vw)$ for all $\vs,\vw\in\dWeights$
by 
\begin{equation}
\vWeightError(\vs,\vw)\defeq\log(\fvWeight(\vs))-\log(\vw)\label{eq:distance_of_sw}
\end{equation}
and attempt to keep $\normInf{\vWeightError(\vs,\vw)}$ small despite
changes that occur due to $\stepRatio$-steps.

Now we wish to show that a $\stepRatio$-step does not increase $\vWeightError(\vs,\vw)$
by too much. To do this, we first prove the following helper lemma.

\begin{lemma} \label{lem:weighted_path:weight_step_helper} For a
weight function $\fvWeight$ and $\vSlack_{0},\vSlack_{1}\in\dInterior$
such that 
\[
\epsilon_{\infty}\defeq\norm{\mSlack_{0}^{-1}(\vs_{1}-\vs_{0})}_{\infty}\leq\frac{1}{32\cWeightCons}\enspace\text{ and }\enspace\epsilon_{g}\defeq\norm{\mSlack_{0}^{-1}(\vs_{1}-\vs_{0})}_{\fmWeight(\vSlack_{0})}\leq\frac{\epsilon_{\infty}}{\cWeightCons}.
\]
we have 
\[
\normFull{\log\left(\frac{\vSlack_{1}}{\vSlack_{0}}\right)+\frac{1}{\cWeightCons}\log\left(\frac{\fvWeight(\vSlack_{1})}{\fvWeight(\vSlack_{0})}\right)}_{\infty}\leq3\epsilon_{\infty}\enspace\text{and}\enspace\normFull{\log\left(\frac{\vSlack_{1}}{\vSlack_{0}}\right)+\frac{1}{\cWeightCons}\log\left(\frac{\fvWeight(\vSlack_{1})}{\fvWeight(\vSlack_{0})}\right)}_{\fmWeight(\vSlack_{0})}\leq\left(1+6\cWeightCons\epsilon_{\infty}\right)\epsilon_{g}.
\]
\end{lemma}

\begin{proof} Let $\vp:\R^{m}\rightarrow\Rm$ be defined for all
$i\in[m]$ and $s\in\dSlack$ by
\[
\vp(\vSlack)_{i}\defeq\log(\vSlack_{i})+\frac{1}{\cWeightCons}\log(\fvWeight(\vSlack_{i})).
\]
Clearly $\jacobian_{\vSlack}[\vp(\vSlack)]=\mSlack^{-1}+\cWeightCons^{-1}\fmWeight^{-1}(\vSlack)\fmWeight'(\vSlack)$.
Therefore, letting $\vSlack_{t}\defeq\vSlack_{0}+t(\vSlack_{1}-\vSlack_{0})$
for all $t\in[0,1]$ we see that for all $i\in[0,1]$ ,
\[
\vp(\vSlack_{i})=\vp(\vSlack_{0})+\int_{0}^{i}\left[\mSlack_{t}^{-1}+\frac{1}{\cWeightCons}\fmWeight^{-1}(\vSlack_{t})\fmWeight'(\vSlack_{t})\right](\vSlack_{1}-\vSlack_{0})dt.
\]
Applying Jensen's inequality and the definition of $\vp$ then yields
that for all $i\in[0,1]$ and any norm $\norm{\cdot}$ ,
\begin{equation}
\normFull{\log\left(\frac{\vSlack_{i}}{\vSlack_{0}}\right)+\frac{1}{\cWeightCons}\log\left(\frac{\fvWeight(\vSlack_{i})}{\fvWeight(\vSlack_{0})}\right)}\leq\int_{0}^{i}\normFull{\left[\iMatrix+\frac{1}{\cWeightCons}\fmWeight^{-1}(\vSlack_{t})\fmWeight'(\vSlack_{t})\mSlack_{t}\right]\mSlack_{t}^{-1}(\vSlack_{1}-\vSlack_{0})}dt.\label{eq:weighted_path:eq1}
\end{equation}
Now for all $t\in[0,1]$ define $\va_{t}\in\dSlack$ by
\[
\va_{t}\defeq\log\left(\frac{\vSlack_{t}}{\vSlack_{0}}\right)-\frac{1}{\cWeightCons}\log\left(\frac{\fvWeight(\vSlack_{t})}{\fvWeight(\vSlack_{0})}\right)
\]
and let $M$ be the supremum over all $i\in[0,1]$ such that $\normInf{\log\fvWeight(\vSlack_{t})-\log\fvWeight(\vSlack_{0})}\leq3.5\cWeightCons\epsilon_{\infty}$
for all $t\in[0,i]$. By Lemma~\ref{lem:appendix:log_helper} and
the fact that $\epsilon_{\infty}\leq\frac{1}{32\cWeightCons}$ this
implies that $\normInf{\fmWeight(\vSlack_{i})^{-1}(\fvWeight(\vSlack_{i})-\fvWeight(\vSlack_{0}))}\leq4\cWeightCons\epsilon_{\infty}$
and $\normInf{\fmWeight(\vSlack_{0})^{-1}(\fvWeight(\vSlack_{0})-\fvWeight(\vSlack_{i}))}\leq4\cWeightCons\epsilon_{\infty}$
for all $i\in[0,M]$. Therefore, choosing $\norm{\cdot}_{\fmWeight(\vSlack_{0})}$
in (\ref{eq:weighted_path:eq1}) and applying Definition~\ref{def:sec_weighted_path:weight_function}
yields that $\forall i\in[0,M]$ ,
\[
\norm{\va_{i}}_{\fmWeight(\vSlack_{0})}<(1+4\cWeightCons\epsilon_{\infty})^{1/2}\int_{0}^{i}\norm{\mSlack_{t}^{-1}(\vSlack_{1}-\vSlack_{0})}_{\fmWeight(\vSlack_{t})}dt\leq\frac{(1+4\cWeightCons\epsilon_{\infty})}{1-\epsilon_{\infty}}\epsilon_{g}\leq\left(1+6\cWeightCons\epsilon_{\infty}\right)\epsilon_{g}.
\]
Similarly, by choosing $\normInf{\cdot}$ in (\ref{eq:weighted_path:eq1}),
we have $\forall i\in[0,M]$ that
\begin{align*}
\normInf{\va_{i}} & <\int_{0}^{i}\left(\norm{\mSlack_{t}^{-1}(\vSlack_{1}-\vSlack_{0})}_{\infty}+\cWeightCons\norm{\mSlack_{t}^{-1}(\vSlack_{1}-\vSlack_{0})}_{\fmWeight(\vSlack_{t})}\right)dt\\
 & <\frac{\epsilon_{\infty}}{1-\epsilon_{\infty}}+\frac{\sqrt{1+4\cWeightCons\epsilon_{\infty}}}{1-\epsilon_{\infty}}\cWeightCons\epsilon_{g}\leq2.2\epsilon_{\infty}
\end{align*}
By the definition of $\va_{i}$, the triangle inequality, and Lemma~\ref{lem:appendix:log_helper}
we then have that
\[
\normInf{\log(\fvWeight(\vSlack_{i}))-\log(\fvWeight(\vSlack_{0}))}<\cWeightCons\left(2.2\epsilon_{\infty}+\normInf{\log(\vSlack_{i})-\log(\vSlack_{0})}\right)<3.5\cWeightCons\epsilon_{\infty}.
\]
Since $\fvWeight$ is continuous we have that $M=1$ and the result
follows. \end{proof}

Using this lemma we bound on how much a $r$-step increases $\vWeightError$
as follows

\begin{lemma} \label{lem:weighted_path:weight_progress} Let $\{\vxNext,\vWeightNext\}=\updateStep_{t}(\vxCurr,\vWeightCurr,\cWeightCons)$
where
\[
\delta_{t}\defeq\delta_{t}(\vxCurr,\vWeightCurr)\leq\frac{1}{8\cWeightStab}\enspace\text{ and }\enspace\epsilon\defeq\normInf{\log(\fvWeight(\vSlackCurr))-\log(\vWeightCurr)}\leq\frac{1}{8}\enspace.
\]
Letting
\[
\vDelta\defeq\log\left(\frac{\fvWeightNext}{\fvWeightCurr}\right)-\log\left(\frac{\vWeightNext}{\vWeightCurr}\right)=\vWeightError(\vSlackNext,\vWeightNext)-\vWeightError(\vSlackCurr,\vWeightCurr),
\]
we have
\[
\norm{\vDelta}_{\infty}\leq4\cWeightStab\delta_{t}\enspace\text{ and }\enspace\norm{\vDelta}_{\mWeightNext}\leq\frac{e^{\epsilon}\cWeightCons}{1+\cWeightCons}\delta_{t}+13\cWeightStab\delta_{t}^{2}.
\]
\end{lemma}

\begin{proof} Recall the following definition of slack sensitivity
\[
\energyStab(\vSlack,\vWeight)=\max_{i\in[m]}\norm{\mWeight^{-1/2}\indicVec i}_{\mProj_{\mSlack^{-1}\ma}\left(\vWeight\right)}=\max_{i\in[m]}\norm{\indicVec i}_{\mSlack^{-1}\ma\left(\ma^{T}\mSlack^{-1}\mWeight\mSlack^{-1}\ma\right)^{-1}\ma^{T}\mSlack^{-1}}.
\]
Since $\normInf{\log(\fvWeight(\vSlackCurr))-\log(\vWeightCurr)}\leq\frac{1}{8}$,
we have 
\begin{equation}
\energyStab(\vSlackCurr,\vWeightCurr)\leq\sqrt{\frac{8}{7}}\energyStab(\vSlackCurr,\fvWeight(\vSlackCurr))\leq1.1c_{\gamma}.\label{eq:weighted_path:w_progress_1}
\end{equation}
Therefore, since $\delta_{t}\leq\frac{1}{64\cWeightStab c_{r}}$,
by Lemma~\ref{lem:weighted_path:stab_update_step} and (\ref{eq:weighted_path:w_progress_1})
we have
\begin{equation}
\normInf{\mWeightCurr^{-1}(\vWeightNext-\vWeightCurr)}\leq\frac{1.1\cWeightCons c_{\gamma}\delta_{t}}{1+\cWeightCons}\leq\frac{1}{2}\enspace\text{and}\enspace\normInf{\mSlackCurr^{-1}(\vSlackNext-\vSlackCurr)}\leq\frac{1.1c_{\gamma}\delta_{t}}{1+\cWeightCons}\leq\frac{1}{2}.\label{eq:weighted_path_w_progress_1}
\end{equation}
Recalling that $\mWeightCurr^{-1}(\vWeightNext-\vWeightCurr)=-\cWeightCons\mSlackCurr^{-1}(\vSlackNext-\vSlackCurr)$
and using that $\cWeightCons\geq1$ and $\epsilon-\epsilon^{2}\leq\log(1+\epsilon)\leq\epsilon$
for $|\epsilon|<\frac{1}{2}$ we have that for all $i\in[m]$
\begin{align}
\left|\log\left(\frac{\new w_{i}}{\old w_{i}}\right)+\cWeightCons\log\left(\frac{\new s_{i}}{\old s_{i}}\right)\right| & =\left|\log\left(1-\cWeightCons\frac{\new s_{i}-\old s_{i}}{\old s{}_{i}}\right)+\cWeightCons\log\left(1+\frac{\new s_{i}-\old s_{i}}{\old s{}_{i}}\right)\right|\nonumber \\
 & \leq2\cWeightCons^{2}\left|\frac{\vSlackNext_{i}-\vSlackCurr_{i}}{\vSlackCurr_{i}}\right|^{2}\label{eq:weighted_path:w_progress_2}
\end{align}
Letting $\norm{\cdot}$ denote either $\norm{\cdot}_{\infty}$ or
$\norm{\cdot}_{\mWeightCurr}$, recalling that $\normInf{\mSlackCurr^{-1}(\vSlackNext-\vSlackCurr)}\leq\frac{1.1c_{\gamma}\delta_{t}}{1+\cWeightCons}\leq\frac{1.1c_{\gamma}\delta_{t}}{\cWeightCons}$,
and applying (\ref{eq:weighted_path:w_progress_2}) yields
\begin{eqnarray}
\norm{\vDelta} & \leq & \normFull{\cWeightCons\log\left(\frac{\vSlackNext}{\vSlackCurr}\right)+\log\left(\frac{\fvWeight(\vSlackNext)}{\fvWeight(\vSlackCurr)}\right)}+\normFull{\log\left(\frac{\vWeightNext}{\vWeightCurr}\right)+\cWeightCons\log\left(\frac{\vSlackNext}{\vSlackCurr}\right)}\nonumber \\
 & \leq & \cWeightCons\normFull{\log\left(\frac{\vSlackNext}{\vSlackCurr}\right)+\frac{1}{\cWeightCons}\log\left(\frac{\fvWeight(\vSlackNext)}{\fvWeight(\vSlackCurr)}\right)}+2.2\cWeightCons\cWeightStab\delta_{t}\norm{\mSlackCurr^{-1}(\vSlackNext-\vSlackCurr)}.\label{eq:weighted_path:w_progress_5}
\end{eqnarray}
By Lemma~\ref{lem:weighted_path:stab_update_step} and (\ref{eq:weighted_path_w_progress_1}),
$\vSlackCurr$ and $\vSlackNext$ meet the conditions of Lemma~\ref{lem:weighted_path:weight_step_helper}
with $\epsilon_{\infty}\leq\frac{1.1\cWeightStab\delta_{t}}{1+\cWeightCons}$
and $\epsilon_{g}\leq\frac{e^{\epsilon/2}\delta_{t}}{1+\cWeightCons}$.
Therefore, letting $\norm{\cdot}$ be $\normInf{\cdot}$ in (\ref{eq:weighted_path:w_progress_5}),
we have
\begin{eqnarray*}
\normInf{\vDelta} & \leq & 3\cWeightCons\epsilon_{\infty}+2.2\cWeightCons c_{\gamma}\delta_{t}\frac{1.1\cWeightStab\delta_{t}}{1+\cWeightCons}\leq4\cWeightStab\delta_{t}.
\end{eqnarray*}
Similarly, letting $\norm{\cdot}$ be $\norm{\cdot}_{\mWeightCurr}$
in (\ref{eq:weighted_path_w_progress_1}) and noting that by definition
of $\epsilon$ yields
\begin{align*}
\norm{\vDelta}_{\mWeightCurr} & \leq e^{\epsilon/2}\cWeightCons\epsilon_{g}(1+6\cWeightCons\epsilon_{\infty})+2.2\cWeightCons\cWeightStab\delta_{t}\frac{\delta_{t}}{1+c_{r}}\\
 & \leq e^{\epsilon}\frac{\cWeightCons}{1+\cWeightCons}\delta_{t}+10\cWeightStab\delta_{t}^{2}.
\end{align*}
Finally, noting that $\normInf{\mWeightCurr^{-1}(\vWeightNext-\vWeightCurr)}\leq1.1\cWeightStab\delta_{t}$
yields the result.

\end{proof}

\subsection{Centering Using Exact Weights \label{sec:weighted_path:exact_centering}}

Here we bound the rate of convergence rate of path following assuming
that we can compute the weight function $\fvWeight$ exactly. We start
by providing a basic lemma regarding how the Newton step size changes
as we change $\vWeight$.

\begin{lemma}[Effect of Weight Change]  \label{lem:weight_change}
Let $\vx\in\dInterior$ and let $\vWeightCurr,\vWeightNext\in\dWeights$
with
\begin{equation}
\epsilon_{\infty}\defeq\normInf{\log(\vWeightNext)-\log(\vWeightCurr)}\leq\frac{1}{2},\label{eq:lem:weight_change}
\end{equation}
it follows that
\[
\delta_{t}(\vx,\vWeightNext)\leq(1+\epsilon_{\infty})\left[\delta_{t}(\vx,\vWeightCurr)+\norm{\log(\vWeightNext)-\log(\vWeightCurr)}_{\mWeightCurr}\right]
\]
\end{lemma}

\begin{proof} Let $\mHessCurr\defeq\ma^{T}\ms^{-1}\mWeightCurr\ms^{-1}\ma$
and let $\mHessNext\defeq\ma^{T}\ms^{-1}\mWeightNext\ms^{-1}\ma$.
By the definition of $\delta_{t}$ and the triangle inequality we
have
\begin{align}
\delta_{t}(\vx,\vWeightNext) & =\norm{t\vc-\ma^{T}\ms^{-1}\vWeightNext}_{\mHessNext^{-1}}\nonumber \\
 & \leq\norm{t\vc-\ma^{T}\ms^{-1}\vWeightCurr}_{\mHessNext^{-1}}+\norm{\ma^{T}\ms^{-1}\vWeightNext-\ma^{T}\ms^{-1}\vWeightCurr}_{\mHessNext^{-1}}\label{eq:wchange:1}
\end{align}
By definition of $\epsilon_{\infty}$ and Lemma~\ref{lem:appendix:log_helper}
$\mHessNext^{-1}\specLeq(1+\epsilon_{\infty})^{2}\mHessCurr^{-1}$
and therefore
\begin{equation}
\norm{t\vc-\ma\ms^{-1}\vWeightCurr}_{\mHessNext^{-1}}\leq(1+\epsilon_{\infty})\delta_{t}(\vx,\vWeightCurr).\label{eq:wchange:2}
\end{equation}
Furthermore, since $\mProj_{\ma\ms^{-1}}(\vWeightNext)\specLeq\iMatrix$
we have
\begin{align}
\norm{\ma^{T}\ms^{-1}\vWeightNext-\ma^{T}\ms^{-1}\vWeightCurr}_{\mHessNext^{-1}} & =\norm{\mWeightNext^{-1/2}(\vWeightNext-\vWeightCurr)}_{\mProj_{\ma\ms^{-1}}(\vWeightNext)}\nonumber \\
 & \leq\normFull{\frac{\vWeightNext-\vWeightCurr}{\sqrt{\vWeightNext\vWeightCurr}}}_{\mWeightCurr}\label{eq:wchange:3}
\end{align}
Using that $\frac{(e^{x}-1)^{2}}{e^{x}}\leq(1+\left|x\right|)^{2}x^{2}$
for $|x|\leq\frac{1}{2}$ and letting $x=\left[\log(\vWeightNext)-\log(\vWeightCurr)\right]_{i}$
we have
\begin{equation}
\normFull{\frac{\vWeightNext-\vWeightCurr}{\sqrt{\vWeightNext\vWeightCurr}}}_{\mWeightCurr}\leq(1+\epsilon_{\infty})\norm{\log(\vWeightNext)-\log(\vWeightCurr)}_{\mWeightCurr}\label{eq:wchange:4}
\end{equation}
Combining (\ref{eq:wchange:1}), (\ref{eq:wchange:2}), (\ref{eq:wchange:3}),
and (\ref{eq:wchange:4}) completes the proof. \end{proof}

\begin{center}
\begin{tabular}{|l|}
\hline 
$\vx^{(new)}=\mathbf{\centeringExact}(\vxCurr)$\tabularnewline
\hline 
\hline 
1. $\vx^{(new)}=\vxCurr-\frac{1}{1+\cWeightCons}\vNewtonStep(\vxCurr,\fvWeight(\vSlackCurr)).$\tabularnewline
\hline 
\end{tabular}
\par\end{center}

With this lemma we can now show how much centering progress we make
by just updating $\vx$ and using the weight function. Note that in
this proof we are just using the $r$-step in the proof, not the algorithm
itself. We will need to use the $r$-step itself only later when we
drop the assumption that we can compute $\vg$ exactly.

\begin{theorem}[Centering with Exact Weights] \label{thm:centering_exact}
Fix a weight function $\fvWeight$, let $\vxCurr\in\dInterior$, and
let
\[
\vx^{(new)}=\mathbf{\centeringExact}(\vxCurr)
\]
If
\[
\delta_{t}\defeq\delta_{t}(\vxCurr,\fvWeightCurr)\leq\frac{1}{80\cWeightStab\cWeightCons}
\]
then
\[
\delta_{t}(\vxNext,\fvWeightNext)\leq\left(1-\frac{1}{4\cWeightCons}\right)\delta_{t}(\vxCurr,\fvWeightCurr).
\]
\end{theorem}

\begin{proof} Let $\{\vxNext,\vWeightNext\}\in\dFull$ be the result
of a $\cWeightCons$ step from $\{\old{\vx},\vWeightCurr\}\in\dFull$.
Note that this $\vSlackNext$ is the same as the $\vSlackNext$ in
the theorem statement.

Now by Lemma~\ref{lem:weighted_path:x_progress} we have that
\begin{equation}
\delta_{t}(\vSlackNext,\vWeightNext)\leq\cWeightStab\delta_{t}^{2}\enspace.\label{eq:thm:centering_exact:1}
\end{equation}
Furthermore, defining $\vDelta$ as in Lemma~\ref{lem:weighted_path:weight_progress}
and noting that $\vWeightCurr=\fvWeight(\vSlackCurr)$ we have
\[
\vDelta\defeq\log\left(\frac{\fvWeight(\vSlackNext)}{\fvWeight(\vSlackCurr)}\right)-\log\left(\frac{\vWeightNext}{\vWeightCurr}\right)=\log\left(\frac{\fvWeight(\vSlackNext)}{\vWeightNext}\right).
\]
we see by Lemma~\ref{lem:weighted_path:weight_progress} that
\begin{equation}
\normInf{\log(\fvWeight(\vSlackNext)/\vWeightNext)}\leq4\cWeightStab\delta_{t}\leq1/2\label{eq:thm:centering_exact:2}
\end{equation}
and
\begin{equation}
\norm{\log(\fvWeight(\vSlackNext)/\vWeightNext)}_{\vWeightNext}\leq\frac{e^{\epsilon}\cWeightCons}{1+\cWeightCons}\delta_{t}+13\cWeightStab\delta_{t}^{2}\label{eq:thm:centering_exact:3}
\end{equation}
with $\epsilon=0$ because we are using exact weight computation.
Applying Lemma~\ref{lem:weight_change} to (\ref{eq:thm:centering_exact:1}),
(\ref{eq:thm:centering_exact:2}), and (\ref{eq:thm:centering_exact:3})
we have
\begin{align*}
\delta_{t}(\vxNext,\fvWeight(\next{\vs})) & \leq(1+4\cWeightStab\delta_{t})\left[\cWeightStab\delta_{t}^{2}+\frac{\cWeightCons}{1+\cWeightCons}\delta_{t}+13\cWeightStab\delta_{t}^{2}\right]\\
 & \leq\frac{\cWeightCons}{1+\cWeightCons}\delta_{t}+20\cWeightStab\cWeightCons\delta_{t}^{2}\\
 & \leq\left(1-\frac{1}{2\cWeightCons}+\frac{1}{4\cWeightCons}\right)\delta_{t}\leq\left(1-\frac{1}{4\cWeightCons}\right)\delta_{t}
\end{align*}
\end{proof}

From this lemma we have that if $\delta_{t}(\vx,\vg(\vs))$ is $O(\cWeightStab^{-1}\cWeightCons^{-1})$
then in $\Theta(\cWeightCons^{-1})$ iterations of $\code{CenteringExact}$
we can decrease $\delta_{t}(\vx,\vg(\vs))$ by a multiplicative constant.
Furthermore by Lemma~\ref{lem:weighted_path:t_step} we see that
we can increase $t$ by a multiplicative $(1+O(\cWeightStab^{-1}\cWeightCons^{-1}\cWeightSize^{-1/2}))$
and maintain $\delta_{t}(\vx,\vg(\vs))=O(\cWeightStab^{-1}\cWeightCons^{-1})$.
Thus we can double $t$ and maintain $\delta_{t}(\vx,\vg(\vs))=O(\cWeightStab^{-1}\cWeightCons^{-1})$
using $O(\cWeightStab^{-1}\cWeightCons^{-2}\cWeightSize^{-1/2})$
iterations of $\code{CenteringExact}$. In Section~\ref{sec:algorithm}
we make this argument rigorously in the more general setting. In the
following sections, we show how to relax this requirement that $\vg$
is computed exactly.

\section{A Weight Function for \texorpdfstring{$\otilde(\sqrt{\rank(\ma)}L)$}{Almost Optimal}
Convergence}

\label{sec:weights_full}

Here, we present the weight function $\vg:\dSlack\rightarrow\dWeights$
that when used in the framework proposed in Section~\ref{sec:weighted_path}
yields an $\otilde(\sqrt{\rank(\ma)}L)$ iteration interior point
method. In Section~\ref{sec:weight_function:weight_function} we
motivate and describe the weight function $\fvWeightFull$, in Section~\ref{sec:weights_full:properties}
we prove that $\fvWeightFull$ satisfies Definition~\ref{def:sec_weighted_path:weight_function}
with nearly optimal $\cWeightSize(\fvWeightFull)$, $\cWeightStab(\fvWeightFull)$,
and $\cWeightCons(\fvWeightFull)$, and in Section~\ref{sec:weights_full:computing}
we show how to compute and correct approximations to $\fvWeightFull$
efficiently.

\subsection{The Weight Function}

\label{sec:weight_function:weight_function}

Our weight function was inspired by the volumetric barrier methods
of \cite{vaidya1996new,anstreicher96}.%
\footnote{ See Section \ref{sec:our_approach} for further intuition.%
} These papers considered using the \emph{volumetric barrier,} $\phi(\vs)=-\log\det(\ma^{T}\ms^{-2}\ma)$
, in addition to the standard log barrier, $\phi(\vs)=-\sum_{i\in[m]}\log s_{i}$.
In some sense the standard log barrier has a good slack sensitivity,
$1$, but a large size, $m$, and the volumetric barrier has a worse
slack sensitivity, $\sqrt{m}$, but better total weight, $n$. By
carefully applying a weighted combination of these two barriers \cite{vaidya1996new}
and \cite{anstreicher96} achieved an $O((m\rank(\ma))^{1/4}L)$ iteration
interior point method at the cost more expensive linear algebra in
each iteration. 

Instead of using a fixed barrier, our weight function $\vg:\dSlack\rightarrow\dWeights$
is computed by solving a convex optimization problem whose optimality
conditions imply both good size and good slack sensitivity. We define
$\vg$ for all $\vSlack\in\dSlack$ by
\begin{equation}
\vg(\vSlack)\defeq\argmin_{\vWeight\in\rPos^{m}}\penalizedObjectiveWeight(\vSlack,\vWeight)\enspace\text{ where }\enspace\penalizedObjectiveWeight(\vSlack,\vWeight)\defeq\onesVec^{T}\vWeight-\frac{1}{\alpha}\log\det(\ma_{s}^{T}\mWeight^{\alpha}\ma_{s})-\beta\sum_{i\in[m]}\log\weight_{i}\label{eq:sec:weights_full:weight_function}
\end{equation}
where here and in the remainder of this section we let $\ma_{s}\defeq\mSlack^{-1}\ma$
and the parameters $\alpha,\beta\in\R$ are chosen later such that
the following hold
\begin{equation}
\alpha\in(0,1)\enspace\text{ , }\enspace\beta\in(0,1)\enspace\text{ , and }\enspace\beta^{1-\alpha}\geq\frac{1}{2}\enspace\text{ .}\label{eq:weights:constants_assumptions}
\end{equation}

To get a sense for why $\vg$ has the desired properties, , suppose
for illustration purposes that $\alpha=1$ and $\beta=0$ and fix
$\vs\in\dSlack$. Using Lemma~\ref{lem:deriv:log_det} and setting
the gradient of \eqref{eq:sec:weights_full:weight_function} to $\vzero$
we see that if $\vg$ exists then
\[
\vg(\vSlack)=\vsigma_{\ma_{s}}(\vg(\vSlack))\defeq\diag\left((\fmWeight(\vSlack))^{1/2}\ma_{s}(\ma_{s}^{T}\fmWeight(\vSlack)\ma_{s})^{-1}\ma_{s}^{T}(\fmWeight(\vSlack))^{1/2}\right)
\]
where we use the definition of $\vsigma_{\ma_{s}}$ from Section~\ref{sec:Notation}.
Consequently,
\[
\max_{i}\norm{\fmWeight{}^{-1/2}\indicVec i}_{\mProj_{\ma_{s}}\left(\vg\right)}=1\enspace\text{ and }\enspace\energyStab(\vSlack,\vg(\vSlack))=1\enspace.
\]
Furthermore, since $(\fmWeight(\vSlack))^{1/2}\ma_{s}(\ma_{s}^{T}\fmWeight(\vSlack)\ma_{s})^{-1}\ma_{s}^{T}(\fmWeight(\vSlack))^{1/2}$
is a projection matrix, $\normOne{\vsigma_{\ma_{s}}(\vg(\vSlack))}=\rank(\ma)$.
Therefore, this would yield a weight function with good $\cWeightStab$
and $\cWeightSize$. 

Unfortunately picking $\alpha=1$ and $\beta=0$ makes the optimization
problem for computing $\vg$ degenerate. In particular for this choice
of $\alpha$ and $\beta$, $\vg(\vSlack)$ could be undefined. In
the follow sections we will see that by picking better values for
$\alpha$ and $\beta$ we can trade off how well $\vg$ performs as
a weight function and how difficult it is to compute approximations
to $\vg$.

\subsection{Weight Function Properties}

\label{sec:weights_full:properties}

Here, we show that $\vg:\rNonNeg\rightarrow\rNonNeg$ as given by
\eqref{eq:sec:weights_full:weight_function} is a weight function
with respect to Definition~\ref{def:sec_weighted_path:weight_function}
and we bound the values of $\cWeightSize(\vg)$, $\cWeightStab(\vg)$,
and $\cWeightCons(\fvWeight)$. The goal of this section is to prove
the following.

\begin{theorem}[Properties of Weight Function] \label{thm:weights_full:weight_properties}
Let us define $\alpha$ and $\beta$ by
\[
\alpha=1-\frac{1}{\log_{2}\left(\frac{2m}{\rank(\ma)}\right)}\enspace\text{ and }\enspace\beta=\frac{\rank(\ma)}{2m}
\]
For this choice of parameters $\vg$ is a weight function meeting
the criterion of Definition~\ref{def:sec_weighted_path:weight_function}
with 
\begin{itemize}
\item Size : $\cWeightSize(\fvWeightFull)=2\rank(\ma)$.
\item Slack Sensitivity: $\cWeightStab(\fvWeightFull)=2$.
\item Step Consistency : $\cWeightCons(\fvWeight)=2\log_{2}\left(\frac{2m}{\rank(\ma)}\right)$.
\end{itemize}
\end{theorem}

We break the proof into several parts. In Lemma~\ref{lem:weights_full:derivatives_of_objective},
we prove basic properties of $\penalizedObjectiveWeight$. In Lemma~\ref{lem:weights_full:existence_and_size}
we prove that $\vg$ is a weight function and bound its size. In Lemma~\ref{lem:weights_full:cond_number}
we bound the slack sensitivity of $\vg$ and in Lemma~\ref{lem:weights_full:consistency}
we show that $\vg$ is consistent. 

We start by computing the gradient and Hessian of $\penalizedObjectiveWeight(\vSlack,\vWeight)$
with respect to $\vWeight$.

\begin{lemma}\label{lem:weights_full:derivatives_of_objective} For
all $\vSlack\in\dSlack$ and $\vWeight\in\dWeights$, we have
\[
\grad_{\vw}\penalizedObjectiveWeight(\vSlack,\vWeight)=\left(\iMatrix-\mSigma\mWeight^{-1}-\beta\mw^{-1}\right)\onesVec\enspace\text{ and }\enspace\hessian_{\vw\vw}\penalizedObjectiveWeight(\vSlack,\vWeight)=\mw^{-1}\left(\mSigma+\beta\iMatrix-\alpha\mLambda\right)\mw^{-1}
\]
where $\mSigma\defeq\mSigma_{\ma_{s}}(\mw^{\alpha}\onesVec)$ and
$\mLapProj\defeq\mLapProj_{\ma_{s}}(\mw^{\alpha}\onesVec)$.

\end{lemma}

\begin{proof} Using Lemma~\ref{lem:deriv:log_det} and the chain
rule we compute the gradient of $\grad_{w}\penalizedObjectiveWeight(\vSlack,\vWeight)$
as follows
\begin{eqnarray*}
\grad_{\vw}\penalizedObjectiveWeight(\vSlack,\vWeight) & = & \onesVec-\frac{1}{\alpha}\mSigma\mw^{-\alpha}\left(\alpha\mWeight^{\alpha-1}\right)-\beta\mWeight^{-1}\onesVec\\
 & = & \left(\iMatrix-\mSigma\mWeight^{-1}-\beta\mWeight^{-1}\right)\onesVec
\end{eqnarray*}
Next, using Lemma~\ref{lem:deriv:lever} and chain rule, we compute
the following for all $i,j\in[m]$,
\begin{align*}
\frac{\partial(\grad_{\vw}\penalizedObjectiveWeight(\vSlack,\vWeight))_{i}}{\partial w_{j}} & =-\frac{w_{i}\mLambda_{ij}\vWeight_{j}^{-\alpha}\left(\alpha\vWeight_{j}^{\alpha-1}\right)-\mSigma_{ij}\iMatrix_{ij}+\beta\iMatrix_{ij}}{\vWeight_{i}^{2}}\\
 & =\frac{\mSigma_{ij}}{w_{i}w_{j}}-\alpha\frac{\mLambda_{ij}}{w_{i}w_{j}}+\frac{\beta\iMatrix_{i=j}}{w_{i}w_{j}}\enspace.\tag{Using that \ensuremath{\mSigma}is diagonal}
\end{align*}
Consequently $\hessian_{\vw\vw}\penalizedObjectiveWeight(\vSlack,\vWeight)=\mw^{-1}\left(\mSigma+\beta\iMatrix-\alpha\mLambda\right)\mw^{-1}$
as desired. 

\end{proof}

Using this lemma, we prove that $\vg$ is a weight function with good
size.

\begin{lemma}\label{lem:weights_full:existence_and_size} The function
$\vg$ is a weight function meeting the criterion of Definition~\ref{def:sec_weighted_path:weight_function}.
For all $\vSlack\in\dSlack$ and $i\in[m]$ we have
\[
\beta\leq g_{i}(\vs)\leq1+\beta\enspace\text{ and }\enspace\norm{\vg(\vSlack)}_{1}=\rank(\ma)+\beta\cdot m.
\]
Furthermore, for all $\vSlack\in\dSlack,$ the weight function obeys
the following equations
\[
\mg(\vSlack)=\left(\mSigma_{g}+\beta\iMatrix\right)\onesVec\enspace\text{ and }\enspace\mg'(\vSlack)=-2\mg(\vSlack)\left(\mg(\vSlack)-\alpha\mLapProj_{g}\right)^{-1}\mLambda_{g}\mSlack^{-1}
\]
where $\mSigma_{g}\defeq\mSigma_{\ma_{s}}(\mg(\vSlack)^{\alpha}\onesVec)$,
$\mLapProj_{g}\defeq\mLapProj_{\ma_{s}}(\mg(\vSlack)^{\alpha}\onesVec)$,
and $\mg'(\vs)$ is the Jacobian matrix of $\vg$ at $\vs$.

\end{lemma}

\begin{proof}

By Lemma~\ref{lem:tool:projection_matrices} and \eqref{eq:weights:constants_assumptions}
we have that for all $\vWeight,\vSlack\in\dSlack$ ,
\[
\mSigma_{\ma_{s}}(\vWeight)\specGeq\mLambda_{\ma_{s}}(\vWeight)\specGeq\alpha\mLambda_{\ma_{s}}(\vWeight).
\]
Therefore, by Lemma~\ref{lem:weights_full:derivatives_of_objective},
$\hessian_{\vw\vw}\penalizedObjectiveWeight(\vSlack,\vWeight)\specGeq\beta\mWeight^{-2}$
and $\penalizedObjectiveWeight(\vSlack,\vWeight)$ is convex for $\vWeight,\vSlack\in\dSlack$.
Using Lemma~\ref{lem:weights_full:derivatives_of_objective}, we
see that that for all $i\in[m]$ it is the case that
\[
\left[\gradient_{\vw}\penalizedObjectiveWeight(\vSlack,\vWeight)\right]_{i}=\frac{1}{w_{i}}\left(w_{i}-\sigma_{i}-\beta\right)
\]
Since $0\leq\sigma_{i}\leq1$ for all $i$ by Lemma~\ref{lem:tool:projection_matrices}
and $\beta\in(0,1)$ by \eqref{eq:weights:constants_assumptions},
we see that if $\vWeight_{i}\in(0,\beta)$ then $\left[\gradient_{\vw}\penalizedObjectiveWeight(\vSlack,\vWeight)\right]_{i}$
is strictly negative and if $\vWeight_{i}\in(1+\beta,\infty)$ then
$\left[\gradient_{\vw}\penalizedObjectiveWeight(\vSlack,\vWeight)\right]_{i}$
is strictly positive. Therefore, for any $\vSlack\in\dWeights$ ,
the $\vw$ that minimizes this convex function $\penalizedObjectiveWeight(\vSlack,\vWeight)$
lies in the box between $\beta$ and $1+\beta$. Since $\penalizedObjectiveWeight$
is strongly convex in this region, the minimizer is unique.

The formula for $\mg(\vSlack)$ follows by setting $\grad_{\vw}\penalizedObjectiveWeight(\vSlack,\vWeight)=\vzero$
and the size of $\vg$ follows from the fact that $\norm{\vsigma}_{1}=\mathrm{tr}\left(\mProj_{\ma_{s}}(\mg(\vSlack)^{\alpha})\right)$.
Since $\mProj_{\ma_{s}}(\mg(\vSlack)^{\alpha}\onesVec)$ is a projection
onto the image of $\mg(\vSlack)^{\alpha/2}\ma_{s}$ and since $\vg(\vSlack)>0$
and $\vSlack>0$, the dimension of the image of $\mg(\vSlack)^{\alpha/2}\ma_{s}$
is the rank of $\ma$. Hence, we have that
\[
\normOne{\fvWeightFull(\vSlack)}\leq\normOne{\vsigma}+\normOne{\beta\onesVec}=\rank(\ma)+\beta\cdot m.
\]

To compute $\mg'(\vSlack)$, we note that for $\vWeight\in\dWeights$
and $\mLambda_{w}\defeq\mLambda_{\mWeight^{\alpha}\ma}(\mSlack^{-2}\onesVec)$,
by Lemma~\ref{lem:deriv:lever} and chain rule, we get the following
for all $i,j\in[m]$,
\begin{align*}
\frac{\partial(\grad_{\vWeight}\penalizedObjectiveWeight(\vSlack,\vWeight))_{i}}{\partial s_{j}} & =-w_{i}^{-1}\mLambda_{ij}s_{j}^{2}\left(-2s_{j}^{-3}\right)=2w_{i}^{-1}\mLambda_{ij}s_{j}^{-1}\enspace.
\end{align*}
Consequently, $\jacobian_{\vSlack}(\grad_{\vWeight}\penalizedObjectiveWeight(\vSlack,\vWeight))=2\mw^{-1}\mLambda_{w}\mSlack^{-1}$
where $\jacobian_{\vSlack}$ denotes the Jacobian matrix of the function
$\grad_{\vw}\penalizedObjectiveWeight(\vs,\vw)$ with respect to $\vs$.
Since we have already shown that $\jacobian_{\vWeight}(\grad_{\vWeight}\penalizedObjectiveWeight(\vSlack,\vWeight))=\hessian_{\vw\vw}f_{t}(\vSlack,\vWeight)=\mw^{-1}\left(\mSigma_{w}+\beta\iMatrix-\alpha\mLambda_{w}\right)\mw^{-1}$
is positive definite (and hence invertible), by applying the implicit
function theorem to the specification of $\vg(\vSlack)$ as the solution
to $\grad_{\vw}\penalizedObjectiveWeight(\vSlack,\vWeight)=\vzero$,
we have
\[
\mg'(\vSlack)=-\left(\jacobian_{\vWeight}(\grad_{\vw}\penalizedObjectiveWeight(\vSlack,\vWeight))\right)^{-1}\left(\jacobian_{\vSlack}(\grad_{\vw}\penalizedObjectiveWeight(\vSlack,\vWeight))\right)=-2\mg(\vSlack)\left(\mg(\vSlack)-\alpha\mLambda_{g}\right)^{-1}\mLambda_{g}\mSlack^{-1}
\]
\end{proof}

Using Lemma~\ref{lem:weights_full:existence_and_size} we now show
that $\vg$ has a good slack sensitivity.

\begin{lemma}[Weight Function Slack Sensitivity]\label{lem:weights_full:cond_number}
For all $\vSlack\in\dSlack$, we have $\energyStab(\vSlack,\vg(\vSlack))\leq2$.
\end{lemma}

\begin{proof} Fix an arbitrary $\vSlack\in\dSlack$ and let $\fvWeight\defeq\fvWeightFull(\vSlack)$,
and $\mSigma\defeq\mSigma_{\ma_{s}}(\vg^{\alpha})$. Recall that by
Lemma~\ref{lem:weights_full:existence_and_size} we know that $\fvWeight=\left(\mSigma+\beta\iMatrix\right)\onesVec$
and $\beta\leq g_{i}\leq1+\beta\leq2$ for all $i\in[m]$. Furthermore,
since $\beta^{1-\alpha}\geq\frac{1}{2}$ and $\alpha\in(0,1)$ by
\eqref{eq:sec:weights_full:weight_function} and clearly $\fmWeight=\fmWeight^{1-\alpha}\fmWeight^{\alpha}$
we have
\begin{equation}
\frac{1}{2}\fmWeight^{\alpha}\specLeq\beta^{1-\alpha}\fmWeight^{\alpha}\specLeq\fmWeight\specLeq(2)^{1-\alpha}\fmWeight^{\alpha}\specLeq2\fmWeight^{\alpha}\label{eq:weights_full:cond1}
\end{equation}
Applying this and using the definition of $\mProj_{\ma_{s}}(\fvWeight)$
yields
\begin{equation}
\ma_{s}(\ma_{s}^{T}\fmWeight\ma_{s})^{-1}\ma_{s}^{T}\specLeq2\ma_{s}(\ma_{s}^{T}\fmWeight^{\alpha}\ma_{s})^{-1}\ma_{s}^{T}=2\fmWeight^{-\alpha/2}\mProj_{\ma_{s}}(\vg^{\alpha})\fmWeight^{-\alpha/2}\enspace.\label{eq:weights_full:cond2}
\end{equation}
Hence, by definition of the weight slack sensitivity we have
\begin{eqnarray*}
\energyStab(\vSlack,\vg) & = & \max_{i}\norm{\fmWeight^{-1/2}\indicVec i}_{\mProj_{\ma_{s}}(\fvWeight)}\\
 & = & \max_{i}\sqrt{\indicVec i^{T}\ma_{s}(\ma_{s}^{T}\fmWeight\ma_{s})^{-1}\ma_{s}^{T}\indicVec i}\\
 & \leq & \max_{i}\sqrt{2\indicVec i^{T}\fmWeight^{-\alpha/2}\mProj_{\ma_{s}}(\fvWeight^{\alpha})\fmWeight^{-\alpha/2}\indicVec i}\\
 & = & \max_{i}\sqrt{2\frac{\sigma_{i}}{g_{i}^{\alpha}}}\leq2\max_{i}\sqrt{\frac{\sigma_{i}}{g_{i}}}\leq2
\end{eqnarray*}
where the last line due to the fact $g_{i}^{1-\alpha}\geq\beta^{1-\alpha}\geq\frac{1}{2}$
and $g_{i}\geq\sigma_{i}$. \end{proof}

Finally, we bound the step consistency of $\fvWeightFull$.

\begin{lemma}[Weight Function Step Consistency]\label{lem:weights_full:consistency}
For all $\vSlack\in\rPos^{m}$, $\vy\in\Rm$, $r\geq\frac{2}{1-\alpha}$,
and
\[
\mb\defeq\iMatrix+\frac{1}{r}\fmWeightFull(\vSlack)^{-1}\fmWeightFull'(\vSlack)\mSlack,
\]
we have
\[
\normFull{\mb\vy}_{\fmWeightFull(\vSlack)}\leq\norm{\vy}_{\fmWeightFull(\vSlack)}\enspace\text{ and }\enspace\normFullInf{\mb\vy}\leq\normInf{\vy}+\frac{2}{1-\alpha}\norm{\vy}_{\fmWeightFull(\vSlack)}.
\]
\end{lemma}

\begin{proof} 

Fix an arbitrary $\vSlack\in\rPos^{m}$ and let $\fvWeight\defeq\fvWeightFull(\vSlack)$,
$\vsigma\defeq\vLever_{\ma_{s}}(\vg^{\alpha})$, $\mSigma\defeq\mSigma_{\ma_{s}}(\vg^{\alpha})$,
$\mProj\defeq\mProj_{\ma_{s}}(\vg^{\alpha})$, $\mLapProj\defeq\mLapProj_{\ma_{s}}(\vg^{\alpha})$.
Also, fix an arbitrary $\vy\in\Rm$ and let $\vz\defeq\mb\vy$.

By Lemma \ref{lem:weights_full:existence_and_size}, $\fmWeight'=-2\fmWeight\left(\fmWeight-\alpha\mLambda\right)^{-1}\mLambda\mSlack^{-1}$
and therefore
\begin{align*}
\mb & =\iMatrix+r^{-1}\mg^{-1}\left(-2\fmWeight\left(\fmWeight-\alpha\mLambda\right)^{-1}\mLambda\mSlack^{-1}\right)\mSlack\\
 & =\left(\fmWeight-\alpha\mLambda\right)^{-1}\left(\fmWeight-\alpha\mLambda\right)-2r^{-1}\left(\fmWeight-\alpha\mLambda\right)^{-1}\mLambda\\
 & =\left(\fmWeight-\alpha\mLambda\right)^{-1}\left(\mg-(\alpha+2r^{-1})\mLambda\right)\enspace.
\end{align*}
By Lemma~\ref{lem:weights_full:existence_and_size}, we have $\mg\succeq\mSigma.$
By the definition of $\mLambda=\mSigma-\mProj^{(2)}$, we have $\mSigma\succeq\mLambda$
and Lemma~\ref{lem:tool:projection_matrices} shows that $\mLambda\succeq\mZero$.
Hence, we have
\[
\mZero\specLeq\mLambda\specLeq\mSigma\specLt\mg.
\]
Using this and $0<2r^{-1}\leq1-\alpha$, we have that
\[
\mZero\specLt\mg-(\alpha+2r^{-1})\mLambda\specLeq\fmWeight-\alpha\mLambda\enspace.
\]
Thus, $\fmWeight-\alpha\mLambda$ is positive definite and therefore
$\vz$ is the unique vector such that
\begin{equation}
\left(\fmWeight-\alpha\mLambda\right)\vz=\left(\mg-(\alpha+2r^{-1})\mLambda\right)\vy\label{eq:weights_full:equality}
\end{equation}

To bound $\norm{\vz}_{\mg}$, we note that since $\mg\specGt\mZero$
we have
\[
\left(\iMatrix-\alpha\fmWeight^{-1/2}\mLambda\fmWeight^{-1/2}\right)\fmWeight^{1/2}\vz=\left(\iMatrix-(\alpha+2r^{-1})\fmWeight^{-1/2}\mLambda\fmWeight^{-1/2}\right)\fmWeight^{1/2}\vy
\]
Furthermore, since $\mZero\specLeq\fmWeight^{-1/2}\mLambda\fmWeight^{-1/2}\specLeq\iMatrix$,
we have that
\[
\mZero\specLeq\iMatrix-(\alpha+2r^{-1})\fmWeight^{-1/2}\mLambda\fmWeight^{-1/2}\specLeq\iMatrix-\alpha\fmWeight^{-1/2}\mLambda\fmWeight^{-1/2}
\]
 and consequently
\begin{align}
\norm{\vz}_{\mg} & =\norm{\left(\iMatrix-\alpha\fmWeight^{-1/2}\mLambda\fmWeight^{-1/2}\right)\fmWeight^{1/2}\vz}_{\left(\iMatrix-\alpha\fmWeight^{-1/2}\mLambda\fmWeight^{-1/2}\right)^{-2}}\nonumber \\
 & \leq\norm{\left(\iMatrix-\alpha\fmWeight^{-1/2}\mLambda\fmWeight^{-1/2}\right)\fmWeight^{1/2}\vz}_{\left(\iMatrix-(\alpha+2k^{-1})\fmWeight^{-1/2}\mLambda\fmWeight^{-1/2}\right)^{-2}}\nonumber \\
 & =\norm{\mg^{1/2}\vy}=\norm{\vy}_{\mg}\enspace.\label{eq:zG_yG}
\end{align}
Therefore, $\norm{\mb\vy}_{\mg}\leq\norm{\vy}_{\mg}$ as desired.

Next, to bound $\norm{\vz}_{\infty}$, we use that $\mLapProj=\mLever-\shurSquared{\mProj}$
and $\vg=\vsigma+\beta\onesVec$ and \eqref{eq:weights_full:equality}
to derive
\begin{eqnarray*}
\left(1-\alpha\right)\mSigma\vz+\beta\vz+\alpha\mProj^{(2)}\vz & = & \mg\vz-\alpha\mLambda\vz\\
 & = & \left(\mg-(\alpha+2r^{-1})\mLambda\right)\vy\\
 & = & \left(1-\alpha-2r^{-1}\right)\mSigma\vy+\beta\vy+\left(\alpha+2r^{-1}\right)\mProj^{(2)}\vy\enspace.
\end{eqnarray*}
toLeft multiplying this equation by $\indicVec i^{T}$ for arbitrary
$i\in[m]$ and using that $\vsigma_{i}\geq0$ then yields that
\begin{align*}
\left((1-\alpha)\vsigma_{i}+\beta\right)\left|\vz_{i}\right| & \leq\left|\alpha\indicVec i^{T}\mProj^{(2)}\vz\right|+\left|\left((1-\alpha-2r^{-1})\vsigma_{i}+\beta\right)\vy_{i}+\left(\alpha+2r^{-1}\right)\indicVec i^{T}\mProj^{(2)}\vy\right|\\
 & \leq\alpha\left|[\mProj^{(2)}\vz]_{i}\right|+\left((1-\alpha)\vsigma_{i}+\beta\right)\normInf{\vy}+\left|[\mProj^{(2)}\vy]_{i}\right|\tag{\ensuremath{0<2r^{-1}\leq(1-\alpha)<1}}\\
 & \leq\alpha\vsigma_{i}\norm{\vz}_{\mSigma}+\left((1-\alpha)\vsigma_{i}+\beta\right)\normInf{\vy}+\vsigma_{i}\norm{\vy}_{\mSigma}\tag{Lemma\,\ref{lem:tool:projection_matrices}}\\
 & \leq\left((1-\alpha)\vsigma_{i}+\beta\right)\normInf{\vy}+(1+\alpha)\vsigma_{i}\norm{\vy}_{\mg}\tag{\ensuremath{\mSigma\specLeq\mg}and }\ref{eq:zG_yG}
\end{align*}
Consequently,
\begin{eqnarray*}
|\vz_{i}| & \leq & \normInf{\vy}+\frac{\left(1+\alpha\right)\vsigma_{i}}{\left((1-\alpha)\vsigma_{i}+\beta\right)}\norm{\vy}_{\mg}\\
 & \leq & \normInf{\vy}+\frac{2}{1-\alpha}\norm{\vy}_{\mg}
\end{eqnarray*}
and therefore $\normInf{\mb\vy}=\normInf{\vz}\leq\normInf{\vy}+\frac{2}{1-\alpha}\norm{\vy}_{\fmWeight}$.

\end{proof}

From Lemmas \ref{lem:weights_full:existence_and_size}, \ref{lem:weights_full:cond_number}
and \ref{lem:weights_full:consistency}, the proof of Theorem \ref{thm:weights_full:weight_properties}
is immediate. Since $m\geq\rank(\ma)$ we have $\log_{2}\left(2m/\rank(\ma)\right)\geq1$
and $\alpha\in(0,1)$. Furthermore $\beta\in(0,1)$ and
\[
\beta^{1-\alpha}=\left(\frac{\rank(A)}{2m}\right)^{\left(\frac{1}{\log_{2}\left(2m/\rank(\ma)\right)}\right)}=\frac{1}{2}
\]
and therefore \eqref{eq:weights:constants_assumptions} is satisfied.
Furthermore, for all $\vSlack\in\dSlack$ we have $\normOne{\fvWeightFull(\vSlack)}\leq2\cdot\rank(\ma)$
by Lemma~\ref{lem:weights_full:existence_and_size}. The bounds on
$\cWeightStab(\fvWeightFull)$ and $\cWeightCons(\fvWeightFull)$
then follow from Lemma~\ref{lem:weights_full:cond_number} and Lemma~\ref{lem:weights_full:consistency}
respectively.

\subsection{Computing and Correcting The Weights}

\label{sec:weights_full:computing} 

Here, we describe how to efficiently compute approximations to the
weight function $\vg:\rNonNeg^{m}\rightarrow\rNonNeg^{m}$ as given
by \eqref{eq:sec:weights_full:weight_function}. The two main technical
tools we use towards this end are the \emph{gradient descent method},
Theorem~\ref{thm:constrained_minimization}, a standard result in
convex optimization, and fast numerical methods for estimating leverage
scores using the Johnson-Lindenstrauss Lemma, Theorem~\ref{thm:weights_full:leverage-score_sampling},
a powerful tool in randomized numerical linear algebra.

Since the weight function, $\vg$, is defined as the minimizer of
a convex optimization problem \eqref{eq:sec:weights_full:weight_function},
we could use the gradient descent method directly to minimize $\penalizedObjectiveWeight$
and hence compute $\vg$. Indeed, in Lemma~\ref{lem:weight_iterative}
we show how applying the gradient descent method in a carefully scaled
space allows us to compute $\vg(\vs)$ to high accuracy in $\otilde(1)$
iterations. Unfortunately, this result makes two assumptions to compute
$\vg(\vs)$: (1) we are given a weight $\vw\in\rNonNeg$ that is not
too far from $\vg(\vs)$ and (2) we compute the gradient of $\penalizedObjectiveWeight$
exactly. 

Assumption (1) is not an issue as we always ensure that $\vg$ does
not change too much between calls to compute $\vg$ and therefore
can always use our previous weights as the approximation to $\vg(\vs)$.
However, naively computing the gradient of $\penalizedObjectiveWeight$
is computationally expensive and hence assumption (2) is problematic.
To deal with this issue we use the fact that by careful application
of Johnson-Lindenstrauss one can compute a multiplicative approximation
to the gradient efficiently and in Theorem \ref{thm:weights_full:approximate_weight}
we show that this suffices to compute an approximation to $\vg$ that
suffices to use in our weighted path following scheme.

First we prove the theorem regarding gradient descent method we use
in our analysis. This theorem shows that if we take repeated projected
gradient steps then we can achieve linear convergence up to bounds
on how much the hessian of the function changes over the domain of
interest. %
\footnote{Note that this theorem is usually stated with $\mh=\iMatrix$, i.e.
the standard Euclidean norm rather than the one induced by $\mh$.
However, Theorem~\ref{thm:constrained_minimization} can be proved
by these standard results just by a change of variables.%
}

\begin{theorem}[Simple Constrained Minimization for Twice Differentiable Function \cite{Nesterov2003}]
\label{thm:constrained_minimization} Let $\mh$ be a positive definite
matrix and $Q\subseteq\Rm$ be a convex set. Let $f(\vx):Q\rightarrow\R$
be a twice differentiable function such that there are constants $L\geq\mu\geq0$
such that for all $\vx\in Q$ we have $\mu\mh\preceq\nabla^{2}f(\vx)\specLeq L\mh$.
If for some $\vx^{(0)}\in Q$ and all $k\geq0$ we apply the update
rule
\[
\vx^{(k+1)}=\argmin_{\vx\in Q}\innerProduct{\grad f(\vx^{(k)})}{\vx-\vx^{(k)}}+\frac{L}{2}\norm{\vx-\vx^{(k)}}_{\mh}^{2}
\]
then for all $k\geq0$ we have
\[
\norm{\vx^{(k)}-\vx^{*}}_{\mh}^{2}\leq\left(1-\frac{\mu}{L}\right)^{k}\norm{\vx^{(0)}-\vx^{*}}_{\mh}^{2}.
\]
\end{theorem} 

To apply this Theorem~\ref{thm:constrained_minimization} to compute
$\vg(\vs)$ we first need to show that there is a region around the
optimal point $\vg(\vs)$ such that the Hessian of $\hat{f}$ does
not change too much. 

\begin{lemma}[Hessian Approximation]\label{lem:Hessian_of_weight}
For $\normInf{\mw^{-1}(\fvWeightFull(\vSlack)-\vw)}\leq\frac{1}{12}$
we have
\[
\frac{2(1-\alpha)}{3}\mw^{-1}\specLeq\hessWW\hat{f}(\vSlack,\vWeight)\specLeq\frac{3}{2}\mw^{-1}.
\]
\end{lemma}

\begin{proof} From Lemma \ref{lem:weights_full:derivatives_of_objective},
we know that
\[
\hessian_{\vw\vw}\penalizedObjectiveWeight(\vSlack,\vWeight)=\mw^{-1}\left(\mSigma+\beta\iMatrix-\alpha\mLambda\right)\mw^{-1}
\]
where $\mSigma=\mSigma_{\ma_{s}}(\vWeight^{\alpha})$ and $\mLapProj=\mLapProj_{\ma_{s}}(\vWeight^{\alpha})$.
Using $\mZero\specLeq\mLambda\specLeq\mSigma$, we have
\[
(1-\alpha)\mw^{-1}\left(\mSigma+\beta\iMatrix\right)\mw^{-1}\specLeq\hessian_{ww}\penalizedObjectiveWeight(\vSlack,\vWeight)\specLeq\mw^{-1}\left(\mSigma+\beta\iMatrix\right)\mw^{-1}
\]
Using that $\normInf{\mw^{-1}(\fvWeightFull(\vSlack)-\vw)}\leq\frac{1}{12}$
and applying Lemma \ref{lem:weights_full:existence_and_size} we have
\[
\mSigma+\beta\iMatrix\specLeq\left(1-\frac{1}{12}\right)^{-2}\mSigma_{\ma_{s}}(\vg^{\alpha})+\beta\iMatrix\specLeq\left(1-\frac{1}{12}\right)^{-2}\fmWeight\specLeq\frac{3}{2}\mw
\]
and
\[
\mSigma+\beta\iMatrix\succeq\left(1-\frac{1}{12}\right)^{2}\mSigma_{\ma_{s}}(\vg^{\alpha})+\beta\iMatrix\succeq\left(1-\frac{1}{12}\right)^{2}\fmWeight\succeq\frac{2}{3}\mw.
\]
\end{proof}

Combining Theorem \ref{thm:constrained_minimization} and Lemma \ref{lem:Hessian_of_weight},
we get the following algorithm to compute the weight function using
the exact computation of the gradient of $\hat{f}$. Note that this
algorithm applies Theorem \ref{thm:constrained_minimization} multiple
times as in each iteration we are taking a gradient step with respect
to a different norm.

\begin{lemma}[Exact Weight Computation]\label{lem:weight_iterative}
Given $\vw^{(0)}\in\dWeights$ such that $\normFull{\mWeight_{(0)}^{-1}(\vg(\vs)-\vWeight^{(0)})}_{\infty}\leq\frac{1-\alpha}{24}$.
Let
\[
Q=\left\{ \vWeight\in\Rm~|~\normFull{\mWeight_{(0)}^{-1}(\vw-\vWeight^{(0)})}_{\infty}\leq\frac{1-\alpha}{24}\right\} .
\]
For all $k\geq0$ let
\[
\vWeight^{(k+1)}=\argmin_{\vWeight\in Q}\normFull{\vWeight-\frac{1}{2}\left(\vWeight^{(k)}+\vsigma_{\ma_{s}}\left(\left(\vWeight^{(k)}\right)^{\alpha}\right)+\beta\right)}_{\mw_{(k)}^{-1}}^{2}
\]
This implies that for all $k$ ,
\[
\normFull{\fmWeightFull(\vSlack)^{-1}(\vg(\vSlack)-\vWeight^{(k)})}_{\infty}^{2}\leq4m^{2}\left(1-\frac{1-\alpha}{12}\right)^{k}.
\]
\end{lemma}

\begin{proof} Note that iterations of Theorem~\ref{thm:constrained_minimization}
can be rewritten as
\begin{eqnarray*}
\vWeight^{(k+1)} & = & \argmin_{\vWeight\in Q}\innerProduct{\left(\iMatrix-\mSigma_{\ma_{s}}\left(\left(\vWeight^{(k)}\right)^{\alpha}\right)\mWeight_{(k)}^{-1}-\beta\mw_{(k)}^{-1}\right)\onesVec}{\vWeight-\vWeight^{(k)}}+\normFull{\vWeight-\vWeight^{(k)}}_{\mw_{(k)}^{-1}}^{2}\\
 & = & \argmin_{\vWeight\in Q}\normFull{\vWeight-\frac{1}{2}\left(\vWeight^{(k)}+\vsigma_{\ma_{s}}\left(\left(\vWeight^{(k)}\right)^{\alpha}\right)+\beta\right)}_{\mw_{(k)}^{-1}}^{2}
\end{eqnarray*}
where the last line simply comes for expanding the quadratic function
and ignoring the constant term. Hence, we see that the iteration on
$\vWeight^{(k+1)}$ is in fact a gradient descent step. To apply Theorem
\ref{thm:constrained_minimization} we note that for any $\vWeight\in Q$
the definition of $Q$ and the fact that $\alpha\in(0,1)$ implies
that $(1-\frac{1}{24})\mWeight_{(0)}\specLeq\mWeight\specLeq(1+\frac{1}{24})\mWeight_{(0)}$.
Therefore Lemma~\ref{lem:Hessian_of_weight} shows that for all $\vWeight^{(k)}\in Q$,
\begin{equation}
\frac{1-\alpha}{2}\mWeight_{(k)}^{-1}\preceq\frac{2(1-\alpha)}{3}\mWeight_{(0)}^{-1}\preceq\hessian_{\vw\vw}\penalizedObjectiveWeight(\vSlack,\vWeight)\preceq\frac{3}{2}\mWeight_{(0)}^{-1}\preceq2\mWeight_{(k)}^{-1}.\label{eq:lip_grad_vw}
\end{equation}
where the left most and right most inequality comes from the fact
they lies in a small region $Q$. Hence, Theorem \ref{thm:constrained_minimization}
and inequality \eqref{eq:lip_grad_vw} shows that
\[
\normFull{\vWeight^{(k+1)}-\fvWeight(\vs)}_{\mw_{(k)}^{-1}}^{2}\leq\left(1-\frac{1-\alpha}{4}\right)\normFull{\vWeight^{(k)}-\fvWeight(\vs)}_{\mw_{(k)}^{-1}}^{2}.
\]
Since $\normFull{\mWeight_{(0)}^{-1}(\vg(\vs)-\vWeight^{(0)})}_{\infty}\leq\frac{1-\alpha}{24}$
and $\vw^{(k)}\in Q$, we know that $\mg(\vs)\specGeq\left(1-\frac{1-\alpha}{24}\right)^{2}\mw_{(k)}$.
Hence, we have
\begin{eqnarray*}
\norm{\vw^{(k)}-\fvWeight(\vs)}_{\fmWeight^{-1}(\vs)}^{2} & \leq & \left(1-\frac{1-\alpha}{24}\right)^{-2}\left(1-\frac{1-\alpha}{4}\right)\norm{\vw^{(k-1)}-\fvWeight(\vs)}_{\fmWeight^{-1}(\vs)}^{2}\\
 & \leq & \left(1-\frac{1-\alpha}{12}\right)\norm{\vw^{(k-1)}-\fvWeight(\vs)}_{\fmWeight^{-1}(\vs)}^{2}\\
 & \leq & \left(1-\frac{1-\alpha}{12}\right)^{k}\norm{\vw^{(0)}-\fvWeight(\vs)}_{\fmWeight^{-1}(\vs)}^{2}
\end{eqnarray*}
The result follows from the facts that
\[
\norm{\vw^{(0)}-\fvWeight(\vs)}_{\fmWeight^{-1}(\vs)}^{2}\leq m\norm{\fmWeightFull(\vs)}_{\infty}\norm{\fmWeightFull^{-1}(\vs)(\vg(\vs)-\vWeight^{(0)})}_{\infty}^{2}\leq\frac{m(1+\beta)}{\left(1-\frac{1-\alpha}{24}\right)^{2}}\norm{\mw_{(0)}^{-1}(\vg(\vs)-\vWeight^{(0)})}_{\infty}^{2}
\]
and lemma \ref{lem:weights_full:existence_and_size} that $\norm{\fmWeightFull^{-1}(\vs)(\vWeight^{(k)}-\vg(\vs))}_{\infty}^{2}\leq\beta^{-1}\norm{\vw^{(k)}-\fvWeight(\vs)}_{\fmWeight^{-1}(\vs)}^{2}$
where $\beta=\frac{\rank(\ma)}{2m}$.

\end{proof}

Unfortunately, we cannot use the previous lemma directly as computing
$\vsigma_{\ma_{s}}$ exactly is too expensive for our purposes. However,
in \cite{spielmanS08sparsRes} they showed that we can compute leverage
scores, $\vsigma_{\ma_{s}}$, approximately by solving only polylogarithmically
many regression problems (See \cite{mahoney11survey} for more details).
These results use the fact that the leverage scores of the the $i^{th}$
constraint, i.e. $[\vsigma_{\ma_{s}}]_{i}$ is the $\ellTwo$ length
of vector $\mProj_{\ma}(\vx)\indicVec i$ and that by the Johnson-Lindenstrauss
lemma these lengths are persevered up to multiplicative error if we
project these vectors onto certain random low dimensional subspace.
Consequently, to approximate the $\vsigma_{\ma_{s}}$ we first compute
the projected vectors and then use it to approximate $\vsigma_{\ma_{s}}$
and hence only need to solve $\otilde(1)$ regression problems. For
completeness, we provide the algorithm and theorem here:

\begin{center}
\begin{tabular}{|l|}
\hline 
\textbf{$\vsigma^{(apx)}=\code{computeLeverageScores}(\ma,\vx,\epsilon)$}\tabularnewline
\hline 
\hline 
1. Let $k=\ceil{24\log(m)/\epsilon^{2}}$. \tabularnewline
\hline 
2. Let $\vec{q}^{(j)}$ be $k$ random $\pm1/\sqrt{k}$ vectors of
length $m$.\tabularnewline
\hline 
3. Compute $\vp^{(j)}=\mx^{1/2}\ma(\ma^{T}\mx\ma)^{-1}\ma^{T}\mx^{1/2}\vec{q}^{(j)}$.\tabularnewline
\hline 
4. Return $\vsigma_{i}^{(apx)}=\sum_{j=1}^{k}\left(\vp_{i}^{(j)}\right)^{2}$.\tabularnewline
\hline 
\end{tabular}
\par\end{center}

\begin{theorem}[\cite{spielmanS08sparsRes}]\label{thm:weights_full:leverage-score_sampling}
For $0<\epsilon<1$ with probability at least $1-\frac{1}{m}$ the
algorithm\textbf{ $\code{computeLeverageScores}$} returns $\vsigma^{(apx)}$
such that for all $i\in[m]$ ,
\[
\left(1-\epsilon\right)\vsigma_{\ma}(\vx)_{i}\leq\vsigma_{i}^{(apx)}\leq\left(1+\epsilon\right)\vsigma_{\ma}(\vx)_{i}.
\]
by solving only $O(\epsilon^{-2}\cdot\log m)$ linear systems.\end{theorem}

Now, we show that we can modify Lemma \ref{lem:weight_iterative}
to use $\code{computeLeverageScores}$ and we prove that this still
provides adequate error guarantees. Our weight computation and the
analysis is as follows.

\begin{center}
\begin{tabular}{|l|}
\hline 
\textbf{$\ensuremath{\vWeight=\code{computeWeight}(\vs,\vw^{(0)},K)}$}\tabularnewline
\hline 
\hline 
1. Let $c_{r}=2\log_{2}\left(\frac{2m}{\rank(\ma)}\right)$, $\alpha=1-\frac{1}{\log_{2}\left(\frac{2m}{\rank(\ma)}\right)}$,
$\beta=\frac{\rank(\ma)}{2m},$ $\epsilon=\frac{K}{48c_{r}\log\left(\frac{2m}{K}\right)}$.\tabularnewline
\hline 
2. $Q=\left\{ \vWeight\in\Rm~|~\normFull{\mWeight_{(0)}^{-1}(\vw-\vWeight^{(0)})}_{\infty}\leq\frac{1}{12c_{r}}\right\} $\tabularnewline
\hline 
3. For $j=1$ to $k$ where $k=\lceil12c_{r}\log\left(\frac{4m}{K}\right)\rceil$\tabularnewline
\hline 
3a. \textbf{$\ $$\vsigma^{(j)}=\code{computeLeverageScores}\left(\mSlack^{-1}\ma,\left(\vWeight^{(j)}\right)^{\alpha},\epsilon\right)$}\tabularnewline
\hline 
3b. \textbf{$\ $}$\vw^{(j)}=\argmin_{\vWeight\in Q}\normFull{\vWeight-\frac{1}{2}\left(\vWeight^{(j-1)}+\vsigma^{(j)}+\beta\onesVec\right)}_{\mw_{(j-1)}^{-1}}^{2}$\tabularnewline
\hline 
4. Output $\vWeight^{(j)}$.\tabularnewline
\hline 
\end{tabular}
\par\end{center}

Note that the convex set $Q$ is aligned with standard basis and hence
the step 3b can be computed by explicit formula \eqref{eq:explicit_formula_w}.

\begin{theorem}[Approximate Weight Computation]\label{thm:weights_full:approximate_weight}
Let $\vs\in\dSlack$, $\normInf{\mw_{(0)}^{-1}(\vg(\vSlack)-\vWeight^{(0)})}\leq\frac{1}{12c_{r}}$%
\footnote{Recall that $c_{r}=\frac{2}{1-\alpha}=2\log\left(\frac{2m}{\rank\left(\ma\right)}\right)\geq2$.%
}, and $K\in(0,1)$. The algorithm\textbf{ $\code{computeWeight}(\vs,\vw^{(0)},K)$}
returns $\vWeight$ such that
\[
\normInf{\fmWeightFull(\vSlack)^{-1}(\vg(\vSlack)-\vWeight)}\leq K
\]
with probability $\left(1-\frac{1}{m}\right)^{\lceil12c_{r}\log\left(\frac{4m}{K}\right)\rceil}$. 

The running time is dominated by the time needed to solve $O(c_{r}^{3}\log^{3}(m/K)\log(m)/K^{2})$
linear systems.\end{theorem}

\begin{proof} Consider an execution of \textbf{$\code{computeWeight}(\vs,\vw^{(0)},K$)}
where each $\code{computeLeverageScores}$ computes $\vsigma_{\ma_{s}}\left((\vw^{(j)})^{\alpha}\right)$
exactly, i.e. $\vsigma^{(j)}=\vsigma_{\ma_{s}}\left((\vw^{(j)})^{\alpha}\right)$,
and let $\vv^{(j)}$denote the $\vw^{(j)}$ computed during this idealized
execution of $\code{computeWeight}$. 

Now suppose that for all $i\in[m]$ we have
\begin{equation}
(1-\epsilon)^{M}\vv_{i}^{(j)}\leq\vWeight_{i}^{(j)}\leq(1+\epsilon)^{M}\vv_{i}^{(j)}\label{eq:weights_full:apx_comp_1}
\end{equation}
 for some $M\geq0$ and $j\in[k-1]$. Since the objective function
and the constraints for step 3b. are axis-aligned we can compute $\vw^{(j)}$
coordinate-wise and we see that
\begin{equation}
\vWeight^{(j+1)}=\code{median}\left(\left(1-\frac{1}{12c_{r}}\right)\vWeight^{(0)},\vWeight^{(j)}+\frac{1}{2}\left(\vsigma_{\ma_{s}}\left(\left(\vWeight^{(j)}\right)^{\alpha}\right)+\beta\right),\left(1+\frac{1}{12c_{r}}\right)\vWeight^{(0)}\right)\label{eq:explicit_formula_w}
\end{equation}
where $[\code{median}\left(\vx,\vy,\vz\right)]_{i}$ is equal to the
median of $x_{i}$, $y_{i}$ and $z_{i}$ for all $i\in[m]$. By \eqref{eq:weights_full:apx_comp_1},
\eqref{eq:explicit_formula_w}, and the fact that $\left(1-\epsilon\right)\sigma_{\ma_{s}}\left(\left(\vw^{(j+1)}\right)^{\alpha}\right)_{i}\leq\vsigma_{i}^{(j+1)}\leq\left(1+\epsilon\right)\sigma_{\ma_{s}}\left(\left(\vw^{(j+1)}\right)^{\alpha}\right)_{i}$
for all $i\in[m]$, we have that
\[
(1-\epsilon)^{M+1}\vv_{i}^{(j+1)}\leq\vWeight_{i}^{(j+1)}\leq(1+\epsilon)^{M+1}\vv_{i}^{(j+1)}.
\]
Since $\vv^{(0)}=\vWeight^{(0)}$ and since $j\in[k-1]$ was arbitrary
we can apply induction and we have that for all $j\in[k]$
\[
(1-\epsilon)^{j}\vv_{i}^{(j)}\leq\vWeight_{i}^{(j)}\leq(1+\epsilon)^{j}\vv_{i}^{(j)}.
\]
Note that $k\epsilon\leq\frac{1}{8}$ and therefore by Taylor series
expansion we have $\norm{\mv_{(k)}^{-1}\left(\vw^{(k)}-\vv^{(k)}\right)}_{\infty}\leq\frac{9}{8}\epsilon k$.
Furthermore since $\vv^{(k)}\in Q$ we know that $\mg(\vs)\specGeq\left(1-\frac{1}{12\cWeightCons}\right)^{2}\mv_{(k)}$.
Putting these together, applying Lemma~\ref{lem:weight_iterative},
and recalling that $k=\lceil12c_{r}\log\left(\frac{4m}{K}\right)\rceil$
we have
\begin{align*}
\normInf{\fmWeightFull(\vSlack)^{-1}(\vg(\vSlack)-\vWeight^{(k)})} & \leq\normInf{\fmWeightFull(\vSlack)^{-1}(\vg(\vSlack)-\vv^{(k)})}+\normInf{\fmWeightFull(\vSlack)^{-1}\left(\vv^{(k)}-\vWeight^{(k)}\right)}\\
 & \leq2m\left(1-\frac{1}{6c_{r}}\right)^{\frac{k}{2}}+\left(1-\frac{1}{12\cWeightCons}\right)^{-2}\normInf{\mv_{(k)}^{-1}(\vv^{(k)}-\vWeight^{(k)})}\\
 & \leq2m\cdot\exp\left(-\frac{k}{12\cWeightCons}\right)+1.5k\epsilon\\
 & \leq\frac{K}{2}+1.5\epsilon\lceil12c_{r}\log\left(\frac{4m}{K}\right)\rceil\leq K
\end{align*}
\end{proof}

Finally, we show how to compute an initial weight without having an
approximate weight to help the computation. The algorithm $\code{computeInitialWeight}(\vs,K)$
computes an initial weight in $\tilde{O}\left(\sqrt{\rank\ma}\right)$
iterations of $\code{computeWeight}$ by computing $\vg$ for a large
enough value of $\beta$ and then decreasing $\beta$ gradually.

\begin{center}
\begin{tabular}{|l|}
\hline 
\textbf{$\ensuremath{\vWeight=\code{computeInitialWeight}(\vs,K)}$}\tabularnewline
\hline 
\hline 
1. Let $c_{r}=2\log_{2}\left(\frac{2m}{\rank(\ma)}\right)$, $\alpha=1-\frac{1}{\log_{2}\left(\frac{2m}{\rank(\ma)}\right)}$,
$\beta=12c_{r}$ and $\vWeight=\beta\onesVec$.\tabularnewline
\hline 
2. Loop until $\beta=\frac{\rank(\ma)}{2m}$\tabularnewline
\hline 
2a. \textbf{$\ $$\vWeight=\code{computeWeight}(\vs,\vWeight,\frac{1}{50c_{r}})$.}\tabularnewline
\hline 
2b. \textbf{$\ $}$\beta=\max\left\{ \left(1-\frac{(1-\alpha)^{3/2}}{1000c_{r}^{2}\sqrt{\rank\left(\ma\right)}}\right)\beta,\frac{\rank(\ma)}{2m}\right\} $.\tabularnewline
\hline 
3. Output $\code{computeWeight}(\vs,\vWeight,K)$.\tabularnewline
\hline 
\end{tabular}
\par\end{center}

\begin{theorem}[Computating Initial Weights]\label{thm:weights_full:initialapproximate_weight}
For $\vs\in\dSlack$ and $K>0$, with constant probability the algorithm
$\code{computeInitialWeight}(\vs,K)$ returns $\vw\in\dWeights$ such
that
\[
\normInf{\fmWeightFull(\vSlack)^{-1}(\vg(\vSlack)-\vWeight)}\leq K.
\]
The total running time of $\code{computeInitialWeight}(\vs,K)$ is
dominated by the time needed to solve $\tilde{O}(\sqrt{\rank\left(\ma\right)}\log(1/K)/K^{2})$
linear systems.\end{theorem}

\begin{proof} Fix $\vs\in\dSlack$ and let $\ma_{s}\defeq\ms^{-1}\ma$.
For all $\beta>0$ let $\vg:\rPos\rightarrow\Rm$ be defined by%
\footnote{Note that early we assumed that $\beta<1$ and here we use much larger
values of $\beta$. However, this bound on $\beta$ was primarily
to assist in bounding $c_{1}$ and does not affect this proof.%
} 
\[
\vg(\beta)\defeq\argmin_{\vWeight\in\rPos^{m}}\onesVec^{T}\vWeight-\frac{1}{\alpha}\log\det(\ma_{s}^{T}\mWeight^{\alpha}\ma_{s})-\beta\sum_{i\in[m]}\log\weight_{i}
\]
The algorithm $\code{computeInitialWeight}(\vs,K)$ maintains the
invariant that before step 2a
\begin{equation}
\normInf{\mw^{-1}(\vg(\beta)-\vWeight)}\leq\frac{1}{12c_{r}}.\label{eq:initialweightinvariant}
\end{equation}
Since $\vg(\beta)=\vsigma(\beta)+\beta$ where $\vsigma(\beta)\defeq\vLever_{\ma_{s}}(\vg^{\alpha}(\beta))$,
we have that for all $i\in[m]$
\[
\beta\leq g(\beta)_{i}\leq1+\beta.
\]
Therefore, in the step 1, the initial weight, $\vw=\beta\onesVec\in\dWeights$
satisfies the invariant \eqref{eq:initialweightinvariant}. After
step 2a, by Theorem~\ref{thm:weights_full:approximate_weight} we
have
\begin{equation}
\normInf{\mg(\beta)^{-1}(\vg(\beta)-\vWeight)}\leq\frac{1}{50c_{r}}.\label{eq:after_step3a}
\end{equation}
Therefore, it suffices to prove that $\vg(\beta)$ is close to $\vg(\beta-\theta)$
for small $\theta$. 

To bound how much $\vg(\beta)$ changes for small changes in $\beta$
we proceed similarly to Lemma~\ref{lem:weights_full:existence_and_size}.
First by the implicit function theorem and direct calculation we know
that
\begin{equation}
\frac{d\vg}{d\beta}=-\left(\jacobian_{\vWeight}(\grad_{\vw}\penalizedObjectiveWeight(\vSlack,\vWeight))\right)^{-1}\left(\jacobian_{\beta}(\grad_{\vw}\penalizedObjectiveWeight(\vSlack,\vWeight))\right)=\mg(\beta)\left(\mg(\beta)-\alpha\mLambda_{g}\right)^{-1}\onesVec\label{eq:weights_full:1}
\end{equation}
where $\mLapProj_{g}\defeq\mLapProj_{\ma_{s}}(\mg(\beta)^{\alpha}\onesVec)$.
Next to estimate how fast $\vg$ can change as a function of $\beta$
we estimate \eqref{eq:weights_full:1} in a similar manner to Lemma~\ref{lem:weights_full:consistency}.
Note that
\[
\mg(\beta)-\alpha\mLambda_{g}\succeq(1-\alpha)\mg(\beta)\succeq(1-\alpha)\mSigma(\beta)
\]
where $\mSigma(\beta)\defeq\mSigma_{\ma_{s}}(\vg^{\alpha}(\beta))$.
Consequently,
\begin{eqnarray}
\normFull{\mg(\beta)^{-1}\frac{d\vg}{d\beta}}_{\mSigma(\beta)}^{2} & \leq & \normFull{\left(\mg(\beta)-\alpha\mLambda_{g}\right)^{-1}\onesVec}_{\mSigma(\beta)}^{2}\nonumber \\
 & \leq & \frac{1}{1-\alpha}\normFull{\onesVec}_{\mSigma(\beta)}^{2}=\frac{\rank\left(\ma\right)}{1-\alpha}.\label{eq:g_beta_2_norm_est}
\end{eqnarray}
Using this estimate of how much $\vg$ changes in the $\mSigma(\beta)$
norm, we now estimate how much $\vg$ changes in the $\ellInf$ norm.
Let $\vz\defeq\left(\mg(\beta)-\alpha\mLambda_{g}\right)^{-1}\onesVec$.
Then, we have
\begin{eqnarray*}
\left(\left(1-\alpha\right)\vsigma_{i}(\beta)+\beta\right)\left|\vz_{i}\right| & \leq & \left|\alpha\indicVec i^{T}\mProj^{(2)}\vz\right|+1\\
 & \leq & \alpha\vsigma_{i}(\beta)\norm{\vz}_{\mSigma(\beta)}+1.
\end{eqnarray*}
Using \eqref{eq:g_beta_2_norm_est} and $\alpha<1$, we have
\[
\normFull{\frac{d\ln\vg}{d\beta}}_{\infty}=\norm{\vz}_{\infty}\leq\max\left(\frac{\alpha\norm{\vz}_{\mSigma(\beta)}}{1-\alpha},\frac{1}{\beta}\right)\leq\max\left(\frac{\sqrt{\rank\left(\ma\right)}}{\left(1-\alpha\right)^{3/2}},\frac{1}{\beta}\right).
\]
Using \eqref{eq:after_step3a}, integrating, and applying Lemma~\ref{lem:appendix:log_helper}
we have that
\[
\normInf{\mg(\beta-\theta)^{-1}(\vg(\beta-\theta)-\vWeight)}\leq\frac{1}{12c_{r}}
\]
for $\theta\leq\frac{(1-\alpha)^{3/2}\beta}{1000c_{r}^{2}\sqrt{\rank\left(\ma\right)}}$.
Hence, this proves that step 2a preserves the invariant \eqref{eq:initialweightinvariant}
at step 2a. Hence, the algorithm satisfies the assumptions needed
for Theorem~\ref{thm:weights_full:approximate_weight} throughout
and $\code{computeWeight}$ ins step 2a works as desired. Since each
iteration $\beta$ decreased by $\tilde{O}\left(1/\sqrt{\rank\left(\ma\right)}\right)$
portion and the initial $\beta$ is $\tilde{O}(1)$ we see that the
algorithm requires only $\tilde{O}\left(\sqrt{\rank\left(\ma\right)}\right)$
iterations. Using Theorem~\ref{thm:weights_full:approximate_weight}
to bound the total number of linear systems solved then yields the
result.

\end{proof}

\section{Approximate Weights Suffice}

\label{sec:approx_path}

In the previous sections, we analyzed a weighted path following strategy
assuming oracle access to a weight function we could compute exactly
and showed how to compute a weight function approximately. In this
section we show why it suffices to compute multiplicative approximations
to the weight function. Ultimately, having access to this ``noisy
oracle'' will only cause us to lose polylogarithmic factors in the
running time as compared to the ``exact oracle'' case.

This is a non-trivial statement as the weight function serves several
roles in our weighted path following scheme. First, it ensures a good
ratio between total weight $\cWeightSize$ and slack sensitivity $\cWeightStab$.
This allows us to take make large increases to the path parameter
$t$ after which we can improve centrality. Second, the weight function
is consistent and does not differ too much from the $\cWeightCons$-update
step direction. This allows us to change the weights between $\cWeightCons$-update
steps without moving too far away from the central path. Given a multiplicative
approximation to the weight function, this first property is preserved
up to an approximation constant however this second property is not. 

To effectively use multiplicative approximations to the weight function
we cannot simply use the weight function directly. Rather we need
to smooth out changes to the weights by using some slowly changing
approximation to the weight function. In this section we show how
this can be achieved in general. First, in Section \ref{sec:approx_path:chasing_zero},
we present the smoothing problem in a general form that we call the
\emph{chasing 0 game} and we provide an effective strategy for playing
this game. Then in Section~\ref{sec:approx_path:centering} we show
how to use this strategy to produce a weighted path following scheme
that uses multiplicative approximations to the weight function.

\subsection{The Chasing 0 Game}

\label{sec:approx_path:chasing_zero}

\global\long\def\trInit{\vx^{(0)}}
 \global\long\def\trCurr{\vx^{(k)}}
 \global\long\def\trNext{\vx^{(k + 1)}}
 \global\long\def\trAdve{\vy^{(k)}}
 \global\long\def\trAfterAdve{\vy}
 \global\long\def\trMeas{\vz^{(k)}}
 \global\long\def\trAfterMeas{\vz}
 \global\long\def\trGradCurr{\grad\Phi_{\mu}(\trCurr)}
 \global\long\def\trGradAdve{\grad\Phi_{\mu}(\trAdve)}
 \global\long\def\trGradMeas{\grad\Phi_{\mu}(\trMeas)}
 \global\long\def\trGradAfterAdve{\grad\Phi_{\mu}(\trAfterAdve)}
 \global\long\def\trGradAfterMeas{\grad\Phi_{\mu}(\trAfterMeas)}
 \global\long\def\trSetCurr{U^{(k)}}

The \emph{chasing 0 game} is as follows. There is a player, an adversary,
and a point $\vx\in\Rm$. The goal of the player is to keep the point
close to $\vzero$ in $\ellInf$ norm and the goal of the adversary
tries to move $\vx$ away from $\vzero\in\Rm.$ The game proceeds
for an infinite number of iterations where in each iteration the adversary
moves the current point $\trCurr\in\Rm$ to some new point $\trAdve\in\Rm$
and the player needs to respond. The player does not know $\trCurr$,
$\trAdve$, or the move of the adversary. All the player knows is
that the adversary moved the point within some convex set $\trSetCurr$
and the player knows some $\trMeas\in\Rn$ that is close to $\trAdve$
in $\ellInf$ norm.%
\footnote{To apply this result to weighted central path following we let the
current points $\trCurr$ denote the difference between $\log(\vWeight)$
and $\log(\vg\left(\vx\right))$. The sets $\trSetCurr$ are then
related to the $\cWeightCons$-update steps and the steps of the player
are related to the weights the path following strategy picks.%
} With this information the player is allowed to move the point a little
more than the adversary. Formally, the player is allowed to set the
next point to $\trNext\in\Rm$ such that $\vDelta^{(k)}\defeq\trNext-\trAdve\in(1+\epsilon)U$
for some fixed $\epsilon>0$. 

The question we would like to address is, how close the player can
keep $\trNext$ to $\vzero$ in $\ellInf$ norm? In particular, we
would like an efficient strategy for computing $\vDelta^{(k)}$ such
that $\normInf{\trCurr}$ is bounded for all $k\geq0$. 

\begin{center}
\begin{tabular}{|l|}
\hline 
Chasing 0 Game:\tabularnewline
\hline 
\hline 
1. Given $R>0,\epsilon>0,\trInit\in\Rm$.\tabularnewline
\hline 
2. For $k=1,2,\cdots$\tabularnewline
\hline 
2a. \textbf{$\ $}The adversary announces symmetric convex set $U^{(k)}\subseteq\Rn$
and $\vu^{(k)}\in U^{(k)}$.\tabularnewline
\hline 
2b. \textbf{$\ $}The adversary sets $\trAdve:=\trCurr+\vu^{(k)}$.\tabularnewline
\hline 
2c. \textbf{$\ $}The adversary announces $\trMeas$ such that $\norm{\trMeas-\trAdve}_{\infty}\leq R$.\tabularnewline
\hline 
2d. \textbf{$\ $}The player chooses $\vDelta^{(k)}\in\left(1+\epsilon\right)\trSetCurr$.\tabularnewline
\hline 
2e. \textbf{$\ $}The player sets $\trNext=\trAdve+\vDelta^{(k)}.$\tabularnewline
\hline 
\end{tabular}
\par\end{center}

We show that assuming that the $\trSetCurr$ are sufficiently bounded
then there is strategy that the player can follow to ensure that that
$\normInf{\trCurr}$ is never too large. Our strategy simply consists
of taking ``gradient steps'' using the following potential function.

\begin{definition}\label{def:potential_function_tracing_zero} For
any $\mu\geq0$ let $p_{\mu}:\R\rightarrow\R$ and $\Phi_{\mu}:\Rm\rightarrow\R$
be given by
\[
\forall x\in\R\enspace:\enspace p_{\mu}(x)\defeq e^{\mu x}+e^{-\mu x}\enspace\text{ and }\enspace\Phi_{\mu}(\vx)\defeq\sum_{i\in[m]}p_{\mu}(x_{i}).
\]
\end{definition} In other words, for all $k$ we simply set $\vDelta^{(k)}$
to be the vector in $(1+\epsilon)\trSetCurr$ that best minimizes
the potential function of the observed position, i.e. $\Phi_{\mu}(\trMeas)$
for an appropriate choice of $\mu$. In the following theorem we show
that this suffices to keep $\Phi_{\mu}(\trCurr)$ small and that small
$\Phi_{\mu}(\trCurr)$ implies small $\normInf{\trCurr}$ and hence
has the desired properties.

\begin{theorem} \label{thm:approx_path:chasing_zero} Suppose that
each $\trSetCurr$ is a symmetric convex set that contains an $\ellInf$
ball of radius $r_{k}$ and is contained in a $\ellInf$ ball of radius
$R_{k}\leq R$.%
\footnote{Formally we assume that if $\vx\in U^{(k)}$ then $\normInf{\vx}\leq R$
and we assume that if $\normInf{\vx}\leq r$ then $\vx\in\trSetCurr$.%
} Let $0<\epsilon<\frac{1}{5}$ and consider the strategy
\[
\vDelta^{(k)}=\left(1+\epsilon\right)\argmin_{\vDelta\in U^{(k)}}\left\langle \nabla\Phi_{\mu}(\vz^{(k)}),\vDelta\right\rangle \enspace\text{ where }\enspace\mu=\frac{\epsilon}{12R}.
\]
Let $\tau\defeq\max_{k}\frac{R_{k}}{r_{k}}$ and suppose $\Phi_{\mu}(\trInit)\leq\frac{12m\tau}{\epsilon}$
(or more specifically $\normInf{\trInit}\leq\frac{12R}{\epsilon}\log\left(\frac{6\tau}{\epsilon}\right)$
) then
\[
\forall k\geq0\enspace:\enspace\Phi_{\mu}(\trNext)\leq\left(1-\frac{\epsilon^{2}r_{k}}{24R}\right)\Phi_{\mu}(\trCurr)+\epsilon m\frac{R_{k}}{2R}\leq\frac{12m\tau}{\epsilon}.
\]
In particular, we have $\norm{\trCurr}_{\infty}\leq\frac{12R}{\epsilon}\log\left(\frac{12m\tau}{\epsilon}\right)$.
\end{theorem}

To prove Theorem~\ref{thm:approx_path:chasing_zero} we first provide
the following lemma regarding properties of the potential function
$\Phi_{\mu}$.

\begin{lemma}[Properties of the Potential Function] \label{lem:smoothing:helper}
For all $\vx\in\Rm$, we have
\begin{equation}
e^{\mu\|\vx\|_{\infty}}\leq\Phi_{\mu}(\vx)\leq2me^{\mu\|\vx\|_{\infty}}\enspace\text{ and }\enspace\mu\Phi_{\mu}(\vx)-2\mu m\leq\normOne{\grad\Phi_{\mu}(\vx)}\label{eq:smoothing:pot_prop_1}
\end{equation}
Furthermore, for any symmetric convex set $U\subseteq\Rm$ and any
$\vx\in\Rm$, let $\vx^{\flat}\defeq\argmax_{\vy\in U}\innerProduct{\vx}{\vy}$%
\footnote{This is a scaled version of $\#$ operator in \cite{Nesterov2012}
and hence we name it differently.%
} and $\norm{\vx}_{U}\defeq\max_{\vy\in U}\innerProduct{\vx}{\vy}$.
Then for all $\vx,\vy\in\Rm$ with $\normInf{\vx-\vy}\leq\delta\leq\frac{1}{5\mu}$
we have
\begin{equation}
e^{-\mu\delta}\norm{\grad\Phi_{\mu}(\vy)}_{U}-\mu\normOne{\grad\Phi_{\mu}(\vy)^{\flat}}\leq\innerProduct{\grad\Phi_{\mu}(\vx)}{\grad\Phi_{\mu}(\vy)^{\flat}}\leq e^{\mu\delta}\norm{\grad\Phi_{\mu}(\vy)}_{U}+\mu e^{\mu\delta}\normOne{\grad\Phi_{\mu}(\vy)^{\flat}}.\label{eq:smoothing:pot_prop_2}
\end{equation}
If additionally $U$ is contained in a $\ellInf$ ball of radius $R$
then
\begin{equation}
e^{-\mu\delta}\norm{\grad\Phi_{\mu}(\vy)}_{U}-\mu mR\leq\norm{\grad\Phi_{\mu}(\vx)}_{U}\leq e^{\mu\delta}\norm{\grad\Phi_{\mu}(\vy)}_{U}+\mu e^{\mu\delta}mR.\label{eq:smoothing:pot_prop_3}
\end{equation}
\end{lemma}

\begin{proof} First we note that for all $x\in\R$ we have
\[
e^{\mu|x|}\leq p_{\mu}(x)\leq2e^{\mu|x|}\enspace\text{and}\enspace p'_{\mu}(x)=\mu\sign(x)\left(e^{\mu|x|}-e^{-\mu|x|}\right)
\]
and therefore we have \eqref{eq:smoothing:pot_prop_1}. 

Next let $x,y\in\R$ such that $|x-y|\leq\delta$. Note that $\left|p'_{\mu}(x)\right|=p'_{\mu}(\left|x\right|)=\mu\left(e^{\mu|x|}-e^{-\mu|x|}\right)$
and since $\left|x-y\right|\leq\delta$ we have that $|x|=|y|+z$
for some $z\in(-\delta,\delta)$. Using that $p'(|x|)$ is monotonic
in $|x|$ we then have
\begin{align}
|p'_{\mu}(x)| & =p'_{\mu}(|x|)=p'_{\mu}(|y|+z)\leq p'_{\mu}(|y|+\delta)\nonumber \\
 & =\mu\left(e^{\mu|y|+\mu\delta}-e^{-\mu|y|-\mu\delta}\right)=e^{\mu\delta}p'(|y|)+\mu\left(e^{\mu\delta-\mu|y|}-e^{-\mu|y|-\mu\delta}\right)\nonumber \\
 & \leq e^{\mu\delta}\left|p'(y)\right|+\mu e^{\mu\delta}.\label{eq:smoothingp2}
\end{align}
By symmetry (i.e. replacing $x$ and $y$) this implies that
\begin{equation}
|p'_{\mu}(x)|\geq e^{-\mu\delta}|p'(y)|-\mu\label{eq:smoothing:p3}
\end{equation}

Since $U$ is symmetric this implies that for all $i\in[m]$ we have
$\sign(\grad\Phi_{\mu}(\vy)^{\flat})_{i}=\grad\Phi_{\mu}(\vy)_{i}=\sign(y_{i})$
. Therefore, if for all $i\in[n]$ we have $\sign(x_{i})=\sign(y_{i})$,
by \eqref{eq:smoothingp2}, we see that 
\begin{eqnarray*}
\innerProduct{\grad\Phi_{\mu}(\vx)}{\grad\Phi_{\mu}(\vy)^{\flat}} & = & \sum_{i}p'_{\mu}(x_{i})\grad\Phi_{\mu}(\vy)_{i}^{\flat}\\
 & \leq & \sum_{i}\left(e^{\mu\delta}p'_{\mu}(y_{i})+\mu e^{\mu\delta}\right)\grad\Phi_{\mu}(\vy)_{i}^{\flat}\\
 & \leq & e^{\mu\delta}\left\langle \grad\Phi_{\mu}(\vy),\grad\Phi_{\mu}(\vy)^{\flat}\right\rangle +\mu e^{\mu\delta}\norm{\grad\Phi_{\mu}(\vy)^{\flat}}_{1}\\
 & = & e^{\mu\delta}\norm{\grad\Phi_{\mu}(\vy)}_{U}+\mu e^{\mu\delta}\normOne{\grad\Phi_{\mu}(\vy)^{\flat}}.
\end{eqnarray*}
Similarly, using \eqref{eq:smoothing:p3}, we have $e^{-\mu\delta}\norm{\grad\Phi_{\mu}(\vy)}_{U}-\mu\normOne{\grad\Phi_{\mu}(\vy)^{\flat}}\leq\innerProduct{\grad\Phi_{\mu}(\vx)}{\grad\Phi_{\mu}(\vy)^{\flat}}$
and hence \eqref{eq:smoothing:pot_prop_2} holds. On the other hand
if $\sign(x_{i})\neq\sign(y_{i})$ then we know that $|x_{i}|\leq\delta$
and consequently $|p'_{\mu}(x_{i})|\leq\mu(e^{\mu\delta}-e^{-\mu\delta})\leq\frac{\mu}{2}$
since $\delta\leq\frac{1}{5\mu}$. Thus, we have
\[
e^{-\mu\delta}\left|p'_{\mu}(y_{i})\right|-\mu\leq-\frac{\mu}{2}\leq\sign\left(y_{i}\right)p'_{\mu}(x_{i})\leq0\leq e^{\mu\delta}\left|p'_{\mu}(y_{i})\right|+\mu e^{\mu\delta}.
\]
Taking inner product on both sides with $\grad\Phi_{\mu}(\vy)_{i}^{\flat}$
and using definition of $\norm{\cdot}_{U}$ and $\cdot^{\flat}$,
we get \eqref{eq:smoothing:pot_prop_2}. Thus, \eqref{eq:smoothing:pot_prop_2}
holds in general.

Finally we note that since $U$ is contained in a $\ellInf$ ball
of radius $R$, we have $\normOne{\vy^{\flat}}\leq mR$ for all $\vy$.
Using this fact, \eqref{eq:smoothing:pot_prop_2}, and the definition
of $\norm{\cdot}_{U}$, we obtain
\[
e^{-\mu\delta}\norm{\grad\Phi_{\mu}(\vy)}_{U}-\mu mR\leq\innerProduct{\grad\Phi_{\mu}(\vx)}{\grad\Phi_{\mu}(\vy)^{\flat}}\leq\norm{\grad\Phi_{\mu}(\vx)}_{U}
\]
where the last line comes from the fact $\grad\Phi_{\mu}(\vy)^{\flat}\in U$
and the definition of $\norm{\cdot}_{U}$. By symmetry \eqref{eq:smoothing:pot_prop_3}
follows. \end{proof}

Using Lemma~\ref{lem:smoothing:helper} we prove Theorem~\ref{thm:approx_path:chasing_zero}.

\begin{proof}{[}Theorem~\ref{thm:approx_path:chasing_zero}{]} For
the remainder of the proof, let $\norm{\vx}_{U^{(k)}}=\max_{\vy\in U^{(k)}}\left\langle \vx,\vy\right\rangle $
and $\vx^{\flat_{(k)}}=\argmax_{\vy\in U^{(k)}}\left\langle \vx,\vy\right\rangle $.
Since $U^{(k)}$ is symmetric, we know that $\vDelta^{(k)}=-\left(1+\epsilon\right)\left(\grad\Phi_{\mu}(\vz^{(k)})\right)^{\flat_{(k)}}$
and therefore by applying the mean value theorem twice we have that
\begin{eqnarray*}
\Phi_{\mu}(\trNext) & = & \Phi_{\mu}(\trAdve)+\left\langle \grad\Phi_{\mu}(\trAfterMeas),\trNext-\trAdve\right\rangle \\
 & = & \Phi_{\mu}(\trCurr)+\left\langle \grad\Phi_{\mu}(\trAfterAdve),\trAdve-\trCurr\right\rangle +\left\langle \grad\Phi_{\mu}(\trAfterMeas),\trNext-\trAdve\right\rangle 
\end{eqnarray*}
for some $\trAfterAdve$ between $\trAdve$ and $\trCurr$ and some
$\trAfterMeas$ between $\vx^{(k+1)}$ and $\vy^{(k)}$. Now, using
that $\trAdve-\trCurr\in U^{(k)}$ and that $\trNext-\trAdve=\vDelta^{(k)}$
we have
\begin{eqnarray}
\Phi_{\mu}(\trNext) & \leq & \Phi_{\mu}(\trCurr)+\norm{\trGradAfterAdve}_{\trSetCurr}-\left(1+\epsilon\right)\left\langle \trGradAfterMeas,\left(\trGradMeas\right)^{\flat_{(k)}}\right\rangle .\label{eq:Phi_est_1}
\end{eqnarray}
Since $U^{k}$ is contained within the $\ellInf$ ball of radius $R_{k}$
Lemma~\ref{lem:smoothing:helper} shows that
\begin{equation}
\norm{\trGradAfterAdve}_{\trSetCurr}\leq e^{\mu R_{k}}\norm{\trGradCurr}_{\trSetCurr}+m\mu R_{k}e^{\mu R_{k}}.\label{eq:Phi_est_2}
\end{equation}
Furthermore, since $\epsilon<\frac{1}{5}$ and $R_{k}\leq R$, by
triangle inequality we have $\normInf{\trAfterMeas-\trMeas}\leq(1+\epsilon)R_{k}+R\leq3R$
and $\normInf{\trMeas-\trCurr}\leq2R$. Therefore, applying Lemma~\ref{lem:smoothing:helper}
twice yields that
\begin{align}
\innerProduct{\trGradAfterMeas}{\trGradMeas^{\flat_{(k)}}} & \geq e^{-3\mu R}\norm{\trGradMeas}_{\trSetCurr}-\mu mR_{k}\nonumber \\
 & \geq e^{-5\mu R}\norm{\trGradCurr}_{\trSetCurr}-2\mu mR_{k}.\label{eq:Phi_est_3}
\end{align}
Combining \eqref{eq:Phi_est_1}, \eqref{eq:Phi_est_2}, and \eqref{eq:Phi_est_3}
then yields that
\[
\Phi_{\mu}(\trNext)\leq\Phi_{\mu}(\trCurr)-\left((1+\epsilon)e^{-5\mu R}-e^{\mu R}\right)\norm{\trGradCurr}_{\trSetCurr}+m\mu R_{k}e^{\mu R}+2(1+\epsilon)m\mu R_{k}.
\]
Since we chose $\mu=\frac{\epsilon}{12R}$, we have
\[
1+\epsilon\leq\frac{\epsilon}{2}+\left(1+6\mu R\right)\leq\frac{\epsilon}{2}e^{5\mu R}+e^{6\mu R}.
\]
Hence, we have $(1+\epsilon)e^{-5\mu R}-e^{\mu R}\leq\frac{\epsilon}{2}$.
Also, since $0<\epsilon<\frac{1}{5}$ we have
\[
m\mu R_{k}e^{\mu R}+2(1+\epsilon)m\mu R_{k}\leq\left(e^{\mu R}+2(1+\epsilon)\right)m\mu R_{k}\leq\epsilon m\frac{7R_{k}}{24R}.
\]
Thus, we have
\[
\Phi_{\mu}(\trNext)\leq\Phi_{\mu}(\trCurr)-\frac{\epsilon}{2}\norm{\trGradCurr}_{\trSetCurr}+\epsilon m\frac{7R_{k}}{24R}.
\]
Using Lemma~\ref{lem:smoothing:helper} and the fact that $U_{k}$
contains a $\ellInf$ ball of radius $r_{k}$, we have
\[
\norm{\trGradCurr}_{\trSetCurr}\geq r_{k}\norm{\trGradCurr}_{1}\geq\frac{\epsilon r_{k}}{12R}\left(\Phi_{\mu}(\trCurr)-2m\right).
\]
Therefore, we have that
\begin{eqnarray*}
\Phi_{\mu}(\trNext) & \leq & \left(1-\frac{\epsilon^{2}r_{k}}{24R}\right)\Phi_{\mu}(\trCurr)+\frac{\epsilon r_{k}}{12R}m+\epsilon m\frac{7R_{k}}{24R}\\
 & \leq & \left(1-\frac{\epsilon^{2}r_{k}}{24R}\right)\Phi_{\mu}(\trCurr)+\epsilon m\frac{R_{k}}{2R}.
\end{eqnarray*}
Hence, if $\Phi_{\mu}(\trCurr)\leq\frac{12m\tau}{\epsilon}$, we have
$\Phi_{\mu}(\trNext)\leq\frac{12m\tau}{\epsilon}$. Since$\Phi_{\mu}(\trInit)\leq\frac{12m\tau}{\epsilon}$
by assumption we have by induction that $\Phi_{\mu}(\trCurr)\leq\frac{12m\tau}{\epsilon}$
for all $k$. The necessary bound on $\normInf{\trCurr}$ then follows
immediately from Lemma~\ref{lem:smoothing:helper}. \end{proof}

\subsection{Centering Step With Noisy Weight}

\label{sec:approx_path:centering}

Here we show how to use the results of the previous section to perform
weighted path following given access only to a multiplicative approximation
of the weight function. In particular, we show how to use Theorem
\ref{thm:approx_path:chasing_zero} to improve the centrality of $\vx$
while maintaining the invariant that $\vw$ is close to $\vg(\vx)$
multiplicatively.

As in Section~\ref{sec:weighted_path} given a feasible point, $\{\vx,\vw\}\in\dFull$,
we measure the distance between the current weights, $\vw\in\dWeights$,
and the weight function, $\vg(\vs)\in\dWeights$, in log scale $\vWeightError(\vs,\vw)\defeq\log(\fvWeight(\vs))-\log(\vw)$.
Our goal is to keep $\normInf{\vWeightError(\vs,\vw)}\leq K$ for
some error threshold $K$. We choose $K$ to be just small enough
that we can still decrease $\delta_{t}(\vx,\vw)$ linearly and still
approximate $\vg(\vs)$, as in general it may be difficult to compute
$\vg(\vs)$ when $\vw$ is far from $\vg(\vs)$. Furthermore, we ensure
that $\vWeightError$ doesn't change too much in either $\|\cdot\|_{\infty}$
or $\|\cdot\|_{\mWeightNext}$ and thereby ensure that the centrality
does not increase too much as we move $\vw$ towards $\vg(\vs)$. 

We meet these goals by playing the chasing 0 game where the vector
we wish to keep near $\vzero$ is $\vWeightError(\vs,\vw)$, the adversaries
moves are $\cWeightCons$-steps, and our moves change $\log(\vWeight)$.
The $\cWeightCons$-step decreases $\delta_{t}$ and since we are
playing the chasing 0 game we keep $\vWeightError(\vs,\vw)$ small.
Finally, since by the rules of the chasing 0 game we do not move $\vw$
much more than $\vg(\vs)$ has moved, we have by similar reasoning
to the exact weight computation case, Theorem~\ref{thm:centering_exact}
that changing $\vw$ does not increase $\delta_{t}$ too much. This
\emph{inexact centering} operation and the analysis are formally defined
and analyzed below.

Most of the parameter balancing involved in this paper lies in the
theorem below. Due to the step consistency, we know know that after
a $\cWeightCons$-steps, the weight does not move too far away that
we can move it back without hurting centrality too much if we can
compute the weight exactly. The Chasing $0$ game shows that we can
mimic this if we compute the weight accurate enough. Therefore, the
balancing is simply about how accurate we need to do.

\begin{center}
\begin{tabular}{|l|}
\hline 
$(\vxNext,\vWeightFin)=\centeringInexact(\vxCurr,\vWeightCurr,K,\code{approxWeight})$\tabularnewline
\hline 
\hline 
1. $R=\frac{K}{60c_{r}\log\left(960\cWeightCons\cWeightStab m^{3/2}\right)}$,
$\delta_{t}=\delta_{t}(\vxCurr,\vWeightCurr)$, $\epsilon=\frac{1}{5\cWeightCons}$
and $\mu=\frac{\epsilon}{12R}.$\tabularnewline
\hline 
2. $\{\vxNext,\vWeightNext\}=\updateStep_{t}(\vxCurr,\vWeightCurr,\cWeightCons)$
as in Definition \ref{def:weighted_path:r_step}. \tabularnewline
\hline 
3. Let $U=\{\vy\in\Rm~|~\norm{\vy}_{\mWeightNext}\leq\frac{\cWeightCons+0.14}{\cWeightCons+1}\delta_{t}\text{ and }\norm{\vy}_{\infty}\leq4\cWeightStab\delta_{t}\}$ \tabularnewline
\hline 
4. Compute $\vz=\code{approxWeight}(\vs,\vWeightNext,R)$.\tabularnewline
\hline 
5. $\vWeightFin:=\exp\left(\log(\vWeightNext)+\left(1+\epsilon\right)\argmin_{\vu\in U}\left\langle \nabla\Phi_{\mu}\left(\log(\vz)-\log\left(\vWeightNext\right)\right),\vu\right\rangle \right)$ \tabularnewline
\hline 
\end{tabular}
\par\end{center}

\begin{flushleft}
Note that in step 5 in $\centeringInexact$, we need to project a
certain vector onto the intersection of ball, $\norm{\cdot}_{\mWeightNext}$,
and box, $\normInf{\cdot}$. In Section~\ref{sec:app:Project_ball_box}
we show that this can be computed in parallel in depth $\tilde{O}(1)$
and work $\tilde{O}(m)$ and therefore this step is not a bottleneck
in the computational cost of our weighted path following schemes.
\par\end{flushleft}

\begin{theorem}[Centering with Inexact Weights] \label{thm:smoothing:centering_inexact_weight}
Given current point $\{\vxCurr,\vWeightCurr\}\in\dFull$, error parameter
$K\leq\frac{1}{8\cWeightCons}$, and approximate weight computation
oracle, $\code{approxWeight}$, such that $\norm{\log(\mbox{\ensuremath{\code{approxWeight}}}(\vs,\vw,R))-\log\left(\vg(\vs)\right)}_{\infty}\leq R$
for $\vs,\vw\in\dWeights$ with $\norm{\log(\vw)-\log\left(\vg(\vs)\right)}_{\infty}\leq2K$,
assume that
\[
\delta_{t}\defeq\delta_{t}(\vxCurr,\vWeightCurr)\leq\frac{K}{240c_{r}c_{\gamma}\log\left(960\cWeightCons\cWeightStab m^{3/2}\right)}\enspace\text{ and }\enspace\Phi_{\mu}\defeq\Phi_{\mu}(\vWeightError(\vxCurr,\vWeightCurr))\leq960\cWeightCons\cWeightStab m^{3/2}
\]
where $\mu=\frac{\epsilon}{12R}$. Let $(\vxNext,\vWeightFin)=\mathbf{\centeringInexact}(\vxCurr,\vWeightCurr,K)$,
then
\[
\delta_{t}(\vxNext,\vWeightFin)\leq\left(1-\frac{0.5}{1+\cWeightCons}\right)\delta_{t}
\]
and
\[
\Phi_{\mu}(\vWeightError(\vxNext,\vWeightFin))\leq\left(1-\frac{\delta_{t}}{600c_{r}^{2}R\sqrt{m}}\right)\Phi_{\mu}(\trCurr)+\frac{2m\cWeightStab\delta_{t}}{5R}\leq960\cWeightCons\cWeightStab m^{3/2}.
\]
Also, we have $\norm{\log(\vg(\next{\vs}))-\log(\vWeightFin)}_{\infty}\leq K$.\end{theorem}

\begin{proof}By Lemma~\ref{lem:weighted_path:weight_progress},
we know that for a $\cWeightCons$-update step, we have $\vWeightError(\vxNext,\vWeightNext)-\vWeightError(\vxCurr,\vWeightCurr)\in\overline{U}$
where $\overline{U}$ is the symmetric convex set given by
\[
\overline{U}\defeq\{\vy\in\R^{m}~|~\norm{\vy}_{\mWeightNext}\leq C_{w}\enspace\text{ and }\enspace\norm{\vy}_{\infty}\leq C_{\infty}\}
\]
where
\[
C_{\infty}=4\cWeightStab\delta_{t}\enspace\text{ and }\enspace C_{w}=\frac{\cWeightCons+1/8}{\cWeightCons+1}\delta_{t}+13\cWeightStab\delta_{t}^{2}.
\]
Note that since $\delta_{t}\leq K\left(240c_{r}c_{\gamma}\log\left(960\cWeightCons\cWeightStab m^{3/2}\right)\right)^{-1}$
we have
\[
C_{\infty}\leq4c_{\gamma}\left(\frac{K}{240c_{r}c_{\gamma}\log\left(960\cWeightCons\cWeightStab m^{3/2}\right)}\right)\leq\frac{K}{60c_{r}\log\left(960\cWeightCons\cWeightStab m^{3/2}\right)}=R
\]
Therefore $\overline{U}$ is contained in a $\ellInf$ ball of radius
$R$. Again using the bound on $\delta_{t}$ we have
\begin{eqnarray}
C_{\weight} & = & \frac{\cWeightCons+\frac{1}{8}}{\cWeightCons+1}\delta_{t}+13\cWeightStab\delta_{t}^{2}\leq\frac{\cWeightCons+\frac{1}{8}}{\cWeightCons+1}\delta_{t}+\frac{0.008}{c_{r}}\delta_{t}\nonumber \\
 & \leq & \frac{\cWeightCons+0.14}{\cWeightCons+1}\delta_{t}.\label{eq:C_w_est}
\end{eqnarray}
Consequently, $\overline{U}\subseteq U$ where we recall that $U$
is the symmetric convex set defined by
\[
U=\{\vy\in\R^{m}~|~\norm{\vy}_{\mWeightNext}\leq\frac{\cWeightCons+0.14}{\cWeightCons+1}\delta_{t}\enspace\text{ and }\enspace\norm{\vy}_{\infty}\leq4\cWeightStab\delta_{t}\}.
\]
Therefore, we can play the chasing 0 game on $\vWeightError(\curr{\vs},\vWeightCurr)$
attempting to maintain the invariant that $\normInf{\vWeightError(\curr{\vs},\vWeightCurr)}\leq K\leq\frac{1}{8\cWeightCons}$
without taking steps that are more than $1+\epsilon$ times the size
of $U$. We pick $\epsilon=\frac{1}{5\cWeightCons}$ so to not interfere
with our ability to decrease $\delta_{t}$ linearly. 

To use the chasing 0 game to maintain $\normInf{\vWeightError(\curr{\vs},\vWeightCurr)}\leq K$
we need to ensure that $R$ satisfies the following
\[
\frac{12R}{\epsilon}\log\left(\frac{12m\tau}{\epsilon}\right)\leq K
\]
where here $\tau$ is as defined in Theorem~\ref{thm:approx_path:chasing_zero}.
To bound $\tau$ we need to lower bound the radius of the $\ellInf$
ball that $U$ contains. Since $\normInf{\fvWeight(\vSlackCurr)}\leq2$
by Definition~\ref{def:sec_weighted_path:weight_function} and since
$\normInf{\vWeightError(\vxCurr,\vWeightCurr)}\leq\frac{1}{8}$ by
assumption we have that $\normInf{\curr{\vWeight}}\leq3$. By Lemma
\ref{lem:weighted_path:stab_update_step} we know that $\normInf{\next{\vWeight}}\leq4$
if $\delta_{t}c_{\gamma}\leq\frac{1}{8}$ and consequently
\[
\forall u\in\Rm\enspace:\enspace\norm{\vu}_{\infty}^{2}\geq\frac{1}{4m}\norm{\vu}_{\mWeightNext}^{2}.
\]
Consequently, if $\norm{\vu}_{\infty}\leq\frac{\delta_{t}}{4\sqrt{m}}$,
then $\vu\in U$. Thus, $U$ contains a a box of radius $\frac{\delta_{t}}{4\sqrt{m}}$
and since $U$ is contained in a box of radius $4\cWeightStab\delta_{t}$,
we have that $\tau\leq16\cWeightStab\sqrt{m}$ and consequently
\[
\frac{12R}{\epsilon}\log\left(\frac{12m\tau}{\epsilon}\right)\leq60c_{r}R\log\left(960\cWeightCons\cWeightStab m^{3/2}\right)\leq K.
\]
This proves that we meet the conditions of Theorem \ref{thm:approx_path:chasing_zero}.
Therefore, we have
\begin{eqnarray*}
\Phi_{\mu}(\vWeightError(\vxNext,\vWeightFin)) & \leq & \left(1-\frac{\epsilon^{2}}{24R}\left(\frac{\delta_{t}}{4\sqrt{m}}\right)\right)\Phi_{\mu}(\trCurr)+\epsilon m\frac{1}{2R}\left(4\cWeightStab\delta_{t}\right)\\
 & = & \left(1-\frac{\delta_{t}}{600c_{r}^{2}R\sqrt{m}}\right)\Phi_{\mu}(\trCurr)+\frac{2m\cWeightStab\delta_{t}}{5R}\\
 & \leq & 960\cWeightCons\cWeightStab m^{3/2}.
\end{eqnarray*}
where we do not need to re-derive the last line because it follows
from Theorem \ref{thm:approx_path:chasing_zero}. 

Consequently, $\normInf{\vWeightError(\vxCurr,\vWeightCurr)}\leq K$
and $\Phi_{\mu}(\vWeightError(\vxNext,\vWeightFin))\leq960\cWeightCons\cWeightStab m^{3/2}$.
Since $K\leq\frac{1}{8}$, we have $\normInf{\fmWeight(\vSlackCurr)^{-1}(\vWeightCurr-\fvWeight(\vSlackCurr))}\leq1.2$
and $\energyStab(\vSlackCurr,\vWeightCurr)\leq2\cWeightStab.$ Consequently,
by Lemma~\ref{lem:weighted_path:x_progress} we have
\[
\delta_{t}(\vxNext,\vWeightNext)\leq\energyStab(\vxCurr,\vWeightCurr)\cdot\delta_{t}^{2}\leq2\cdot\cWeightStab\cdot\delta_{t}^{2}
\]
Let
\[
\epsilon_{\infty}\defeq\normInf{\log(\vWeightFin)-\log(\vWeightNext)}\enspace\text{ and }\enspace\epsilon_{w}\defeq\norm{\log(\vWeightFin)-\log(\vWeightNext)}_{\mWeightNext}.
\]
By our bounds on $U$, we have
\[
\epsilon_{\infty}\leq(1+\epsilon)R\leq\frac{1}{100c_{r}}\text{ and }\epsilon_{w}=(1+\epsilon)\left[\frac{\cWeightCons+0.14}{\cWeightCons+1}\delta_{t}\right]\leq\frac{c_{r}+0.37}{c_{r}+1}\delta_{t}.
\]
Using Lemma~\ref{lem:weight_change}, we have that
\begin{align*}
\delta_{t}(\vxNext,\vWeightFin) & \leq(1+\epsilon_{\infty})\left[\delta_{t}(\vxNext,\vWeightNext)+\epsilon_{w}\right]\leq3\cWeightStab\delta_{t}^{2}+(1+\epsilon_{\infty})\epsilon_{w}\\
 & \leq\left(1+\frac{1}{100c_{r}}\right)\left(\frac{c_{r}+0.34}{c_{r}+1}\right)\delta_{t}+3\cWeightStab\delta_{t}^{2}\leq\left(\frac{c_{r}+0.5}{c_{r}+1}\right)\delta_{t}
\end{align*}
\end{proof} 

\section{The Algorithm}

\label{sec:algorithm}

In this section we show how to put together the results of the previous
sections to solve a linear program. First, in Section~\ref{sub:path_following}
we provide a path following routine that allows us to move quickly
from one approximate central path point to another. Using this subroutine,
in Section~\ref{sec:algorithm:solving_lp} we show how to obtain
an algorithm for solving a linear program in $\otilde(\sqrt{\rank(\ma)}L)$
iterations that consist of solving linear systems in the original
constraint matrix. In the Appendix we provide additional proof details
such as how these algorithm only require approximate linear system
solvers (Appendix~\ref{sec:app:Inexact-Linear-Algebra}) and how
to initialize our interior point technique and round approximate solutions
to optimal ones (Appendix~\ref{sec:app:bit_complexity}).

\subsection{Path Following}

\label{sub:path_following}

We start by analyzing the running time of $\code{pathFollowing}$
a subroutine for following the weighted central path. 

\begin{center}
\begin{tabular}{|l|}
\hline 
\textbf{$\ensuremath{(\vxNext,\next{\vWeight})=\code{pathFollowing}(\vxCurr,\vWeightCurr,t_{\text{start}},t_{\text{end}})}$}\tabularnewline
\hline 
\hline 
1. $\cWeightCons=2\log_{2}\left(\frac{2m}{\rank(\ma)}\right),t=t_{\text{start}},K=\frac{1}{24c_{r}}.$\tabularnewline
\hline 
2. While $t<t_{\text{end}}$\tabularnewline
\hline 
2a. $\ $$(\vxNext,\vWeightFin)=\centeringInexact(\vxCurr,\vWeightCurr,K,\ensuremath{\code{computeWeight}})$\tabularnewline
\hline 
2b. $\ $$t^{(new)}:=t\left(1+\frac{1}{10^{10}c_{r}^{3}\log\left(\cWeightCons m\right)\sqrt{\rank(\ma)}}\right)$.\tabularnewline
\hline 
2c. $\ $$\vxCurr:=\vxNext$, $\vWeightCurr:=\vWeightFin$, $t:=t^{(new)}$\tabularnewline
\hline 
2d. $\ $For every $\frac{m}{100c_{r}\log(c_{r}m)}$ steps, check
if the current $\vx$, $\vWeight$ satisfies the $\delta$ and $\Phi$
invariants.\tabularnewline
\hphantom{2d.} If it does not satisfies, roll back to the last time
the invariants were met.\tabularnewline
\hline 
3. Output $(\vx^{(old)},\vw^{(old)})$.\tabularnewline
\hline 
\end{tabular}
\par\end{center}

\begin{theorem}[Main Result] \label{thm:path_following}Given $\{\vxCurr,\vWeightCurr\}\in\dFull$
and $t_{\text{start}}\leq t_{\text{end}}$. Suppose that
\[
\delta_{t_{\text{start}}}(\vxCurr,\vWeightCurr)\leq\frac{1}{11520c_{r}^{2}\log\left(1920\cWeightCons m^{3/2}\right)}\enspace\text{ and }\enspace\Phi_{\mu}(\vWeightError(\vxCurr,\vWeightCurr))\leq1920\cWeightCons m^{3/2}
\]
where $\mu=2\log\left(52\cWeightCons m\right)/K$. Let $\ensuremath{(\vxNext,\vWeightFin)=\code{pathFollowing}(\vxCurr,\vWeightCurr,t_{\text{start}},t_{\text{end}})}$,
then
\[
\delta_{t_{\text{end}}}(\vxNext,\next{\vWeight})\leq\frac{1}{11520c_{r}^{2}\log\left(1920\cWeightCons m^{3/2}\right)}\enspace\text{ and }\enspace\Phi_{\mu}(\vWeightError(\vxNext,\next{\vWeight}))\leq1920\cWeightCons m^{3/2}.
\]
Furthermore, computing $(\vxNext,\next{\vWeight})$ takes $\tilde{O}\left(\sqrt{\rank(\ma)}\log\left(\frac{t_{\text{end}}}{t_{\text{start}}}\right)\right)$
iterations in expectation where the cost of each iteration is dominated
by the time need to solve $\tilde{O}(1)$ linear system solves.

\end{theorem}

\begin{proof} This algorithm maintains the invariant that
\[
\delta_{t}(\vxCurr,\vWeightCurr)\leq\frac{1}{11520c_{r}^{2}\log\left(1920\cWeightCons m^{3/2}\right)}\text{ and }\Phi_{\mu}(\vWeightError(\vxCurr,\vWeightCurr))\leq1920\cWeightCons m^{3/2}
\]
in each iteration in the beginning of the step (2a). Note that our
oracle $\code{computeWeight}$ satisfies the assumption of Theorem~\ref{thm:smoothing:centering_inexact_weight}
since $2K\leq\frac{1}{12c_{r}}.$ Hence, $\centeringInexact$ can
use $\code{computeWeight}$ to find the approximations of $\vg(\vSlackNext)$.
Hence, Theorem \ref{thm:smoothing:centering_inexact_weight} shows
that we have
\[
\delta_{t}(\vxNext,\vWeightFin)\leq\left(1-\frac{0.5}{1+\cWeightCons}\right)\delta_{t}\enspace\text{ and }\enspace\Phi_{\mu}(\vWeightError(\vxNext,\vWeightFin))\leq1920\cWeightCons m^{3/2}.
\]
Using the fact $\cWeightSize(\vg)\leq2\rank(\ma)$ and that $\vWeightNext$
is within a multiplicative factor of two of $\vg(\vSlackNext)$ by
Lemma \ref{lem:weighted_path:t_step} we have
\begin{align*}
 & \delta_{t^{(new)}}(\vxNext,\vWeightFin)\\
 & \leq\left(1+\frac{1}{10^{10}c_{r}^{3}\log\left(\cWeightCons m\right)\sqrt{\rank(\ma)}}\right)\left(1-\frac{0.5}{1+\cWeightCons}\right)\delta_{t}+\frac{\sqrt{\normOne{\vw^{(new)}}}}{10^{10}c_{r}^{3}\log\left(\cWeightCons m\right)\sqrt{\rank(\ma)}}\\
 & \leq\frac{1}{11520c_{r}^{2}\log\left(1920\cWeightCons m^{3/2}\right)}
\end{align*}

Theorem \ref{thm:weights_full:approximate_weight} shows that with
probability $\left(1-\frac{1}{m}\right)^{\left\lceil 12c_{r}\log\left(\frac{4m}{K}\right)\right\rceil }$,
$\code{computeWeight}$ outputs a correct answer. Therefore, for each
$\frac{m}{100c_{r}\log(c_{r}m)}$ iterations there is constant probability
that the whole procedure runs correctly. Hence, we only need to know
how long it takes to check the current state satisfies $\delta_{t}$
and $\Phi_{\mu}$ invariants. We can check the $\delta_{t}$ easily
using only $1$ linear system solve. To check $\Phi_{\mu}$, we need
to compute the weight function exactly. To do this, we use lemma \ref{lem:weight_iterative}
and note that computing the leverage scores exactly takes $m$ linear
system solve. Therefore, the averaged cost of step 2d is just $\tilde{O}(1)$
linear system solves and this justified the total running time.

\end{proof}

\subsection{Solving a Linear Program}

\label{sec:algorithm:solving_lp}

Here we show how to use the properties of $\code{pathFollowing}$
proved in Theorem~\ref{thm:path_following} to obtain a linear program
solver. Given the previous theorem all that remains is to show how
to get the initial central point and round the optimal point to a
vertex. We defer much of the proof of how to obtain an initial point,
deal with unbounded solutions, and round to an optimal vertex to Lemma~\ref{lem:the_modified_LP}
proved in Appendix~\ref{sec:app:bit_complexity}.

\begin{theorem}\label{thm:LPsolve}Consider a linear programming
problem of the form
\begin{equation}
\min_{\vx\in\Rn~:~\ma\vx\geq\vb}\vc^{T}\vx\label{eq:algo:problem}
\end{equation}
where $\ma\in\Rmn,$ $\vb\in\Rm$, and $\vc\in\Rn$ have integer coefficients.
Let $L$ denote the bit complexity of \eqref{eq:algo:problem} and
suppose that for any positive definite diagonal matrix $\md\in\Rmm$
with condition number $2^{\tilde{O}(L)}$ there is an algorithm $\code{solve}(\ma,\vb,\md,\epsilon)$
such that
\begin{equation}
\norm{\code{solve}(\ma,\vb,\md,\epsilon)-\left(\md\ma\right)^{+}\vb}_{\ma^{T}\md^{2}\ma}\leq\epsilon\norm{\left(\md\ma\right)^{+}\vb}_{\ma^{T}\md^{2}\ma}\label{eq:solver_assumption-1}
\end{equation}
in time $O\left(\mathcal{T}\log(1/\epsilon)\right)$ for any $\epsilon>0$
with success probability greater than $1-\frac{1}{m}$. Then, there
is an algorithm to solve \eqref{eq:algo:problem} in expected time
$\tilde{O}\left(\sqrt{\rank(\ma)}\left(\mathcal{T}+\nnz(\ma)\right)L\right)$,
i.e, find the active constraints of an optimal solution or prove that
the program is unfeasible or unbounded. 

Using \cite{nelson2012osnap} as the $\code{Solve}$ algorithm, we
obtain an algorithm that solves \eqref{eq:algo:problem} in time
\[
\tilde{O}\left(\sqrt{\rank(\ma)}\left(\nnz(\ma)+\left(\rank(\ma)\right)^{\omega}\right)L\right).
\]
where $\omega<2.3729$ \cite{williams2012matrixmult} is the matrix
multiplication constant.

\end{theorem}

\begin{proof} Applying the Lemma~\ref{lem:the_modified_LP} we obtain
a modified linear program
\begin{equation}
\min\left\langle \next{\vc},\vx\right\rangle \text{ given }\ma_{(new)}\vx\geq\next{\vb}\label{eq:modified_LP}
\end{equation}
which is bounded and feasible with $O(n)$ variables, $O(m)$ constraints,
$O(\rank(\ma))$ rank and $\tilde{O}(L)$ bit complexity. Also, we
are given an explicit interior point $\vx_{0}$. 

To obtain an initial weighted central path point, we can use Theorem~\ref{thm:weights_full:initialapproximate_weight}.
However, $\vx$ may not be close to central path, i.e. $\delta_{t}$
could be large. To fix this, we can temporarily change the cost function
such that $\delta_{t}=0$. In particular, we can set $\vc_{\text{modified}}=\ma^{T}\mSlackX^{-1}\vWeight$
and get $\delta_{t}=0$ for this modified cost function. One can think
of Theorem~\ref{thm:path_following} as showing that we can get the
central path point from a certain cost function $t_{\text{start}}\vc$
to another cost function $t_{\text{end}}\vc$ in time that depends
only logarithmically on the multiplicative difference between these
two vectors. Clearly, instead of increasing $t$ we can decrease $t$
similarly. Hence, we can decrease $t$ such that we get the central
path point $\vx_{\text{center}}$ for the cost function $2^{-\tilde{\Theta}\left(L\right)}\vc_{\text{modified}}$.
Since $2^{-\tilde{\Theta}\left(L\right)}$ is close enough to zero,
it can be shown that $\delta_{t}$ is small also for the cost function
$2^{-\tilde{\Theta}\left(L\right)}\vc$. Then, we could use Theorem
\ref{thm:path_following} to increase $t$ and obtain the central
path point for $t=2^{\tilde{\Theta}\left(L\right)}$.

Then, we can use $\centeringInexact$ to make $\delta_{t}$ becomes
and hence $\vc^{T}\vx_{t}$ close to $\vc^{T}\vx$. By a standard
duality gap theorem,%
\footnote{See \cite{lsMaxflow} or \cite{Nesterov2003} for a more detailed
treatment of this fact in a more general regime.%
} we know that the duality gap of $\vx_{t}$ is less than $\norm{\vWeight}_{1}/t$
and in this case it is less than $2^{-\tilde{\Theta}\left(L\right)}$
because $\norm{\vWeight}_{1}\leq2\rank\left(\ma\right)$. Now, we
can use the conclusion of the Lemma~\ref{lem:the_modified_LP} to
find the active constraints of an optimal solution of the original
linear program or prove that it is infeasible or unbounded.

During the algorithm, we only called the function $\centeringInexact$
$\tilde{O}(L)$ times and hence the algorithm only executes $\tilde{O}(L)$
linear system solves. In Section~\ref{sec:app:Inexact-Linear-Algebra},
we show that these linear systems do not need to be solved exactly
and that inexact linear algebra suffices. Using this observation and
letting using \cite{nelson2012osnap} as the $\code{solve}$ routine
yields the total running time of 
\[
\tilde{O}\left(\sqrt{\rank(\ma)}\left(\nnz(\ma)+\left(\rank(\ma)\right)^{\omega}\right)L\right).
\]
\end{proof} 

In Section \ref{sec:app:Project_ball_box}, we show that the projection
problem in $\centeringInexact$ can be computed in $\otilde(1)$ depth
and $\tilde{O}(m)$ work and other operations are standard parallelizable
linear algebra operations. Therefore, we achieve the first $\otilde(\sqrt{\rank\left(\ma\right)}L)$
depth polynomial work method for solving linear programs.

\begin{theorem}\label{thm:LPsolve_p} There is an $\otilde(\sqrt{\rank\left(\ma\right)}L)$
depth polynomial work algorithm to solve linear program of the form
\[
\min_{\vx\in\Rn~:~\ma\vx\geq\vb}\vc^{T}\vx
\]
where $L$ denote the bit complexity of the linear program. \end{theorem}

\subsection{Accelerating the Solver}

\label{sec:algorithm:accelerated}

In this section, we show that how we can apply acceleration methods
for decreasing the iterations of interior point techniques can be
applied to our algorithm to yield a faster method. In particular we
show how to adapt techniques of Vaidya \cite{vaidya1989speeding}
for using fast matrix multiplication to obtain a faster running time.
Our goal here is to provide a simple exposition of how the iteration
costs of our method can be decrease. We make no attempt to explore
the running time of our algorithm in all regimes and we note that
since our algorithm only needs to solve linear systems in scalings
of the original constraint matrix there may be techniques to improve
our algorithm further in specific regimes by exploiting structure
in $\ma$.

To accelerate our path following method, we note that we solve systems
of two forms: we solve systems in $\ma^{T}\ms^{-1}\mWeight\ms^{-1}\ma$
to update $\vx$ and we solve systems in $\ma^{T}\ms^{-1}\mWeight^{\alpha}\ms^{-1}\ma$
to update $\vWeight$. Since we have proven in Lemma \ref{lem:weights_full:cond_number}
that two system are spectrally similar, we only need to know how to
solve system of the form $\ma^{T}\ms^{-1}\mWeight\ms^{-1}\ma$ and
then we can use preconditioning to solve either system. Furthermore,
we note similarly to Vaidya \cite{vaidya1989speeding} that the $\mSlack$
and $\mw$ matrices do not change too much from iteration to iteration
and therefore a sequence of the necessary linear system can be solved
faster than considering them individually. Below we state and Appendix~\ref{sec:app:acceleration}
we prove a slight improvement of a result in \cite{vaidya1989speeding}
formally analyzing one way of solving these systems faster.

\begin{theorem}\label{thm:alg:accel} Let $\vd^{(i)}\in\dSlack$
be a sequence of $r$ positive vectors. Suppose that the number of
times that $d_{j}^{(i)}\neq d_{j}^{(i+1)}$ for any $i\in[r]$ and
$j\in[m]$ is bounded by $Cr^{2}$ for some $C\geq1$. Then if we
are given the $\vd^{(i)}$ in a sequence, in each iteration $i$ we
can compute $\left(\ma^{T}\md_{i}\ma\right)^{-1}\vx_{i}$ for $\md_{i}=\mDiag(\vd_{i})$
and arbitrary $\vx_{i}\in\Rn$ with average cost per iteration
\[
\otilde\left(\frac{mn^{\omega-1}}{r}+n^{2}+C^{\omega}r^{2\omega}+C^{\omega-1}nr^{\omega}\right)
\]
where $\omega<2.3729$ \cite{williams2012matrixmult} is the matrix
multiplication constant.

\end{theorem}

Using Theorem~\ref{thm:alg:accel} we simply need to estimate how
much the diagonal entries $\ms^{-1}\mWeight\ms^{-1}$ to obtain a
faster linear program solver. We prove the following. 

\begin{theorem}For any $\frac{n}{m}\leq\beta\leq1$ and $r>1$, there
is an
\begin{equation}
\tilde{O}\left(\sqrt{m\beta}\left(\nnz(\ma)+n^{2}+\frac{mn^{\omega-1}}{r}+\beta^{-\omega}r^{2\omega}+\beta^{-(\omega-1)}nr^{\omega}\right)L\right)\label{eq:accel:1}
\end{equation}
 time algorithm for solving linear programming problems of the form
\[
\min\vc^{T}\vx\text{ given }\ma\vx\geq\vb
\]
where $\ma\in\R^{m\times n}$.

\end{theorem}

\begin{proof} Instead of using $\beta=\frac{n}{m}$ in the weight
function we let $\beta\in[\frac{n}{m},1]$ be arbitrary as in the
theorem statement. Looking at the analysis in Section~\ref{sec:weights_full}
we see that this yields a weight function with $c_{1}=O\left(\beta m\right)$,
$c_{\gamma}=O(1)$ and $c_{r}=\tilde{O}(1)$. Consequently, it takes
$\tilde{O}\left(\sqrt{\beta m}L\right)$ iterations to solve the linear
program.

We separate the sequence of the linear systems involved into groups
of size $r$. To use the previous theorem to compute $\left(\ma^{T}\md_{j}\ma\right)^{-1}\vx$
for each group of operations, we need to estimate the change of the
diagonal entries $\ms^{-1}\mWeight\ms^{-1}$. For the change of $\ms$,
Lemma \ref{lem:weighted_path:stab_update_step} shows that
\[
\norm{\log\left(\vs_{j}\right)-\log\left(\vs_{j+1}\right)}_{\mWeight_{j}}=O(1).
\]
Since we have added $\beta$ in the weight function, we have $\vWeight_{i}\geq\beta$
and 
\[
\norm{\log\left(\vs_{j}\right)-\log\left(\vs_{j+1}\right)}_{2}=O(\beta^{-1/2}).
\]
Therefore, in a period of $r$ operations, at most $O\left(\beta^{-1}r^{2}\right)$
coordinates can change multiplicatively by a constant factor. Similarly,
we can use inequality \eqref{eq:C_w_est} to analyze the change of
$\mWeight$. 

Therefore, we can maintain a vector $\vd$ such that $\md$ is spectrally
similar to $\ms^{-1}\mWeight\ms^{-1}$ while only changing $\vd$
a total of $O\left(\beta^{-1}r^{2}\right)$ over a sequence of $r$
operations. Using Theorem~\ref{thm:alg:accel} and using $\ma\md\ma$
as pre-conditioner for the necessary linear system solves, we can
solve the linear system with average cost
\[
\tilde{O}\left(\nnz(\ma)+n^{2}+\frac{mn^{\omega-1}}{r}+\beta^{-\omega}r^{2\omega}+\beta^{-(\omega-1)}nr^{\omega}\right).
\]
Using that the total number of iterations is $\tilde{O}\left(\sqrt{\beta m}L\right)$
then yields \eqref{eq:accel:1}.

\end{proof}

\section{Acknowledgments}

We thank Yan Kit Chim, Andreea Gane, Jonathan A. Kelner, Lap Chi Lau,
Aleksander Mądry, Cameron Musco, Christopher Musco, Lorenzo Orecchia,
Ka Yu Tam and Nisheeth Vishnoi for many helpful conversations. This
work was partially supported by NSF awards 0843915 and 1111109, NSF
Graduate Research Fellowship (grant no. 1122374) and Hong Kong RGC
grant 2150701. Finally, we thank the referees for extraordinary efforts
and many helpful suggestions.

\bibliographystyle{plain}
\bibliography{main}

\begin{thebibliography}{10}

\bibitem{anstreicher96}
Kurt~M. Anstreicher.
\newblock Volumetric path following algorithms for linear programming.
\newblock {\em Math. Program.}, 76:245--263, 1996.

\bibitem{clarkson2013low}
Kenneth~L Clarkson and David~P Woodruff.
\newblock Low rank approximation and regression in input sparsity time.
\newblock In {\em Proceedings of the 45th annual ACM symposium on Symposium on
  theory of computing}, pages 81--90. ACM, 2013.

\bibitem{daitch2008faster}
Samuel~I Daitch and Daniel~A Spielman.
\newblock Faster approximate lossy generalized flow via interior point
  algorithms.
\newblock In {\em Proceedings of the 40th annual ACM symposium on Theory of
  computing}, pages 451--460. ACM, 2008.

\bibitem{dantzig1951maximization}
George~B Dantzig.
\newblock Maximization of a linear function of variables subject to linear
  inequalities.
\newblock {\em New York}, 1951.

\bibitem{freund_weighted}
RobertM. Freund.
\newblock Projective transformations for interior-point algorithms, and a
  superlinearly convergent algorithm for the w-center problem.
\newblock {\em Mathematical Programming}, 58(1-3):385--414, 1993.

\bibitem{GoldbergRao}
Andrew~V. Goldberg and Satish Rao.
\newblock Beyond the flow decomposition barrier.
\newblock {\em J. ACM}, 45(5):783--797, 1998.

\bibitem{gonzaga1992path}
Clovis~C Gonzaga.
\newblock Path-following methods for linear programming.
\newblock {\em SIAM review}, 34(2):167--224, 1992.

\bibitem{karmarkar1984new}
Narendra Karmarkar.
\newblock A new polynomial-time algorithm for linear programming.
\newblock In {\em Proceedings of the sixteenth annual ACM symposium on Theory
  of computing}, pages 302--311. ACM, 1984.

\bibitem{Kelner2013}
Jonathan~A. Kelner, Lorenzo Orecchia, Aaron Sidford, and Zeyuan~Allen Zhu.
\newblock {A Simple, Combinatorial Algorithm for Solving SDD Systems in
  Nearly-Linear Time}.
\newblock January 2013.

\bibitem{khachiyan1980polynomial}
Leonid~G Khachiyan.
\newblock Polynomial algorithms in linear programming.
\newblock {\em USSR Computational Mathematics and Mathematical Physics},
  20(1):53--72, 1980.

\bibitem{khachiyan1996rounding}
Leonid~G Khachiyan.
\newblock Rounding of polytopes in the real number model of computation.
\newblock {\em Mathematics of Operations Research}, 21(2):307--320, 1996.

\bibitem{khachiyan1988method}
LG~Khachiyan, SP~Tarasov, and II~Erlikh.
\newblock The method of inscribed ellipsoids.
\newblock In {\em Soviet Math. Dokl}, volume~37, pages 226--230, 1988.

\bibitem{klivans2001randomness}
Adam~R Klivans and Daniel Spielman.
\newblock Randomness efficient identity testing of multivariate polynomials.
\newblock In {\em Proceedings of the thirty-third annual ACM symposium on
  Theory of computing}, pages 216--223. ACM, 2001.

\bibitem{KMP11}
Ioannis Koutis, Gary~L. Miller, and Richard Peng.
\newblock A nearly-m log n time solver for sdd linear systems.
\newblock In {\em Foundations of Computer Science (FOCS), 2011 IEEE 52nd Annual
  Symposium on}, pages 590 --598, oct. 2011.

\bibitem{lee2013ACDM}
Yin~Tat Lee and Aaron Sidford.
\newblock Efficient accelerated coordinate descent methods and faster
  algorithms for solving linear systems.
\newblock In {\em The 54th Annual Symposium on Foundations of Computer Science
  (FOCS)}, 2013.

\bibitem{lsMaxflow}
Yin~Tat Lee and Aaron Sidford.
\newblock Path finding ii: An$\backslash$\~{} o (m sqrt (n)) algorithm for the
  minimum cost flow problem.
\newblock {\em arXiv preprint arXiv:1312.6713}, 2013.

\bibitem{li2012iterative}
Mu~Li, Gary~L Miller, and Richard Peng.
\newblock Iterative row sampling.
\newblock 2012.

\bibitem{lovaszV06}
L{\'a}szl{\'o} Lov{\'a}sz and Santosh Vempala.
\newblock Simulated annealing in convex bodies and an {\it o}$^{\mbox{*}}$({\it
  n}$^{\mbox{4}}$) volume algorithm.
\newblock {\em J. Comput. Syst. Sci.}, 72(2):392--417, 2006.

\bibitem{madryFlow}
Aleksander Madry.
\newblock Navigating central path with electrical flows: from flows to
  matchings, and back.
\newblock In {\em Proceedings of the 54th Annual Symposium on Foundations of
  Computer Science}, 2013.

\bibitem{mahoney11survey}
Michael~W. Mahoney.
\newblock Randomized algorithms for matrices and data.
\newblock {\em Foundations and Trends in Machine Learning}, 3(2):123--224,
  2011.

\bibitem{megiddo_weighted}
Nimrod Megiddo.
\newblock Pathways to the optimal set in linear programming.
\newblock In Nimrod Megiddo, editor, {\em Progress in Mathematical
  Programming}, pages 131--158. Springer New York, 1989.

\bibitem{mizuno1993adaptive}
Shinji Mizuno, Michael~J Todd, and Yinyu Ye.
\newblock On adaptive-step primal-dual interior-point algorithms for linear
  programming.
\newblock {\em Mathematics of Operations research}, 18(4):964--981, 1993.

\bibitem{mut2013tight}
Murat Mut and Tam{\'a}s Terlaky.
\newblock A tight iteration-complexity upper bound for the mty
  predictor-corrector algorithm via redundant klee-minty cubes.
\newblock 2013.

\bibitem{nelson2012osnap}
Jelani Nelson and Huy~L Nguy{\^e}n.
\newblock Osnap: Faster numerical linear algebra algorithms via sparser
  subspace embeddings.
\newblock {\em arXiv preprint arXiv:1211.1002}, 2012.

\bibitem{Nesterov2003}
Yu~Nesterov.
\newblock {\em {Introductory Lectures on Convex Optimization: A Basic Course}},
  volume~I.
\newblock 2003.

\bibitem{Nesterov:2008:RCS:1451525.1451532}
Yu. Nesterov.
\newblock Rounding of convex sets and efficient gradient methods for linear
  programming problems.
\newblock {\em Optimization Methods Software}, 23(1):109--128, February 2008.

\bibitem{Nesterov2012}
Yu~Nesterov.
\newblock {Efficiency of coordinate descent methods on huge-scale optimization
  problems}.
\newblock {\em SIAM Journal on Optimization}, 22(2):341--362, 2012.

\bibitem{nesterov1997self}
Yu~E Nesterov and Michael~J Todd.
\newblock Self-scaled barriers and interior-point methods for convex
  programming.
\newblock {\em Mathematics of Operations research}, 22(1):1--42, 1997.

\bibitem{Nesterov1994}
Yurii Nesterov and Arkadii~Semenovich Nemirovskii.
\newblock {\em Interior-point polynomial algorithms in convex programming},
  volume~13.
\newblock Society for Industrial and Applied Mathematics, 1994.

\bibitem{papadimitriou1998combinatorial}
Christos~H Papadimitriou and Kenneth Steiglitz.
\newblock {\em Combinatorial optimization: algorithms and complexity}.
\newblock Courier Dover Publications, 1998.

\bibitem{renegar1988polynomial}
James Renegar.
\newblock A polynomial-time algorithm, based on newton's method, for linear
  programming.
\newblock {\em Mathematical Programming}, 40(1-3):59--93, 1988.

\bibitem{spielmanS08sparsRes}
Daniel~A Spielman and Nikhil Srivastava.
\newblock Graph sparsification by effective resistances.
\newblock {\em SIAM Journal on Computing}, 40(6):1913--1926, 2011.

\bibitem{spielman2004nearly}
Daniel~A Spielman and Shang-Hua Teng.
\newblock Nearly-linear time algorithms for graph partitioning, graph
  sparsification, and solving linear systems.
\newblock In {\em Proceedings of the thirty-sixth annual ACM symposium on
  Theory of computing}, pages 81--90. ACM, 2004.

\bibitem{strang1991inverse}
Gilbert Strang.
\newblock Inverse problems and derivatives of determinants.
\newblock {\em Archive for Rational Mechanics and Analysis}, 114(3):255--265,
  1991.

\bibitem{todd1994scaling}
Michael~J Todd.
\newblock Scaling, shifting and weighting in interior-point methods.
\newblock {\em Computational Optimization and Applications}, 3(4):305--315,
  1994.

\bibitem{vaidya89convexSet}
Pravin~M. Vaidya.
\newblock A new algorithm for minimizing convex functions over convex sets
  (extended abstract).
\newblock In {\em FOCS}, pages 338--343, 1989.

\bibitem{vaidya1989speeding}
Pravin~M Vaidya.
\newblock Speeding-up linear programming using fast matrix multiplication.
\newblock In {\em Foundations of Computer Science, 1989., 30th Annual Symposium
  on}, pages 332--337. IEEE, 1989.

\bibitem{vaidya1987speeding}
Pravin~M Vaidya.
\newblock An algorithm for linear programming which requires o (((m+ n) n 2+(m+
  n) 1.5 n) l) arithmetic operations.
\newblock {\em Mathematical Programming}, 47(1-3):175--201, 1990.

\bibitem{vaidya90parallel}
Pravin~M. Vaidya.
\newblock Reducing the parallel complexity of certain linear programming
  problems (extended abstract).
\newblock In {\em FOCS}, pages 583--589, 1990.

\bibitem{vaidya1996new}
Pravin~M Vaidya.
\newblock A new algorithm for minimizing convex functions over convex sets.
\newblock {\em Mathematical Programming}, 73(3):291--341, 1996.

\bibitem{vaidya1993technique}
Pravin~M Vaidya and David~S Atkinson.
\newblock A technique for bounding the number of iterations in path following
  algorithms.
\newblock {\em Complexity in Numerical Optimization}, pages 462--489, 1993.

\bibitem{williams2012matrixmult}
Virginia~Vassilevska Williams.
\newblock Multiplying matrices faster than coppersmith-winograd.
\newblock In {\em Proceedings of the forty-fourth annual ACM symposium on
  Theory of computing}, pages 887--898. ACM, 2012.

\bibitem{ye2011interior}
Yinyu Ye.
\newblock {\em Interior point algorithms: theory and analysis}, volume~44.
\newblock John Wiley \& Sons, 2011.

\end{thebibliography}

\appendix

\section{Glossary}

\label{sec:glossary}

Here we summarize problem specific notation we use throughout the
paper. For many quantities we included the typical order of magnitude
as they appear during our algorithms.
\begin{itemize}
\item Linear program related: constraint matrix $\ma\in\R^{m\times n}$
, cost vector $\vc\in\Rn$, constraint vector $\vb\in\Rm$, solution
$\vx\in\mathbb{R}^{n}$, weights of constraints $\vWeight\in\mathbb{R}^{m}$
where $m$ is the number of constraints and $n$ is the number of
variables.
\item Bit complexity: $L=\log(m)+\log(1+d_{max})+\log(1+\max\{\normInf{\vc},\normInf{\vb}\})$
where $d_{max}$ is the largest absolute value of the determinant
of a square sub-matrix of $\ma$.
\item Slacks: $\vs(\vx)=\ma\vx-\vb$.
\item Matrix version of variables: $\ms$ is the diagonal matrix corresponds
to $\vs$, $\mw$ corresponds to $\vWeight$, $\mg$ corresponds to
$\vg$.
\item Penalized objective function (\ref{eq:penalized_object_function}):
$\penalizedObjective(\vx,\vWeight)=t\cdot\vc^{T}\vx-\sum_{i\in[m]}\weight_{i}\log\slackXi.$
\item Newton step (\ref{eq:weighted_path:newton_step}): $\vNewtonStep_{t}(\vx,\vWeight)=(\hessXX\penalizedObjective(\vx,\vWeight))^{-1}\gradX\penalizedObjective(\vx,\vWeight)=\left(\ma^{T}\ms^{-1}\mw\ms^{-1}\ma\right)^{-1}\left(t\vc-\ma^{T}\ms^{-1}\vWeight\right).$
\item Centrality (\ref{eq:weighted_path:energy_x}): $\delta_{t}(\vx,\vw)=\norm{\vNewtonStep_{t}(\vx,\vWeight)}_{\hessXX\penalizedObjective(\vx,\vWeight)}\approx\frac{1}{\polylog(m)}.$
\item Slack Sensitivity(\ref{eq:weighted_path:energy_x}): $\energyStab(\vSlack,\vWeight)=\max_{i\in[m]}\norm{\mWeight^{-1/2}\indicVec i}_{\mProj_{\mSlack^{-1}\ma}\left(\vWeight\right)}\approx1.$
\item Properties of weight function (Def \ref{def:sec_weighted_path:weight_function}):
\textbf{size} $\cWeightSize(\fvWeight)=\normOne{\fvWeight(\vSlack)}\approx\rank\left(\ma\right)$,
\textbf{slack sensitivity} $\cWeightStab(\fvWeight)=\sup_{\vs}\energyStab(\vSlack,\fvWeight(\vSlack))\approx1$,
\textbf{step consistency} $\cWeightCons(\fvWeight)\approx\log\left(\frac{m}{\rank\ma}\right)$.
\item Difference between $\vg$ and $\vWeight$ (\ref{eq:distance_of_sw}):
$\vWeightError(\vs,\vw)=\log(\fvWeight(\vs))-\log(\vw).$
\item Potential function for tracing $0$ (Def \ref{def:potential_function_tracing_zero}):
$\Phi_{\mu}(\vx)=e^{\mu x}+e^{-\mu x}\approx\poly(m)$.
\item The weight function proposed (\ref{eq:sec:weights_full:weight_function}):
\[
\vg(\vSlack)=\argmin_{\vWeight\in\rPos^{m}}\penalizedObjectiveWeight(\vSlack,\vWeight)\enspace\text{ where }\enspace\penalizedObjectiveWeight(\vSlack,\vWeight)=\onesVec^{T}\vWeight-\frac{1}{\alpha}\log\det(\ma_{s}^{T}\mWeight^{\alpha}\ma_{s})-\beta\sum_{i}\log\weight_{i}
\]
where $\ma_{s}=\mSlack^{-1}\ma$, $\alpha\approx1-1/\log_{2}\left(\frac{m}{\rank(\ma)}\right)$,
$\beta\approx\rank(\ma)/m$.
\end{itemize}

\section{Technical Tools}

\label{sec:app:lemmas}

In this section, we provide and prove various mathematical facts that
we use throughout the paper.

\subsection{Matrix Properties}

First, we prove various properties regarding projection matrices that
we use throughout the paper.

\begin{lemma}[Projection Matrices] \label{lem:tool:projection_matrices}
Let $\mProj\in\Rnn$ be an arbitrary projection matrix and let $\mSigma=\mDiag(\mProj)$.
For all $i,j\in[n]$ and $\vx\in\Rn$ we have the following
\begin{description}
\item [{(1)}] $\mSigma_{ii}=\sum_{j\in[n]}\mProj_{ij}^{(2)},$ 
\item [{(2)}] $\mZero\specLeq\mProj^{(2)}\specLeq\mSigma\specLeq\iMatrix$, 
\item [{(3)}] $\mProj_{ij}^{(2)}\leq\mSigma_{ii}\mSigma_{jj}$, 
\item [{(4)}] $|\indicVec i^{T}\mProj^{(2)}\vx|\leq\mSigma_{ii}\norm{\vx}_{\mSigma}$. 
\end{description}
\end{lemma}

\begin{proof} To prove (1), we simply note that by definition of
a projection matrix $\mProj=\mProj\mProj$ and therefore
\[
\mSigma_{ii}=\mProj_{ii}=\indicVec i^{T}\mProj\indicVec i=\indicVec i^{T}\mProj\mProj\indicVec i=\sum_{j\in[n]}\mProj_{ij}^{2}=\sum_{j\in[n]}\mProj_{ij}^{(2)}
\]

To prove (2), we observe that since $\mProj$ is a projection matrix,
all its eigenvectors are either 0 or 1. Therefore, $\mSigma\specLeq\iMatrix$
and by (1) $\mSigma-\mProj^{(2)}$ is diagonally dominant. Consequently,
$\mSigma-\mProj^{(2)}\specGeq0$. Rearranging terms and using the
well known fact that the shur product of two positive semi-definite
matrices is positive semi-definite yields (2).

To prove (3), we use $\mProj=\mProj\mProj$, Cauchy-Schwarz, and (1)
to derive
\[
\mProj_{ij}=\sum_{k\in[n]}\mProj_{ik}\mProj_{kj}\leq\sqrt{\left(\sum_{k\in[n]}\mProj_{ik}^{2}\right)\left(\sum_{k\in[n]}\mProj_{kj}^{2}\right)}=\sqrt{\mSigma_{ii}\mSigma_{jj}}\enspace.
\]
Squaring then yields (3).

To prove (4), we note that by the definition of $\mProj^{(2)}$ and
Cauchy-Schwarz, we have
\begin{equation}
\abs{\indicVec i^{T}\mProj^{(2)}\vx}=\left|\sum_{j\in[n]}\mProj_{ij}^{(2)}\vec{x}_{j}\right|\leq\sqrt{\left(\sum_{j\in[n]}\mSigma_{jj}\vec{x}_{j}^{2}\right)\cdot\sum_{j\in[n]}\frac{\mProj_{ij}^{(4)}}{\mSigma_{jj}}}\label{eq:lem:proj1}
\end{equation}
Now, by (1) and (3), we know that
\begin{equation}
\sum_{j\in[n]}\frac{\mProj_{ij}^{4}}{\mSigma_{jj}}\leq\sum_{j\in[n]}\frac{\mProj_{ij}^{2}\mSigma_{ii}\mSigma_{jj}}{\mSigma_{jj}}=\mSigma_{ii}\sum_{j\in[n]}\mProj_{ij}^{2}=\mSigma_{ii}^{2}\label{eq:lem:proj2}
\end{equation}
Since $\norm{\vx}_{\mSigma}\defeq\sqrt{\sum_{j\in[n]}\mSigma_{jj}\vx_{j}^{2}}$,
combining \eqref{eq:lem:proj1} and \eqref{eq:lem:proj2} yields $\abs{\indicVec i^{T}\mProj^{(2)}\vx}\leq\mSigma_{ii}\norm{\vx}_{\mSigma}$
as desired. \end{proof}

\subsection{Taylor Expansions and Multiplicative Approximations}

\label{sec:app:lemmas:log_mult}

Throughout this paper we use $\log(\va)-\log(\vb)$ as a convenient
way of working with $\mb^{-1}(\va-\vb)$ or $\ma^{-1}(\vb-\va)$.
In this section we make this connection rigorous by providing several
helper lemmas used throughout the paper.

\begin{lemma}[Log Notation] \label{lem:appendix:log_helper} Suppose
$\normInf{\log(\va)-\log(\vb)}=\epsilon\leq1/2$ then
\begin{align*}
\norm{\mb^{-1}(\va-\vb)}_{\infty} & \leq\epsilon+\epsilon^{2}.
\end{align*}
If $\normInf{\mb^{-1}(\va-\vb)}=\epsilon\leq1/2$, then
\[
\normInf{\log(\va)-\log(\vb)}\leq\epsilon+\epsilon^{2}.
\]
\end{lemma}

\begin{proof} Using the Taylor expansion of $e^{x}$ and $\log(1+x)$,
we get the following two inequalities which prove the claim
\begin{eqnarray*}
1+x & \leq & e^{x}\leq1+x+x^{2}\text{ for }\left|x\right|\leq\frac{1}{2},\\
x-x^{2} & \leq & \log(1+x)\leq x\text{ for }\left|x\right|\leq\frac{1}{2}.
\end{eqnarray*}
 \end{proof}

\subsection{Matrix Calculus}

Here, we derive various matrix calculus formulas used in Section \ref{sec:weights_full}.
These are now somewhat standard and also discussed in \cite{vaidya1996new,anstreicher96}
but we derive them here for completeness. In this section, we define
\[
\mr_{\ma}(\vWeight)_{ij}\defeq\va_{i}^{T}(\ma^{T}\mw\ma)^{-1}\va_{j}.
\]
We start by computing the derivative of the \emph{volumetric barrier
function}, $f(\vw)\defeq\log\det(\ma^{T}\mw\ma)$.

\begin{lemma}[Derivative of Volumetric Barrier] \label{lem:deriv:log_det}
For $\ma\in\Rnm$, let $f:\rPos^{m}\rightarrow\R$ be given by $f(\vw)\defeq\log\det(\ma^{T}\mw\ma)$.
Then the following holds
\[
\forall\vw\in\rPos^{m}\enspace:\enspace\grad f(\vw)=\diag(\mr_{\ma}(\vWeight))\defeq\mLever_{\ma}(\vWeight)\mWeight^{-1}\onesVec.
\]
\end{lemma}

\begin{proof} For all $i\in[m]$ and $\vw\in\Rm$, we know that
\[
\frac{\partial}{\partial\vw_{i}}f(\vw)=\lim_{\alpha\rightarrow0}\frac{1}{\alpha}\left[f(\vw+\alpha\indicVec i)-f(\vw)\right]=\lim_{\alpha\rightarrow0}\frac{1}{\alpha}\left[\log\det(\ma^{T}\mw\ma+\alpha\va_{i}\va_{i}^{T})-\log\det(\ma^{T}\mw\ma)\right].
\]
Applying the matrix determinant lemma then yields that
\[
\frac{\partial}{\partial\vw_{i}}f(\vw)=\lim_{\alpha\rightarrow0}\frac{1}{\alpha}\left[\log\left(\det(\ma^{T}\mw\ma)\cdot(1+\alpha\va_{i}^{T}(\ma^{T}\mw\ma)^{-1}\va_{i})\right)-\log\left(\det(\ma^{T}\mw\ma)\right)\right].
\]
Therefore,
\[
\frac{\partial}{\partial\vw_{i}}f(\vw)=\lim_{\alpha\rightarrow0}\frac{\log(1+\alpha\mr(\vw)_{ii})}{\alpha}=\mr(\vw)_{ii}.
\]
\end{proof}

Next we bound the rate of change of entries of the resistance matrix.

\begin{lemma}[Derivative of Effective Resistance] \label{lem:deriv:effres}
For all $\ma\in\Rmn$, $\vWeight\in\rPos^{m}$, and $i,j,k\in[m]$
we have
\[
\frac{\partial}{\partial\vw_{k}}\left[\mResist_{\ma}(\vWeight)\right]_{ij}=-\mResist_{\ma}(\vw)_{ik}\mResist_{\ma}(\vw)_{kj}
\]
where $\diag(\mr_{\ma}(\vWeight))\defeq\mLever_{\ma}(\vWeight)\mWeight^{-1}\onesVec$.\end{lemma}

\begin{proof} By definition, we have that
\begin{equation}
\frac{\partial}{\partial\vWeight_{k}}\mResist_{\ma}(\vWeight)_{ij}=\lim_{\alpha\rightarrow0}\frac{1}{\alpha}\left[\mr(\vw+\alpha\indicVec k)_{ij}-\mr(\vw)_{ij}\right]\label{eq:deriv_eff_res_1}
\end{equation}
and
\begin{equation}
\mr(\vw+\alpha\indicVec k)_{ij}=\indicVec i^{T}\ma(\ma^{T}\mw\ma+\alpha\ma^{T}\indicVec k\indicVec k^{T}\ma)^{-1}\ma^{T}\indicVec j\enspace.\label{eq:deriv_eff_res_2}
\end{equation}
Furthermore, by applying the Sherman-Morrison formula, we know that
\begin{equation}
(\ma^{T}\mw\ma+\alpha\ma^{T}\indicVec k\indicVec k^{T}\ma)^{-1}=(\ma^{T}\mw\ma)^{+}-\frac{\alpha(\ma^{T}\mw\ma)^{-1}\ma^{T}\indicVec k\indicVec k^{T}\ma(\ma^{T}\mw\ma)^{-1}}{1+\alpha\indicVec k^{T}\ma(\ma^{T}\mw\ma)^{-1}\ma^{T}\indicVec k}.\label{eq:deriv_eff_res_3}
\end{equation}
Combining \eqref{eq:deriv_eff_res_1}, \eqref{eq:deriv_eff_res_2},
and \eqref{eq:deriv_eff_res_3} yields the result. \end{proof}

Finally, we use this to derive the Jacobian of leverage scores.

\begin{lemma}[Derivative of Leverage Scores] \label{lem:deriv:lever}
For all $\ma\in\Rmn$, $\vWeight\in\rPos^{m}$ we have the following
\[
\jacobian_{\vw}(\vLever_{\ma}(\vw))=\mLapProj_{\ma}(\vw)\mw^{-1}.
\]
\end{lemma}

\begin{proof} Since by definition $\vsigma_{\ma}(\vw)_{i}=\vw_{i}\mResist_{\ma}(\vw)_{ii}$
by the previous lemma, we have that
\[
\frac{\partial}{\partial\vw_{j}}\vsigma_{\ma}(\vw)_{i}=\indicVec{i=j}\mResist(\vw)_{ii}-\vw_{i}\mResist(\vw)_{ij}^{(2)}.
\]
Writing this in matrix form and recalling the definition of the Jacobian
then yields
\[
\jacobian_{\vWeight}(\vLever_{\ma}(\vw))=\mDiag(\mResist_{\ma}(\vw))-\mw\mResist_{\ma}(\vw)^{(2)}.
\]
Right multiplying by $\iMatrix=\mw\mw^{-1}$ and recalling the definition
of $\mLambda_{\ma}$ then yields the result. \end{proof}

\section{Projecting Onto Ball Intersect Box}

\label{sec:app:Project_ball_box}In the algorithm $\centeringInexact$,
we need to compute
\begin{equation}
\argmin_{\vu\in U}\left\langle \va,\vu\right\rangle \label{sec:app:proj:1}
\end{equation}
where $U=\{\vx\in\Rm~|~\norm{\vx}_{\mWeight}\leq b\text{ and }\norm{\vx}_{\infty}\leq c\}$
for some $\vw\geq\vzero$, i.e. we need to project $\va$ onto the
intersection of the ball, $\left\{ \vx\in\Rm\,|\,\norm{\vx}_{\mw}\leq b\right\} $,
and the box $\left\{ \vx\in\Rm\,|\,\norm{\vx}_{\infty}\leq c\right\} $.
In this section we show how this can be computed in nearly linear
time and in particular it can be computed in parallel in depth $\otilde(1)$
and work $\otilde(m)$.

Note that by rescaling we can rewrite \eqref{sec:app:proj:1} as
\begin{equation}
\argmax_{\norm{\vx}_{2}\leq1,-l_{i}\leq x_{i}\leq l_{i}}\left\langle \va,\vx\right\rangle \label{eq:projection_problem}
\end{equation}
for some $l_{i}$. Let us consider a simple algorithm which first
ignore the box constraint and find the best vector $\va$. If $\va$
does not violate any box constraint, then it is the solution. Otherwise,
we pick a most violated constraint $i$, i.e. the coordinate with
highest $\left|a_{i}\right|/l_{i}$. Then, we threshold this coordinates
and repeat the procedure on the remaining coordinate.

\begin{center}
\begin{tabular}{|l|}
\hline 
$\vx=\code{projectOntoBallBox}(\va)$\tabularnewline
\hline 
\hline 
1. Set $\va=\va/\norm{\va}_{2}$. \tabularnewline
\hline 
2. Sort the coordinate such that $\left|a_{i}\right|/l_{i}$ is in
descending order.\tabularnewline
\hline 
3. For $i=0,\cdots,m$\tabularnewline
\hline 
3a. Set $\vx=\begin{cases}
\text{sign}\left(\va_{j}\right)l_{j} & \text{if }j\in\{1,2,\cdots,i\}\\
\sqrt{\frac{1-\sum_{k=0}^{i}l_{k}^{2}}{1-\sum_{k=0}^{i}a_{k}^{2}}}\va_{j} & \text{otherwise}
\end{cases}.$\tabularnewline
\hline 
3b. If $\vx$ is a feasible solution, output $\vx$.\tabularnewline
\hline 
\end{tabular}
\par\end{center}

\begin{lemma}The algorithm $\code{projectOntoBallBox}$ outputs a
solution of the problem \eqref{eq:projection_problem}.\end{lemma}

\begin{proof}We claim that for all $k\leq i$ where $i$ is the last
step in the algorithm, we have
\[
\max_{\vx\in\Omega}\left\langle \va,\vx\right\rangle =\max_{\vx\in\Omega_{k}}\left\langle \va,\vx\right\rangle 
\]
where $\Omega=\{x\ :\ \norm{\vx}_{2}\leq1,-l_{i}\leq x_{i}\leq l_{i}\}$
and $\Omega_{k}=\Omega\cap\left\{ x\ :\ \left|x_{i}\right|=l_{i}\text{ for }i\in\{1,2,\cdots,k\}\right\} $.
Since $\vx$ is feasible at the last step, we have
\begin{eqnarray*}
\vx_{\text{last}} & = & \argmax_{\vx\in\Omega_{k}}\left\langle \va,\vx\right\rangle \\
 & = & \argmax_{\vx\in\Omega}\left\langle \va,\vx\right\rangle .
\end{eqnarray*}
Therefore, the correctness of the algorithm follows from the claim.

Now, we prove the claim by induction. The base case is trivial because
$\Omega=\Omega_{0}$. Now proceed by contradiction and suppose that
\begin{equation}
\max_{\vx\in\Omega_{k}}\left\langle \va,\vx\right\rangle >\max_{\vx\in\Omega_{k+1}}\left\langle \va,\vx\right\rangle .\label{eq:contradiction_proj}
\end{equation}
Let $\vy=\argmax_{\vx\in\Omega_{k}}\left\langle \va,\vx\right\rangle $.
If for all $j>k$, we have $\left|y_{j}\right|<l_{j}$. Then, the
$\vx$ found in the $\left(k+1\right)^{th}$ iteration is exactly
$\vy$ and it is feasible and hence the algorithm outputs $\vy$.
Otherwise, there is $j$ such that $\left|y_{j}\right|=l_{j}$. Since
$\vy\notin\Omega_{k+1}$, we have $\left|y_{k+1}\right|<l_{k+1}$
and hence $j>k+1$.

Consider
\[
\vz(t)=\vy+\frac{\text{sign}\left(y_{k+1}\right)t}{\left|y_{k+1}\right|+\epsilon}\onesVec_{k+1}-\frac{\text{sign}\left(y_{j}\right)t}{l_{j}}\onesVec_{j}
\]
where $\epsilon$ is a very small positive number. Note that $\left.\frac{d}{dt}\norm{\vz(t)}^{2}\right|_{t=0}=2\frac{\left|y_{k+1}\right|}{\left|y_{k+1}\right|+\epsilon}-2<0$
and hence $\norm{\vz(t)}_{2}\leq1$ for $t>0$ but close to $0$.
Also, we have
\[
\frac{d}{dt}\left\langle \va,\vz\right\rangle =\frac{\left|a_{k+1}\right|}{\left|y_{k+1}\right|+\epsilon}-\frac{\left|a_{j}\right|}{l_{j}}.
\]
Take $\epsilon=l_{k+1}-\left|y_{k+1}\right|$, then we have
\[
\frac{d}{dt}\left\langle \va,\vz\right\rangle =\frac{\left|a_{k+1}\right|}{l_{k+1}}-\frac{\left|a_{j}\right|}{l_{j}}>0
\]
because $j>k+1$ and $\left|a_{i}\right|/l_{i}$ is in descending
order. Therefore, $\vz(t)$ is a feasible and better solution for
small positive $t$. Hence, it proves $\vy$ is not the optimal solution
of $\max_{\vx\in\Omega_{k}}\left\langle \va,\vx\right\rangle $ that
contradicts to the definition of $\vy$.

Hence, $\max_{\vx\in\Omega}\left\langle \va,\vx\right\rangle =\max_{\vx\in\Omega_{k}}\left\langle \va,\vx\right\rangle $
and the algorithm outputs an optimal solution.\end{proof}

\begin{center}
\begin{tabular}{|l|}
\hline 
$\vx=\code{projectOntoBallBoxParallel}(\va)$\tabularnewline
\hline 
\hline 
1. Set $\va=\va/\norm{\va}_{2}$. \tabularnewline
\hline 
2. Sort the coordinate such that $\left|a_{i}\right|/l_{i}$ is in
descending order.\tabularnewline
\hline 
3. Precompute $\sum_{k=0}^{i}l_{k}^{2}$ and $\sum_{k=0}^{i}a_{k}^{2}$
for all $i$. \tabularnewline
\hline 
4. Find the first $i$ such that $\frac{1-\sum_{k=0}^{i}l_{k}^{2}}{1-\sum_{k=0}^{i}a_{k}^{2}}\leq\frac{l_{i+1}^{2}}{a_{i+1}^{2}}$.\tabularnewline
\hline 
5. Output $\vx=\begin{cases}
\text{sign}\left(\va_{j}\right)l_{j} & \text{if }j\in\{1,2,\cdots,i\}\\
\sqrt{\frac{1-\sum_{k=0}^{i}l_{k}^{2}}{1-\sum_{k=0}^{i}a_{k}^{2}}}\va_{j} & \text{otherwise}
\end{cases}.$\tabularnewline
\hline 
\end{tabular}
\par\end{center}

The algorithm $\code{projectOntoBallBoxParallel}$ is a parallel and
more efficient version $\code{projectOntoBallBox}$. All other operations
in our algorithm are standard linear algebra and hence the following
theorem shows that our linear programming solver is indeed parallelizable. 

\begin{lemma} The algorithm $\code{projectOntoBallBoxParallel}$
outputs an solution of the optimization problem \eqref{eq:projection_problem}
in depth $\tilde{O}(1)$ and work $\tilde{O}(m)$.

\end{lemma}

\begin{proof}Note that in the algorithm $\code{projectOntoBallBox}$,
the value
\[
\frac{1-\sum_{k=0}^{i}l_{k}^{2}}{1-\sum_{k=0}^{i}a_{k}^{2}}
\]
is increasing through the algorithm. To see this, note that in step
3b, if $\vx$ is not feasible, that means there is $j$ such that
\[
\frac{1-\sum_{k=0}^{i}l_{k}^{2}}{1-\sum_{k=0}^{i}a_{k}^{2}}>\frac{l_{j}^{2}}{a_{j}^{2}}.
\]
Since $a_{i}/l_{i}$ is in descending order, $j=i+1$. Therefore,
we have
\[
\frac{1-\sum_{k=0}^{i}l_{k}^{2}}{1-\sum_{k=0}^{i}a_{k}^{2}}>\frac{l_{i+1}^{2}}{a_{i+1}^{2}}.
\]
Hence, we have
\[
\frac{1-\sum_{k=0}^{i+1}l_{k}^{2}}{1-\sum_{k=0}^{i+1}a_{k}^{2}}>\frac{1-\sum_{k=0}^{i}l_{k}^{2}}{1-\sum_{k=0}^{i}a_{k}^{2}}.
\]

Using this fact, it is easy to see the algorithm $\code{projectOntoBallBoxParallel}$
and the algorithm $\code{projectOntoBallBox}$ outputs the same vector.
Obviously, all steps can be computed in depth $\tilde{O}(1)$ and
work $\tilde{O}(m)$. 

\end{proof}

\section{Inexact Linear Algebra\label{sec:app:Inexact-Linear-Algebra}}

Throughout much of our analysis of weighted path following we assumed
that linear systems in $\ma$ could be solved exactly. In this section
we relax this assumption and discuss the effect of using inexact linear
algebra in our linear programming algorithms. We show that rather
than computing $\left(\ma^{T}\md\ma\right)^{-1}\vx$ precisely for
positive diagonal matrix $\md$ it suffices to solve these systems
approximately. 

Throughout this section we assume that for any matrix $\ma\in\Rnm$
and vector $\vb\in\Rm$ there is an algorithm $\code{solve}(\ma,\vb)$
which outputs an vector $\vx$ such that
\begin{equation}
\norm{\vx-\ma^{+}\vb}_{\ma^{T}\ma}\leq\epsilon\norm{\ma^{+}\vb}_{\ma^{T}\ma}.\label{eq:solver_assumption}
\end{equation}
Since $\ma$ is full rank, we can write $\vc=\ma^{T}\vd$ for some
$\vd$. From equation \eqref{eq:weighted_path:newton_step}, the Newton
step is
\begin{align*}
\vNewtonStep_{t}(\vx,\vWeight) & =(\ma^{T}\ms^{-1}\mWeight\ms^{-1}\ma)^{-1}\ma^{T}\ms^{-1}\sqrt{\mWeight}\left(t\frac{\vs\vd}{\sqrt{\vWeight}}-\sqrt{\vWeight}\right)\\
 & =\left(\sqrt{\mWeight}\ms^{-1}\ma\right)^{+}\left(t\frac{\vs\vd}{\sqrt{\vWeight}}-\sqrt{\vWeight}\right).
\end{align*}
Suppose that we compute $\vh_{t}$ by the algorithm $\code{solve}$
above, then we have
\begin{eqnarray*}
\normFull{\code{solve}\left(\sqrt{\mWeight}\ms^{-1}\ma,t\frac{\vs\vd}{\sqrt{\vWeight}}-\sqrt{\vWeight}\right)-\vh_{t}}_{\ma^{T}\ms^{-1}\mWeight\ms^{-1}\ma} & \leq & \epsilon\norm{\vh_{t}}_{\ma^{T}\ms^{-1}\mWeight\ms^{-1}\ma}\\
 & = & \epsilon\delta_{t}\left(\vx,\vWeight\right).
\end{eqnarray*}
Hence, the outcome of $\code{solve}$ differs from the Newton step
$\vh_{t}$ by a relative small amount in $\normFull{\cdot}_{\ma^{T}\ms^{-1}\mWeight\ms^{-1}\ma}$.
Hence, it suffices to prove that $\delta_{t}$ is stable under this
small amount in $\normFull{\cdot}_{\ma^{T}\ms^{-1}\mWeight\ms^{-1}\ma}$
and hence is the algorithm $\code{solve}$ will only increase $\delta$
by a little compared with using exact linear algebra.

\begin{lemma}Let $\energyStab\defeq\energyStab(\vx,\vWeight)$ and
$\next{\vx}=\vx+\vDelta$. Let $\eta=\normFull{\vDelta}_{\ma^{T}\ms^{-1}\mWeight\ms^{-1}\ma}\leq\frac{1}{8\energyStab}$.
Then, we have
\[
\delta_{t}\left(\next{\vx},\vWeight\right)\leq\left(1-\gamma\eta\right)^{-1}\left(\delta_{t}\left(\vx,\vWeight\right)+\eta\right).
\]
\end{lemma} 

\begin{proof}By the same proof in Lemma \ref{lem:weighted_path:stab_update_step},
we have that
\[
\normInf{\ms^{-1}(\vSlackNext-\vs)}\leq\energyStab\eta.
\]
Therefore, we have
\begin{eqnarray*}
\delta_{t}\left(\next{\vx},\vWeight\right) & = & \normFull{t\vc-\ma^{T}\mSlackNext^{-1}\vWeight}_{\left(\ma^{T}\mSlackNext^{-1}\mWeight\mSlackNext^{-1}\ma\right)^{-1}}\\
 & \leq & \left(1+\gamma\eta\right)\normFull{t\vc-\ma^{T}\mSlackNext^{-1}\vWeight}_{\left(\ma^{T}\mSlackNext^{-1}\mWeight\mSlackNext^{-1}\ma\right)^{-1}}\\
 & \leq & \left(1+\gamma\eta\right)\normFull{t\vc-\ma^{T}\ms^{-1}\vWeight}_{\left(\ma^{T}\ms^{-1}\mWeight\ms^{-1}\ma\right)^{-1}}+\normFull{\ma^{T}\left(\frac{\vWeight}{\vs}-\frac{\vWeight}{\next{\vs}}\right)}_{\left(\ma^{T}\ms^{-1}\mWeight\ms^{-1}\ma\right)^{-1}}\\
 & = & \left(1+\gamma\eta\right)\delta_{t}\left(\vx,\vWeight\right)+\normFull{\frac{\next{\vs}-\vs}{\next{\vs}}}_{\mWeight}\\
 & \leq & \left(1+\gamma\eta\right)\delta_{t}\left(\vx,\vWeight\right)+(1-\gamma\eta)^{-1}\normFull{\frac{\next{\vs}-\vs}{\vs}}_{\mWeight}.
\end{eqnarray*}
By the same proof in Lemma \ref{lem:weighted_path:stab_update_step},
we have that
\[
\norm{\ms^{-1}(\vSlackNext-\vs)}_{\mWeight}\leq\eta.
\]
Thus, we have the result.\end{proof}

Therefore, as long as we choose $\epsilon$ small enough, the algorithm
$\code{solve}$ gives an accurate enough $\next{\vx}$ for the centering
step. Similarly, it is easy to see that it also gives accurate enough
$\next{\vWeight}$ because the error of $\next{\vWeight}$ due to
$\code{solve}$ is small in $\norm{\cdot}_{\mWeight}$ norm and the
tracing 0 game can afford for this error.

At last, we need to check $\code{solve}$ gives us a way to compute
weight function. Since the weight function computation relies on the
function $\code{computeLeverageScores}$, we only need to know if
we can compute $\vl$ in the $\code{computeLeverageScores}$ with
high enough accuracy. Now, we use the notation is the $\code{computeLeverageScores}$.
Without loss of generality, we can assume $\mx=\iMatrix$. Let $\vl^{(apx)}$
and $\vp^{(apx)}$ be the approximate $\vl$ and $\vp$ computed by
the algorithm $\code{Solve}$. Then, we have
\begin{eqnarray*}
\normFull{\left(\vl^{(j)}\right)^{(apx)}-(\ma^{T}\ma)^{+}\ma^{T}\vec{q}^{(j)}}_{\ma^{T}\ma} & = & \normFull{\left(\vl^{(j)}\right)^{(apx)}-\ma^{+}\vec{q}^{(j)}}_{\ma^{T}\ma}\\
 & \leq & \epsilon\norm{\ma^{+}\vec{q}^{(j)}}_{\ma^{T}\ma}\\
 & = & \epsilon\norm{\ma^{T}\vec{q}^{(j)}}_{\left(\ma^{T}\ma\right)^{-1}}\\
 & \leq & \epsilon\norm{\vec{q}^{(j)}}_{2}\leq\epsilon\sqrt{\frac{n}{k}}.
\end{eqnarray*}
Hence, for any $i,j$, we have
\begin{eqnarray*}
\normFull{\vp_{i}^{(j)}-\left(\vp_{i}^{(apx)}\right)^{(j)}}_{\infty} & \leq & \normFull{\vp^{(j)}-\left(\vp^{(apx)}\right)^{(j)}}_{2}\\
 & = & \normFull{\ma\left(\left(\vl^{(j)}\right)^{(apx)}-\vl^{(j)}\right)}_{2}\\
 & \leq & \epsilon\sqrt{\frac{n}{k}}.
\end{eqnarray*}
Therefore, we have
\begin{eqnarray*}
\sqrt{\sum_{j=1}^{k}\left(\vp_{i}^{(j)}\right)^{2}}-\sqrt{\sum_{j=1}^{k}\left(\left(\vp_{i}^{(apx)}\right)^{(j)}\right)^{2}} & \leq & \sqrt{\sum_{j=1}^{k}\left(\vp_{i}^{(j)}-\left(\vp_{i}^{(apx)}\right)^{(j)}\right)^{2}}\\
 & \leq & \epsilon\sqrt{nk}.
\end{eqnarray*}
Therefore, if $\epsilon\leq\sqrt{\frac{1}{m\polylog(m)}}$, the error
is small enough for $\code{computeLeverageScores}$.

\section{Bit Complexity and Linear Program Reductions\label{sec:app:bit_complexity}}

In this section, we show how to reduce solving an arbitrary linear
program to finding a low cost solution in a bounded linear program
for which we have an explicit interior point. Throughout this section
let$\ma\in\R^{m\times n}$, $\vb\in\Rm$, $\vc\in\Rn$, and consider
the following general linear program
\begin{equation}
\min_{\vx\in\Rn~:~\ma\vx\geq\vb}\vc^{T}\vx\label{eq:bit_complexity_LP1}
\end{equation}
We assume that the entries of $\ma$, $\vb$, and $\vc$ are integers
and we let $OPT$ denote the optimal value of \eqref{eq:bit_complexity_LP1}
and we let $L$ denote the bit complexity of \eqref{eq:bit_complexity_LP1}
where
\[
L\defeq\log(m)+\log(1+d_{max}(\ma))+\log(1+\max\{\normInf{\vc},\normInf{\vb}\})
\]
 and $d_{max}(\ma)$ denotes the largest absolute value of the determinant
of a square sub-matrix of $\ma$. Our goal is to efficiently transform
\eqref{eq:bit_complexity_LP1} to a linear program of the same form
\begin{equation}
\min_{\vx\in\R^{n'}~:~\ma'\vx\geq\vb'}\vc'^{T}\vx\label{eq:app:problem_mod}
\end{equation}
where $\ma'\in\R^{m'\times n'},$ $\vb'\in\R^{m'}$, and $\vc'\in\R^{n'}$
are integer, and $\nnz(\ma'$), $n'$, $m'$, and the bit complexity
of \eqref{eq:app:problem_mod} denoted, $L'$, are comparable to $\nnz(\ma)$,
$n$, $m$, and $L$. Furthermore, we require that \eqref{eq:app:problem_mod}
is bounded, has an explicit efficiently computable interior point,
and that we can convert any low cost feasible solution to a solution
of \eqref{eq:bit_complexity_LP1} in linear time.

While there are standard tools to perform reductions to ensure that
\eqref{eq:bit_complexity_LP1} is bounded and has an explicit initial
feasible point or to ensure that the optimal integral solution can
be easily computed explicitly, we need to particularly careful when
using these reductions to ensure that $\nnz(\ma)$, $n$, and $m$
are not increased significantly. As the running times of our path
following techniques in Section~\eqref{sec:algorithm} depend crucially
on these parameters in this section we prove the following Lemma claiming
that such an efficient reduction is possible.

\begin{lemma}\label{lem:the_modified_LP} In $O(\nnz(\ma)+n+m)$
time we can compute integer $\ma'\in\R^{m'\times n'},$ $\vb'\in\R^{m'}$,
$\vc'\in\R^{n'}$, $\vx'\in\R^{m'}$. Such that $\nnz(\ma')=O(\nnz(\ma)+n+m)$,
$n'=O(n)$, $m'=O(m)$, $\ma'\vx'\geq\vb'$, and \eqref{eq:app:problem_mod}
is bounded and has bit complexity at most $12L_{1}+7\log(20n)$. Furthermore,
if we can find a feasible point in \eqref{eq:app:problem_mod} such
that the cost of that point is at most the $OPT+2^{-12\left(L+\log(20n)\right)}$
where $OPT$ is the value of \eqref{eq:app:problem_mod} then we can
either 
\begin{enumerate}
\item Find the active constraints of a basic feasible optimal solution \eqref{eq:bit_complexity_LP1}
using only one matrix vector multiplication by $\ma$; or
\item Prove that \eqref{eq:bit_complexity_LP1} is infeasible or unbounded.
\end{enumerate}
\end{lemma} 

We break this proof into two parts. First in Lemma~\ref{lem:app:feasible_and_bounded}
we show how to transform \eqref{eq:bit_complexity_LP1} so that the
linear program is bounded and has an explicit feasible point. Then
in Lemma~\ref{lem:app:sol} we follow the approach of \cite{daitch2008faster}
and show that we can perturb the cost of a linear program to make
the optimal solution unique and thereby make it easy to compute an
exact integral solution. 

\begin{lemma}\label{lem:app:feasible_and_bounded} Consider the following
modified linear program
\begin{equation}
\min\vc^{T}\vx+n2^{3L+4}z\text{ such that }\ma\vx+z\onesVec\geq\vb,2^{L+1}\geq z\geq0,2^{L+1}\onesVec\geq\vx\geq-2^{L+1}\onesVec\label{eq:bit_complexity_LP2}
\end{equation}
where $\ma$, $\vb$, and $\vc$ are as in \eqref{eq:bit_complexity_LP1}
and $L$ is the bit complexity of \eqref{eq:bit_complexity_LP1}.
\eqref{eq:bit_complexity_LP2} is bounded with an explicit interior
point $\vx=0,z=2^{L}+1$. Furthermore, \eqref{eq:bit_complexity_LP1}
is bounded and feasible with an optimal solution $\vx$ if and only
if $(\vx,0)$ is an optimal solution of \eqref{eq:bit_complexity_LP2}
with $2^{L}\geq x_{i}\geq-2^{L}$, \eqref{eq:bit_complexity_LP1}
is unbounded if and only if there is a basic feasible solution, $(\vx,z)$,
of \eqref{eq:bit_complexity_LP2} with $|x_{i}|>2^{L}$ for some $i$,
and \eqref{eq:bit_complexity_LP1} is infeasible if and only if there
is basic feasible solution, $(\vx,z)$, of \eqref{eq:bit_complexity_LP2}
with $\vz\neq0$. Furthermore, \eqref{eq:bit_complexity_LP2} can
be written in the form \eqref{eq:app:problem_mod} such that all these
properties hold with $\nnz(\ma')=O(\nnz(\ma)+n+m)$, $n'=O(n)$, $m'=O(m)$,
and $L'\leq$ $4L+2\log(16n)$.

\end{lemma} 

\begin{proof}

Case 1: Suppose \eqref{eq:bit_complexity_LP1} is bounded and feasible.
It is known that any basic feasible solution of \eqref{eq:bit_complexity_LP1}
is a vector of rational numbers with both absolute value of numerator
and denominator are bounded by $2^{L}$ \cite{papadimitriou1998combinatorial}.
Therefore, $-n2^{2L}\leq OPT\leq n2^{2L}$. Given any feasible solution
$\vx$ of \eqref{eq:bit_complexity_LP1}, the point $\left(\vx,z=0\right)$
is a feasible solution of \eqref{eq:bit_complexity_LP2} with same
cost value. Hence, the linear program \eqref{eq:bit_complexity_LP2}
is feasible and the optimal value of \eqref{eq:bit_complexity_LP2}
is at most $n2^{2L}$.

On the other hand, clearly \eqref{eq:bit_complexity_LP2} is feasible
because $\vx=\vzero,z=2^{L}+1$ is an interior point. Furthermore,
\eqref{eq:bit_complexity_LP2} is bounded and therefore has some optimal
value. Consider any optimal basic feasible solution $\left(\vx,z\right)$
of \eqref{eq:bit_complexity_LP2}, we have $\vc^{T}\vx$ is between
$-n2^{2L+1}$ and $n2^{2L+1}$. Also, $z$ is a rational number with
the absolute value of denominator are bounded by $2^{L}$ using Cramer's
rule. Therefore, we have $z\geq2^{-L_{1}}$ or $z=0$. If $z\geq2^{-L}$,
then the total cost is at least $n2^{3L+4}2^{-L}-n2^{2L+1}>n2^{2L}$.
However, as we argued above, the optimal value of \eqref{eq:bit_complexity_LP2}
is at most $n2^{2L}$. Therefore, optimal solution has $z=0$ and
$2^{L}\geq x_{i}\geq-2^{L}$ for all $i$. 

Case 2: Suppose \eqref{eq:bit_complexity_LP1} is not feasible. In
this case, any feasible point $\left(\vx,z\right)$ in \eqref{eq:bit_complexity_LP2}
has $z\neq0$ and by the reasoning in the previous section any basic
feasible solution has cost greater than $n2^{2L}$.

Case 3: Suppose \eqref{eq:bit_complexity_LP1} is not bounded. Let
$\text{OPT}_{k}=\min\vc^{T}\vx\text{ such that }\ma\vx\geq\vb,k+2^{L}\geq x_{i}\geq-2^{L}-k$.
Thus, we have $\text{OPT}_{1}<\text{OPT}_{0}$ and any optimal point
of the case $k=1$ has some coordinate larger than $2^{L}$ or smaller
$-2^{L}$. By similar argument as above, we have that the optimal
point of \eqref{eq:bit_complexity_LP2} is of the form $(\vx,0)$
and some coordinate of $\vx$ is larger than $2^{L}$ or smaller $-2^{L}$.

To compute the bit complexity of \eqref{eq:bit_complexity_LP2} note
that we can write \eqref{eq:bit_complexity_LP2} in the form of \eqref{eq:app:problem_mod}
by choosing
\begin{equation}
\ma'=\left[\begin{array}{cc}
\ma & \onesVec\\
\iMatrix & \vzero\\
-\iMatrix & \vzero\\
\vzero^{T} & 1\\
\vzero^{T} & -1
\end{array}\right]\,,\,\vb'=\left(\begin{array}{c}
\vb\\
-2^{L_{1}+1}\\
2^{L_{1}+1}\\
0\\
2^{L_{1}+1}
\end{array}\right)\,,\,\vc'=\left(\begin{array}{c}
\vc\\
n2^{3L+4}
\end{array}\right)\,\text{ where }\,\mbox{\ensuremath{\iMatrix\in\Rmm}\,\text{and}\,\ensuremath{\vzero\in\Rm}}\label{eq:app_lemmas:big_constraint}
\end{equation}
Thus $n'=n+1$, $m'=3m+2$, and it is easy to see that
\begin{eqnarray*}
d_{max}(\ma') & = & d_{max}\left(\left[\begin{array}{cc}
\ma & \onesVec\end{array}\right]\right)\leq n\cdot d_{max}\left(\ma\right).
\end{eqnarray*}
Therefore, the bit complexity of \eqref{eq:bit_complexity_LP2} is
at most $\log(1+nd_{max}(\ma))+\log(1+n2^{3L+4})\leq4L+2\log(16n)$
as desired.

\end{proof}

Following the approach in \cite{daitch2008faster} to use the following
isolation lemma, we show that it is possible to transform the linear
program into one with unique optimal solution by randomly perturbing
the cost function.

\begin{lemma}[\cite{klivans2001randomness}]\label{lem:isolation_lemma}
Given any collection of linear functions on $n$ variables $c_{1},c_{2},\cdots,c_{n}$
with coefficients in the range $\{-K,-K-1,\cdots,K-1,K\}$. If $c_{1},\cdots,c_{n}$
are independently chosen uniformly at random in $\{-2Kn,\cdots,2Kn\}$.
Then, with probability greater than $\frac{1}{2}$, there is a unique
linear function of minimum value at $c_{1},c_{2},\cdots,c_{n}$.

\end{lemma}

Note that for we can think every vertex $\vx$ is a linear function
$\vc^{T}\vx$ on the cost variables $\vc$. Although there are exponentially
many vertices, the above lemma shows that the minimizer is attained
at a unique vertex (linear function).

\begin{lemma}\label{lem:app:sol} Suppose that \eqref{eq:bit_complexity_LP1}
is feasible and bounded and consider the following modified linear
program
\begin{equation}
\min\left(2^{2L+3}n\vc+\vr\right)^{T}\vx\text{ given }\ma\vx\geq\vb.\label{eq:bit_complexity_LP4}
\end{equation}
where each coordinate in $\vr\in\Rm$ is chosen uniformly at random
from the integers $\{-2^{L+1}n,\cdots,2^{L+1}n\}$. 

Let $OPT$' denote the optimal value of the linear program \eqref{eq:bit_complexity_LP4}.
Given any feasible solution for the linear program \eqref{eq:bit_complexity_LP4}
with cost less than $OPT+n^{-1}2^{-3L-2}$, we can find the active
constraints of a basic feasible optimal solution of \eqref{eq:bit_complexity_LP1}
by using only one matrix vector multiplication with $\ma$. Furthermore,
the bit complexity of \eqref{eq:bit_complexity_LP4} is at most $3L+\log(8n)$.

\end{lemma} 

\begin{proof}

Since the set of basic solutions to \eqref{eq:bit_complexity_LP4}
and \eqref{eq:bit_complexity_LP1} are the same, we know that any
basic feasible solution of \eqref{eq:bit_complexity_LP4} is a vector
of rational numbers with absolute value of numerator and denominator
both bounded by $2^{L}$. Consequently our perturbation of the cost
function maintains that an optimum solution to \eqref{eq:bit_complexity_LP4}
is an optimal solution to \eqref{eq:bit_complexity_LP1}. Hence, the
Isolation Lemma shows that with probability greater than $\frac{1}{2}$,
the linear program \eqref{eq:bit_complexity_LP4} has a unique solution
$\vx^{*}$.

Now consider the polytope $P_{t}=\{\vx\text{ such that }\ma\vx\geq\vb\text{ and }\left(2^{2L+3}n\vc+\vr\right)^{T}\vx\leq OPT+t2^{-2L-1}\}$
for $t>0$. Since \eqref{eq:bit_complexity_LP4} has a unique solution,
by a similar argument as before, $P_{1}$ contains only one basic
feasible solution of \eqref{eq:bit_complexity_LP4} and hence $P_{t}-\vx^{*}=t\left(P_{1}-\vx^{*}\right)$
for any $t\leq1$. Also, for any $\vx\in P_{1}$, $\vx$ is in the
polytope of $\{\ma\vx\geq\vb\}$ and hence $\norm{\vx}_{\infty}\leq2^{L}.$
Therefore, for any $\vx\in P_{t}$, we have $\norm{\vx-\vx^{*}}_{\infty}\leq t\cdot2^{L+1}$
for any $t\leq1$. Therefore, for any $\vx\in P_{t}$, $\norm{\ma\vx-\ma\vx^{*}}_{\infty}\leq nt2^{2L+1}.$
Since $\ma\vx^{*}$ is a vector of rational numbers with the absolute
value of denominator are bounded by $2^{L}$, we can distinguish if
a constraint is satisfied or not when $nt2^{2L+1}<2^{-L-1}.$

\end{proof}

Combining Lemma~\ref{lem:app:feasible_and_bounded} and Lemma~\ref{lem:app:sol}
proves Lemma~\ref{lem:the_modified_LP}.

\section{Numerical Linear Algebra for Acceleration}

\label{sec:app:acceleration}

Here we prove Theorem~\ref{thm:alg:accel} needed for the accelerated
linear program solver. Below we restate the theorem for convenience.

\begin{theorem} Let $\vd^{(i)}\in\dSlack$ be a sequence of $r$
positive vectors. Suppose that the number of times that $d_{j}^{(i)}\neq d_{j}^{(i+1)}$
for any $i\in[r]$ and $j\in[m]$ is bounded by $Cr^{2}$ for some
$C\geq1$. Then if we are given the $\vd^{(i)}$ in a sequence, in
each iteration $i$ we can compute $\left(\ma^{T}\md_{i}\ma\right)^{-1}\vx_{i}$
for $\md_{i}=\mDiag(\vd_{i})$ and arbitrary $\vx_{i}\in\Rn$ with
the average cost per iteration
\[
\otilde\left(\frac{mn^{\omega-1}}{r}+n^{2}+C^{\omega}r^{2\omega}+C^{\omega-1}nr^{\omega}\right)
\]
where $\omega<2.3729$ \cite{williams2012matrixmult} is the matrix
multiplication constant.

\end{theorem}

\begin{proof} For all $i\in[r]$ let $\mb_{i}=\ma^{T}\md_{i}\ma$.
Since $\md_{1}\in\R^{m\times m}$ is diagonal and $\ma\in\Rnm$ we
can compute $\md_{1}\ma$ trivially in $O(mn)$ time. Furthermore
from this we can compute $\mb_{1}=\ma^{T}\md_{1}\ma$ in $O(mn^{\omega-1})$
time using fast matrix multiplication by splitting $\ma$ into $\frac{m}{n}$
blocks of size $n$ and using that $m>n$. Furthermore, using fast
matrix multiplication we can then compute $\mb_{1}^{-1}$ in $O(n^{\omega})$
time and similarly we can compute $\mb_{1}^{-1}\ma^{T}$ in $O(mn^{\omega-1})$
time. Now, we show how to use this computation of $\mb_{1}^{-1}$
and $\mb_{1}^{-1}\ma^{T}$ in $O(mn^{\omega-1})$ time to decrease
the running time of future iterations.

For all $k>1$, let $\md_{k}=\md_{1}+\mDelta_{k}$ for some diagonal
$\mDelta_{k}\in\Rmm$ and let $r_{k}\defeq\nnz(\mDelta_{k})$. Let
$\mProj_{k}\in\R^{r_{k}\times n}$ be the $1-0$ matrix that selects
the rows of $\ma$ for which the diagonal entry in $\mDelta_{k}$
is nonzero, let $\ms_{k}\in\R^{r_{k}\times r_{k}}$ be the diagonal
matrix whose diagonal entries are the non-zero diagonal entries of
$\mDelta_{k}$ and $\ma_{k}\defeq\mProj_{k}\ma$.

Note that $\mDelta_{k}=\mProj_{k}^{T}\ms_{k}\mProj_{k}$ and hence
by the Woodbury matrix identity, we have
\begin{align}
\mb_{i}^{-1} & =\left(\ma^{T}\md_{1}\ma+\ma^{T}\mProj_{k}^{T}\ms_{k}\mProj_{k}\ma\right)^{-1}\nonumber \\
 & =\mb_{1}^{-1}-\mb_{1}^{-1}\ma_{k}^{T}\left(\ms_{k}^{-1}+\ma_{k}\mb_{1}^{-1}\ma_{k}^{T}\right)^{-1}\ma_{k}\mb_{1}^{-1}\label{eq:accel}
\end{align}
Assume we have computed $\ma_{k}\mb_{k}^{-1}\ma_{k}^{T}\in\R^{r_{k}\times r_{k}}$
explicitly, we can use fast matrix multiplication to compute $\left(\ms_{k}^{-1}+\ma_{k}\mb_{k}^{-1}\ma_{k}^{T}\right)^{-1}$
in time $O(r_{k}^{\omega})$. Then, we can use \eqref{eq:accel} to
compute $\mb_{i}^{-1}\vx_{i}$ in just
\[
O\left(\nnz\left(\mb_{1}^{-1}\right)+\nnz(\ma_{k})+\nnz\left(\left(\ms_{k}^{-1}+\ma_{k}\mb_{1}^{-1}\ma_{k}^{T}\right)^{-1}\right)\right)=O(nr_{k}+n^{2})
\]
time. Consequently, not counting the time to compute $\ma_{k}\mb_{k}^{-1}\ma_{k}^{T}\in\R^{r_{k}\times r_{k}}$,
we have that the average cost of computing $\mb_{i}^{-1}\vx_{i}$
is
\begin{equation}
\otilde\left(\frac{mn^{\omega-1}}{r}+n^{2}+nr_{k}+r_{k}^{\omega}\right)=\otilde\left(\frac{mn^{\omega-1}}{r}+n^{2}+C^{\omega}r^{2\omega}\right)\label{eq:accel:form}
\end{equation}
because $r_{k}\leq Cr^{2}$ and $nr_{k}\leq2n^{2}+2r_{k}^{2}.$

All that remains is to estimate the cost of computing $\ma_{k}\mb_{k}^{-1}\ma_{k}^{T}$.
For notational simplicity, we order the rows of $\ma$ such that $\ma_{k}^{T}=[\ma_{k-1}^{T}\,\mr_{k}^{T}]$
where $\mr_{k}\in\R^{u_{k}\times n}$ where $u_{k}=r_{k}-r_{k-1}$.
From this, to compute $\ma_{k}\mb_{k}^{-1}\ma_{k}^{T}$ we see that
it suffices to compute
\[
\left(\begin{array}{cc}
\ma_{k}\mb_{1}^{-1}\ma_{k}^{T} & \ma_{k}\mb_{1}^{-1}\mU_{k}^{T}\\
\mU_{k}\mb_{1}^{-1}\ma_{k} & \mU_{k}\mb_{1}^{-1}\mU_{k}^{T}
\end{array}\right)
\]
Now, since we precomputed $\mb_{1}^{-1}\ma^{T}$ and $\mU$ is just
a subset of the rows of $\ma$, we see that we can compute $\mb_{1}^{-1}\mU_{k}^{T}$
by extracting columns from $\mb_{1}^{-1}\ma^{T}$. Thus, we see that
the time to compute $\ma_{k}^{T}\mb_{k}^{-1}\ma_{k}$ is dominated
by the time to multiply a matrix of size at most $r_{k}\times n$
and $n\times u_{k}$. We can do this by multiplying $O\left(\frac{r_{k}}{u_{k}}\cdot\frac{n}{u_{k}}\right)$
matrices of size $u_{k}\times u_{k}$ which can be done in $O(r_{k}nu_{k}^{\omega-2})$
time. Thus the average cost of computing $\ma_{k}^{T}\mb_{k}^{-1}\ma_{k}$
is
\[
O\left(\sum_{1\leq k<r}\left(\frac{1}{r}\right)\cdot\left(r_{k}nu_{k}^{\omega-2}\right)\right)\leq O(Crn\cdot r\cdot\left(Cr\right)^{\omega-2})=O(C^{\omega-1}nr^{\omega})
\]
where we used the fact that since $\sum_{k}u_{k}=r_{k}$, $r_{k}\leq Cr^{2}$
and the minimum value of $\sum_{k}u_{k}^{\omega-2}$ is achieve when
each $u_{k}=Cr$.

\end{proof}

\end{document}